\documentclass[10pt,a4paper]{article}

\usepackage[english]{babel}
\usepackage{amsmath}
\usepackage{amssymb}
\usepackage{stmaryrd}
\usepackage{ifthen}
\usepackage{MnSymbol}

\newcommand{\macrospath}{.}

\usepackage{tikz}
\usetikzlibrary{matrix,arrows}
\usetikzlibrary{decorations.shapes}
\usetikzlibrary{decorations.text}
\usetikzlibrary{decorations.pathmorphing}
\usetikzlibrary{decorations.markings}

\tikzset{
node distance=1.3cm, auto,
every node/.style={font=\tiny },
ocenter/.style={baseline={([yshift=-.5ex, xshift=-.5ex]current bounding box)}},  
labelBeginAbove/.style={postaction={decorate,decoration={markings,mark=at position 0 with {\node[inner sep= 0.6pt, above=1pt]{\tiny #1};}} } },
labelBeginBelow/.style={postaction={decorate,decoration={markings,mark=at position 0 with {\node[inner sep= 0.6pt, below=1pt]{\tiny #1};}}}},
labelEndAbove/.style={postaction={decorate,decoration={markings,mark=at position 1 with {\node[inner sep= 0.6pt, above=1pt]{\tiny #1};}}}},
labelEndBelow/.style={postaction={decorate,decoration={markings,mark=at position 1 with {\node[inner sep= 0.6pt, below=1pt]{\tiny #1};}}}},
labelEndRight/.style={postaction={decorate,decoration={markings,mark=at position 1 with {\node[inner sep= 0.6pt, right=1pt]{\tiny #1};}}}},
labelEndLeft/.style={postaction={decorate,decoration={markings,mark=at position 1 with {\node[inner sep= 0.6pt, left=1pt]{\tiny #1};}}}}
}

\newcommand{\nodeHorDist}{2cm}
\newcommand{\nodeVerDist}{1cm}

\newcommand{\commutesredgenEEABCD}[8]{
      \begin{tikzpicture}[ocenter]
       \node (s) {\normalsize\ensuremath{#3}};
       \node at (s.center)  [right=1.7*\nodeHorDist](s1){\normalsize\ensuremath{#4}};
       \node at (s.center)  [below=\nodeVerDist](s2) {\normalsize\ensuremath{#5}};
       \node at (s2.center) [right=1.7*\nodeHorDist](t) {\normalsize\ensuremath{#6}};
       \node at (s.center)  [below=\nodeVerDist/2](eq1){\normalsize\ensuremath{#1}};
       \node at (s1.center) [below=\nodeVerDist/2](eq2){\normalsize\ensuremath{#2}};
       \draw[-o] (s) to node {#7} (s1);
       \draw[-o, dashed] (s2) to node {#8} (t);
      \end{tikzpicture} 
}

\newcommand{\commutesredEEABCD}[6]{
\commutesredgenEEABCD{#1}{#2}{#3}{#4}{#5}{#6}{}{}
}

\newcommand{\commutesdbEEABCD}[6]{
\commutesredgenEEABCD{#1}{#2}{#3}{#4}{#5}{#6}{\db}{\db}
}

\newcommand{\commutesdbvEEABCD}[6]{
\commutesredgenEEABCD{#1}{#2}{#3}{#4}{#5}{#6}{\dbv}{\dbv}
}

\newcommand{\commuteslsvEEABCD}[6]{
\commutesredgenEEABCD{#1}{#2}{#3}{#4}{#5}{#6}{\lsvsym}{\lsvsym}
}

\newcommand{\ignore}[1]{}
\newcommand{\sep}{\hspace*{0.5cm}}

\newcommand{\myinput}[1]{\ifthenelse{\boolean{withimages}}{\input{#1}}{}}

\newcommand{\reflemma}[1]{Lemma~\ref{l:#1}}
\newcommand{\refth}[1]{Theorem~\ref{tm:#1}}
\newcommand{\refprop}[1]{Proposition~\ref{prop:#1}}
\newcommand{\refsect}[1]{Sect.~\ref{sect:#1}}
\newcommand{\reftab}[1]{Tab.~\ref{tab:#1}}
\newcommand{\refeq}[1]{(\ref{eq:#1})}
\newcommand{\reffig}[1]{Fig.~\ref{fig:#1}}

\newcommand{\case}[1]{{\bf Case #1.}}
\newcommand{\casealt}[1]{{\bf #1.}}
\newcommand{\caselight}[1]{\textit{#1}}

\newcommand{\ie}{\textit{i.e.}}
\newcommand{\eg}{\textit{e.g.}}
\newcommand{\ih}{\textit{i.h.}}
\newcommand{\lat}{\mbox{$\lambda$-term}}

\newcommand{\linlogic}{linear logic}

\newcommand{\pns}{proof nets}

\newcommand{\wlhr}{WLHR}

\newcommand{\deff}[1]{\textbf{#1}}

\newcommand{\red}[1]{{\color{red} {#1}}}

\newcommand{\black}[1]{{\color{black} {#1}}}

\newcommand{\ben}[1]{{\red{#1}}}
\renewcommand{\ben}[1]{{#1}}

\newcommand{\cben}[2]{{\red{#2}}}
\renewcommand{\cben}[2]{{#2}}

\newcommand{\defeq}{\mathrel{:=}}
\newcommand{\grameq}{\mathrel{::=}}
\newcommand{\set}[1]{\{#1\}}

\newcommand{\card}[1]{\# #1}

\newcommand{\size}[1]{|#1|}

\newcommand{\LeftRightarrow}{\Lleftarrow\!\!\!\!\Rrightarrow}

\newcommand{\db}{{\tt dB}}
\newcommand{\dbv}{{\tt dBv}}
\newcommand{\lssym}{{\tt ls}}
\newcommand{\ls}{{\tt ls}}
\newcommand{\vartt}{{\tt var}}
\newcommand{\gc}{{\tt gc}}
\newcommand{\Bsym}{{\tt B}}
\newcommand{\lsvsym}{{\tt lsv}}
\newcommand{\vsym}{{\tt v}}

\newcommand{\admsym}{c}
\newcommand{\noadmsym}{p}
\newcommand{\mulsym}{m}
\newcommand{\expsym}{e}

\newcommand{\lsc}{LSC}
\newcommand{\wlsc}{WLSC}
\newcommand{\wlscname}{{\tt Name}}
\newcommand{\wlscvaluelr}{{\tt Value}^{\tt LR}}
\newcommand{\wlscvaluerl}{{\tt Value}^{\tt RL}}
\newcommand{\wlscneed}{{\tt Need}}

\renewcommand{\l}{\lambda}
\newcommand{\isub}[2]{\{#1/#2\}}
\renewcommand{\isub}[2]{\{#1{\shortleftarrow}#2\}}
\newcommand{\esub}[2]{[#1/#2]}
\renewcommand{\esub}[2]{[#1{\shortleftarrow}#2]}
\newcommand{\fv}[1]{{\tt fv}(#1)}
\newcommand{\varsplit}[3]{#1_{[#3]_{#2}}}

\newcommand{\rootRew}[1]{\mapsto_{#1}}
\newcommand{\Rew}[1]{\rightarrow_{#1}}

\newcommand{\rtodb}{\rootRew{\db}}
\newcommand{\rtols}{\rootRew{\lssym}}
\newcommand{\rtolsc}[1]{\stackrel{#1}{\mapsto}_\lssym}

\newcommand{\rtodbv}{\rootRew{\db\vsym}}
\newcommand{\rtolsv}{\rootRew{\lssym\vsym}}

\newcommand{\rtolsvc}[1]{\stackrel{#1}{\mapsto}_{\lssym\vsym}}

\newcommand{\rtowhls}{\rtols} % this one is redundant but used in a proof
\newcommand{\towhl}{\stackrel{\mathtt{wh}}{\multimap}}
\newcommand{\towhldb}{\stackrel{\mathtt{wh}}{\multimap}_\db}
\newcommand{\towhlls}{\stackrel{\mathtt{wh}}{\multimap}_\lssym}
\renewcommand{\towhl}{\togen}
\renewcommand{\towhldb}{\togenm}
\renewcommand{\towhlls}{\togene}

\newcommand{\towhlcek}{\stackrel{\mathtt{ns}}{\multimap}_\vsym}
\newcommand{\towhlcekdb}{\stackrel{\mathtt{ns}}{\multimap}_\dbv}
\newcommand{\towhlcekls}{\stackrel{\mathtt{ns}}{\multimap}_\lsvsym}
\renewcommand{\towhlcek}{\togen}
\renewcommand{\towhlcekdb}{\togenm}
\renewcommand{\towhlcekls}{\togene}

\newcommand{\towhllam}{\stackrel{\mathtt{s}}{\multimap}_\vsym}
\newcommand{\towhllamdb}{\stackrel{\mathtt{s}}{\multimap}_\dbv}
\newcommand{\towhllamls}{\stackrel{\mathtt{s}}{\multimap}_\lsvsym}
\renewcommand{\towhllam}{\togen}
\renewcommand{\towhllamdb}{\togenm}
\renewcommand{\towhllamls}{\togene}

\newcommand{\tostructsym}{\LeftRightarrow}

\newcommand{\esym}{{\mathtt e}}
\newcommand{\msym}{{\mathtt m}}
\newcommand{\rtogenm}{\mapsto_\msym}
\newcommand{\rtogene}{\mapsto_\esym}
\newcommand{\togen}{\multimap}
\newcommand{\togenm}{\multimap_\msym}
\newcommand{\togene}{\multimap_\esym}

\newcommand{\togenx}{\multimap_{\mathtt x}}

\newcommand{\tom}{\Rew{\msym}}
\newcommand{\toe}{\Rew{\esym}}

\newcommand{\alphaequiv}{=_\alpha}
\newcommand{\eqstruct}{\equiv}
\newcommand{\eqstructneed}{\equiv_{\tt Need}}
\newcommand{\tostruct}{\eqstruct}
\newcommand{\tostructneed}{\eqstructneed}
\newcommand{\tostructgc}{\tostruct_{gc}}

\newcommand{\tostructdup}{\tostruct_{dup}}
\newcommand{\tostructap}{\tostruct_{@}}
\newcommand{\tostructapl}{\tostruct_{@l}}

\newcommand{\tostructes}{\tostruct_{[\cdot]}}
\newcommand{\tostructcom}{\tostruct_{com}}

\newcommand{\eqmam}{\equiv_{\tiny \mbox{MAM}}}
\newcommand{\eqmamsym}{\LeftRightarrow}

\newcommand{\tm}{t}
\newcommand{\tmtwo}{u}
\newcommand{\tmthree}{w}
\newcommand{\tmfour}{r}
\newcommand{\tmfive}{q}
\newcommand{\tmsix}{p}

\newcommand{\tmp}{\tm'}
\newcommand{\tmtwop}{\tmtwo'}
\newcommand{\tmthreep}{\tmthree'}
\newcommand{\tmfourp}{\tmfour'}
\newcommand{\tmfivep}{\tmfive'}

\newcommand{\tmpp}{\tm''}
\newcommand{\tmtwopp}{\tmtwo''}
\newcommand{\tmthreepp}{\tmthree''}

\newcommand{\var}{x}
\newcommand{\vartwo}{y}
\newcommand{\varthree}{z}

\newcommand{\val}{v}
\newcommand{\valtwo}{v'}

\newcommand{\ctxholep}[1]{\langle #1\rangle}
\newcommand{\ctxhole}{\ctxholep{\cdot}}

\newcommand{\ctx}{C}

\newcommand{\ctxp}[1]{\ctx\ctxholep{#1}}

\newcommand{\sctx}{L}
\newcommand{\sctxtwo}{\sctx'}
\newcommand{\sctxthree}{\sctx''}
\newcommand{\sctxp}[1]{\sctx\ctxholep{#1}}
\newcommand{\sctxtwop}[1]{\sctxtwo\ctxholep{#1}}
\newcommand{\sctxthreep}[1]{\sctxthree\ctxholep{#1}}
\newcommand{\sctxOne}{\sctx_1}
\newcommand{\sctxTwo}{\sctx_2}
\newcommand{\sctxOnep}[1]{\sctxOne\ctxholep{#1}}
\newcommand{\sctxTwop}[1]{\sctxTwo\ctxholep{#1}}
\newcommand{\wctx}{W}
\newcommand{\wctxtwo}{\wctx'}

\newcommand{\wctxp}[1]{\wctx\ctxholep{#1}}

\newcommand{\arbctxp}[1]{\arbctxp{#1}}
\newcommand{\arbctxtwop}[1]{\arbctxtwop{#1}}

\newcommand{\evctx}{\genevctx}
\newcommand{\evctxtwo}{\genevctxtwo}
\newcommand{\evctxp}[1]{\genevctxp{#1}}
\newcommand{\evctxtwop}[1]{\genevctxtwop{#1}}
\newcommand{\sctxvsTwo}[2]{\varsplit{\sctxTwo\,\!}{#1}{#2}}

\newcommand{\sctxvsTwop}[3]{\sctxvsTwo{#1}{#2}\ctxholep{#3}}

\newcommand{\sctxal}{\widehat{\sctx}}

\newcommand{\sctxpal}[1]{\sctxal\ctxholep{#1}}

\newcommand{\evctxthree}{\whctx''}
\newcommand{\evctxfour}{\whctx'''}
\newcommand{\evctxthreep}[1]{\evctxthree\ctxholep{#1}}

\newcommand{\evctxpthree}[1]{\evctxthree\ctxholep{#1}}

\newcommand{\evctxal}{\widehat{\evctx}}
\newcommand{\evctxaltwo}{\widehat{\evctxtwo}}

\newcommand{\evctxpal}[1]{\evctxal\ctxholep{#1}}
\newcommand{\evctxpaltwo}[1]{\evctxaltwo\ctxholep{#1}}

\newcommand{\whctx}{H}
\newcommand{\whctxtwo}{\whctx'}
\newcommand{\whctxp}[1]{\whctx\ctxholep{#1}}
\newcommand{\whctxtwop}[1]{\whctxtwo\ctxholep{#1}}

\newcommand{\cbvctx}{V}
\newcommand{\cbvctxtwo}{\cbvctx'}
\newcommand{\cbvctxthree}{\cbvctx''}

\newcommand{\cbvctxp}[1]{\cbvctx\ctxholep{#1}}
\newcommand{\cbvctxtwop}[1]{\cbvctxtwo\ctxholep{#1}}
\newcommand{\cbvctxthreep}[1]{\cbvctxthree\ctxholep{#1}}

\newcommand{\scbvctx}{S}
\newcommand{\scbvctxtwo}{\scbvctx'}
\newcommand{\scbvctxp}[1]{\scbvctx\ctxholep{#1}}
\newcommand{\scbvctxtwop}[1]{\scbvctxtwo\ctxholep{#1}}

\newcommand{\cbndctx}{N}
\newcommand{\cbndctxtwo}{\cbndctx'}
\newcommand{\cbndctxthree}{\cbndctx''}

\newcommand{\cbndctxp}[1]{\cbndctx\ctxholep{#1}}
\newcommand{\cbndctxtwop}[1]{\cbndctxtwo\ctxholep{#1}}
\newcommand{\cbndctxthreep}[1]{\cbndctxthree\ctxholep{#1}}

\newcommand{\tomach}{\Rew{}}
\newcommand{\tomachm}{\Rew{\mulsym}}
\newcommand{\tomachmone}{\Rew{\mulsym_1}}
\newcommand{\tomachmtwo}{\Rew{\mulsym_2}}
\newcommand{\tomache}{\Rew{\expsym}}
\newcommand{\tomacha}{\Rew{\admsym}}
\newcommand{\tomachaone}{\Rew{\admsym_1}}
\newcommand{\tomachatwo}{\Rew{\admsym_2}}

\newcommand{\tomachnoa}{\Rew{\noadmsym}}
\newcommand{\tomachx}{\Rew{\mathtt{x}}}
\newcommand{\admnf}[1]{\mathtt{nf}_{\admsym}(#1)}

\newcommand{\code}{\overline{\tm}}
\newcommand{\codetwo}{\overline{\tmtwo}}
\newcommand{\codethree}{\overline{\tmthree}}
\newcommand{\codefour}{\overline{\tmfour}}
\newcommand{\codeval}{\overline{\val}}

\newcommand{\clos}{c}

\newcommand{\decclos}{\decode{\clos}}

\newcommand{\env}{e}
\newcommand{\envtwo}{e'}
\newcommand{\envthree}{e''}
\newcommand{\envfour}{e'''}
\newcommand{\envp}[1]{\env\ctxholep{#1}}
\newcommand{\envtwop}[1]{\envtwo\ctxholep{#1}}
\newcommand{\envthreep}[1]{\envthree\ctxholep{#1}}

\newcommand{\decenv}{\decode{\env}}
\newcommand{\decenvtwo}{\decode{\envtwo}}
\newcommand{\decenvthree}{\decode{\envthree}}
\newcommand{\decenvfour}{\decode{\envfour}}
\newcommand{\decenvp}[1]{\decenv\ctxholep{#1}}
\newcommand{\decenvtwop}[1]{\decenvtwo\ctxholep{#1}}
\newcommand{\decenvthreep}[1]{\decenvthree\ctxholep{#1}}
\newcommand{\decenvfourp}[1]{\decenvfour\ctxholep{#1}}

\newcommand{\genv}{E}
\newcommand{\genvtwo}{E'}
\newcommand{\genvthree}{E''}

\newcommand{\genvthreep}[1]{\genvthree\ctxholep{#1}}

\newcommand{\decgenv}{\decode{\genv}}
\newcommand{\decgenvtwo}{\decode{\genvtwo}}

\newcommand{\decgenvp}[1]{\decgenv\ctxholep{#1}}
\newcommand{\decgenvtwop}[1]{\decgenvtwo\ctxholep{#1}}

\newcommand{\stempty}{\epsilon}
\newcommand{\cons}{::}
\newcommand{\fnst}[1]{\mathbf{f}(#1)}
\newcommand{\argst}[1]{\mathbf{a}(#1)}
\newcommand{\stack}{\pi}
\newcommand{\stacktwo}{\pi'}
\newcommand{\stackthree}{\pi''}
\newcommand{\stackp}[1]{\stack\ctxholep{#1}}

\newcommand{\decstack}{\decode{\pi}}

\newcommand{\decstackp}[1]{\decstack\ctxholep{#1}}

\newcommand{\state}{s}
\newcommand{\statetwo}{s'}
\newcommand{\statethree}{s''}
\newcommand{\kamstate}[3]{#1\mid#2\mid#3}
\newcommand{\cekstate}[3]{#1\mid#2\mid#3}
\newcommand{\lamstate}[3]{#1\mid#2\mid#3}
\newcommand{\mamstate}[3]{#1\mid#2\mid#3}
\newcommand{\scekstate}[4]{#1\mid#2\mid#3\mid#4}
\newcommand{\wamstate}[4]{#1\mid#2\mid#3\mid#4}
\newcommand{\mgwamstate}[3]{#1\mid#2\mid#3}
\newcommand{\pwamstate}[4]{#1\mid#2\mid#4\mid#3}

\newcommand{\envslice}{\!\upharpoonleft}
\newcommand{\envslicevar}[1]{\envslice_{#1}}
\newcommand{\pwamclosed}[1]{#1 \text{ is closed}}

\newcommand{\rename}[1]{#1^\alpha}

\newcommand{\exec}{\rho}

\newcommand{\decode}[1]{\llbracket #1\rrbracket}
\renewcommand{\decode}[1]{\underline{#1}}
\newcommand{\sizeadm}[1]{\size{#1}_{\admsym}}
\newcommand{\sizelog}[1]{\size{#1}_{\noadmsym}}

\newcommand{\fstack}{D}
\newcommand{\fstacktwo}{D'}

\newcommand{\dstackp}[1]{\decode{\stack}\ctxholep{#1}}
\newcommand{\dstacktwop}[1]{\decode{\stacktwo}\ctxholep{#1}}
\newcommand{\dstackthreep}[1]{\decode{\stackthree}\ctxholep{#1}}
\newcommand{\denvp}[1]{\decode{\env}\ctxholep{#1}}
\newcommand{\denvtwop}[1]{\decode{\envtwo}\ctxholep{#1}}
\newcommand{\denvsubonep}[1]{\decode{\env_1}\ctxholep{#1}}
\newcommand{\denvsubtwop}[1]{\decode{\env_2}\ctxholep{#1}}
\newcommand{\dfstackp}[1]{\decode{\fstack}\ctxholep{#1}}
\newcommand{\dgenvp}[1]{\decode{\genv}\ctxholep{#1}}

\newcommand{\dgenvsubonep}[1]{\decode{\genv_1}\ctxholep{#1}}
\newcommand{\dgenvsubtwop}[1]{\decode{\genv_2}\ctxholep{#1}}

\newcommand{\deriv}{d}
\newcommand{\derivtwo}{e}

\newcommand{\sizee}[1]{|#1|_{e}}
\newcommand{\sizem}[1]{|#1|_{m}}

\newcommand{\sizep}[1]{|#1|_p}

\newcommand{\headst}[1]{\mathbf{h}(#1)}

\usepackage{amsthm}

    \newtheorem{theorem}{Theorem}[section]
    \newtheorem{lemma}[theorem]{Lemma}
    \newtheorem{corollary}[theorem]{Corollary}
    \newtheorem{proposition}[theorem]{Proposition}
    \newtheorem{definition}[theorem]{Definition}

\newcommand{\distil}{{\tt D}}
\newcommand{\calculus}{{\tt C}}
\newcommand{\mach}{{\tt M}}

\newcommand{\wnamed}{well-named}

\renewcommand{\dump}{D}
\newcommand{\decdump}{\decode \dump}
\newcommand{\decdumpp}[1]{\decdump\ctxholep{#1}}

\newcommand{\closprop}{every closure in $\state$ is closed}
\newcommand{\globclosprop}[1]{the global closure #1 of $\state$ is closed}
\newcommand{\subprop}{any code in $\state$ is a literal subterm of $\code$}

\newcommand{\gwnameprop}{the global closure of $\state$ is \wnamed}
\newcommand{\lwnameprop}{any closure in $\state$ is \wnamed}
\newcommand{\envprop}{the length of any environment in $\state$ is bound by $\size\code$}
\newcommand{\genvprop}{the length of the global environment in $\state$ is bound by $\sizem\exec$}

\newcommand{\invariantshyp}[1]{be a #1 reachable state whose initial code $\code$ is \wnamed}

\newcommand{\globinvariantshyp}[1]{be a #1 state reached by an execution $\exec$ of initial \wnamed\ code $\code$}

\newcommand{\distillationStatement}{is a reflective distillery. In particular, on a reachable state $\state$ we have}

\newcommand{\pwammark}{\Box}
\newcommand{\dual}{\bot}

\newcommand{\support}[1]{\Delta(#1)}

\newcommand{\todb}{\Rew{\db}}
\newcommand{\tols}{\Rew{\ls}}
\newcommand{\toB}{\Rew{\Bsym}}
\newcommand{\togc}{\Rew{\gc}}
\newcommand{\toneq}{\Rew{\neq}}
\newcommand{\toap}{\Rew{@}}
\newcommand{\toabs}{\Rew{\l}}
\newcommand{\tovar}{\Rew{\vartt}}

\newcommand{\toeval}{\stackrel{\evctx}{\rightarrow}_{\esym}}

\newcommand{\tomachnoap}{\Rew{\neg c_1}}
\newcommand{\sizenoap}[1]{|#1|_{\neg c_1}}
\newcommand{\sizeap}[1]{|#1|_{c_1}}
\newcommand{\sizecomtwo}[1]{|#1|_{c_2}}

\newcommand{\WAM}{WAM}
\newcommand{\wam}{\WAM}

\newcommand{\withproofs}[1]{\ifthenelse{\boolean{withproofs}}{#1}{}}

\newcommand{\withoutproofs}[1]{\ifthenelse{\boolean{withproofs}}{}{#1}}

\newboolean{withproofs}
\setboolean{withproofs}{true}

\begin{document}

\setlength{\pdfpageheight}{\paperheight}
\setlength{\pdfpagewidth}{\paperwidth}

%\conferenceinfo{ICFP 2014}{Month d--d, 20yy, City, ST, Country} 
%\copyrightyear{20yy} 
%\copyrightdata{978-1-nnnn-nnnn-n/yy/mm} 
%\doi{nnnnnnn.nnnnnnn}

%\titlebanner{banner above paper title}        % These are ignored unless
%\preprintfooter{short description of paper}   % 'preprint' option specified.

\title{Distilling Abstract Machines (Long Version)}

\author{Beniamino Accattoli
\and 
Pablo Barenbaum
\and
Damiano Mazza
}

%\author{Beniamino Accattoli}
%%           {Carnegie Mellon University \& Universit\`a di Bologna}
%%           {beniamino.accattoli@gmail.com}
%\author{Pablo Barenbaum}
%%           {University of Buenos Aires -- CONICET}
%%           {pbarenbaum@dc.uba.ar}
%\author{Damiano Mazza}
%%           {CNRS, UMR 7030, LIPN, Universit\'e Paris 13, Sorbonne Paris Cit\'e}
%%           {Damiano.Mazza@lipn.univ-paris13.fr}
%
\maketitle

\begin{abstract}
% !TEX root = main.tex
It is well-known that many environment-based abstract machines can be seen as strategies in lambda calculi with explicit substitutions (ES). Recently, graphical syntaxes and linear logic led to the linear substitution calculus (LSC), a new approach to ES that is halfway between big-step calculi and traditional calculi with ES. This paper studies the relationship between the LSC and environment-based abstract machines. While traditional calculi with ES simulate abstract machines, the LSC rather distills them: some transitions are simulated while others vanish, as they map to a notion of structural congruence. The distillation process unveils that abstract machines in fact implement weak linear head reduction, a notion of evaluation having a central role in the theory of linear logic. We show that such a pattern applies uniformly in call-by-name, call-by-value, and call-by-need, catching many machines in the literature. We start by distilling the KAM, the CEK, and the ZINC, and then provide simplified versions of the SECD, the lazy KAM, and Sestoft's machine. Along the way we also introduce some new machines with global environments. Moreover, we show that distillation preserves the time complexity of the executions, i.e. the LSC is a complexity-preserving abstraction of abstract machines.

\end{abstract}

%\category{CR-number}{subcategory}{third-level}

% general terms are not compulsory anymore, 
% you may leave them out
%\terms
%term1, term2
%
%\keywords
%Abstract machines, explicit substitutions, linear logic.

% !TEX root = main.tex

\section{Introduction}
In the theory of higher-order programming languages, abstract machines and explicit substitutions are two tools used to model the execution of programs on real machines while omitting many details of the actual implementation. Abstract machines can usually be seen as evaluation strategies in calculi of explicit substitutions (see at least \cite{DBLP:journals/jfp/HardinM98,DBLP:journals/tocl/BiernackaD07,DBLP:journals/lisp/Lang07,DBLP:journals/lisp/Cregut07}), that can in turn be interpreted as small-step cut-elimination strategies in sequent calculi \cite{DBLP:journals/toplas/AriolaBS09}. 

Another tool providing a fine analysis of higher-order evaluation is \linlogic, especially via the new perspectives on cut-elimination provided by \emph{\pns}, its graphical syntax. Explicit substitutions (ES) have been connected to \linlogic\ by Kesner and co-authors in a sequence of works \cite{DBLP:journals/mscs/CosmoKP03,DBLP:journals/iandc/KesnerL07,DBLP:conf/mfcs/KesnerR09}, culminating in the \emph{linear substitution calculus} (\lsc), a new formalism with ES behaviorally isomorphic to \pns\ (introduced in~\cite{DBLP:conf/csl/AccattoliK10}, developed in~\cite{DBLP:conf/flops/AccattoliP12,DBLP:conf/rta/Accattoli12,DBLP:conf/rta/AccattoliL12,DBLP:journals/corr/abs-1302-6337,DBLP:conf/popl/AccattoliBKL14}, and bearing similarities with calculi by De Bruijn~\cite{deBruijn87}, Nederpelt~\cite{Ned92}, and Milner~\cite{DBLP:journals/entcs/Milner07}). Since linear logic can model all evaluation schemes (call-by-name/value/need) \cite{DBLP:journals/tcs/MaraistOTW99}, the LSC can express them modularly, by minor variations on rewriting rules and evaluation contexts. In this paper we revisit the relationship between environment-based abstract machines and ES. Traditionally, calculi with ES simulate machines. The \lsc, instead, distills them. 

\ben{\paragraph{A Bird's Eye View.} In a simulation, every machine transition is simulated by some steps in the calculus with ES. In a distillation---that will be a formal concept, not just an analogy---only some of the machine transitions are simulated, while the others are mapped to the structural equivalence of the calculus, a specific trait of the \lsc. Now, structural equivalence commutes with evaluation, \ie\ it can be postponed. Thus, the transitions sent on the structural congruence fade away, without compromising the result of evaluation. Additionally, we show that machine executions and their distilled representation in the \lsc\  have the same asymptotic length, \ie\ the distillation process preserves the complexity of evaluation. The main point is that the \lsc\ is an arguably simpler than abstract machines, and also---as we will show---it can uniformly represent many different machines in the literature.}

\paragraph{Traditional vs Contextual ES.} Traditional calculi with ES (see \cite{DBLP:journals/corr/abs-0905-2539} for a survey) implement $\beta$-reduction $(\l \var. \tm)\tmtwo  \Rew{\beta}  \tm\isub\var\tmtwo$ 
introducing an annotation (the explicit substitution $\esub\var\tmtwo$),
\[\begin{array}{rcl}
		(\l \var. \tm)\tmtwo & \toB & \tm\red{\esub\var\tmtwo} 
\end{array}\]
and percolating it through the term structure,
\begin{equation}
	\begin{array}{rcl}
		(\tm\tmthree)\red{\esub\var\tmtwo} & \toap & \tm\red{\esub\var\tmtwo}\tmthree\red{\esub\var\tmtwo}\\
	(\l\var.\tm)\red{\esub\vartwo\tmtwo}	& \toabs & \l\var.\tm\red{\esub\vartwo\tmtwo}\\
	\end{array}
	\label{eq:percolations}
\end{equation}
until
they reach variable occurrences on which they finally substitute or
get garbage collected,
\[\begin{array}{rcl}
		\var\red{\esub\var\tmtwo}			& \tovar & \red{\tmtwo}\\
		\vartwo\red{\esub\var\tmtwo}	& \toneq & \vartwo
	\end{array}\]

 The LSC, instead, is based on a \emph{contextual} view of evaluation and substitution, also known as \emph{at a distance}. The idea is that one can get rid of the rules percolating through the term structure---\ie\ $@$ and $ \l$---by introducing contexts $\ctx$ (i.e. terms with a hole $\ctxhole$) and generalizing the base cases, obtaining just two rules, \emph{linear substitution} (ls) and \emph{garbage collection} (gc):
\[\begin{array}{rcl@{\sep}ll}
		\ctxp\var\red{\esub\var\tmtwo}			& \tols & \ctxp{\red{\tmtwo}}\red{\esub\var\tmtwo}\\
		\tm\red{\esub\var\tmtwo}	& \togc & \tm&\mbox{ if $\var\notin\fv\tm$}
	\end{array}\]
Dually, the rule creating substitutions ($\Bsym$) is generalized to act up to a context of substitutions $\red{\esub{\black{\ldots}}{\black{\ldots}}}\defeq \red{\esub{\var_1}{\tmthree_1}}\ldots\red{\esub{\var_k}{\tmthree_k}}$ obtaining rule $\db$ ($\Bsym$ at a {\tt d}istance):
\[\begin{array}{rll}
		(\l \var. \tm)\red{\esub{\black{\ldots}}{\black{\ldots}}} \tmtwo & \todb & \tm\red{\esub\var\tmtwo} \red{\esub{\black{\ldots}}{\black{\ldots}}}
\end{array}\]

\paragraph{Logical Perspective on the \lsc.} \sloppy{From a sequent calculus point of view, rules $ @$ and $ \l$, corresponding to \emph{commutative} cut-elimination cases, are removed and integrated---via the use of contexts---directly in the definition of the \emph{principal} cases $\Bsym$, $\vartt$ and $\neq$, obtaining the contextual rules $\db$, $\ls$, and $\gc$.} This is the \cben{term analogous}{analogous for terms} of the removal of commutative cases provided by \pns. From a linear logic point of view, $\todb$ can be identified with the multiplicative cut-elimination case $\tom$, while $\tols$ and $\togc$ correspond to exponential cut-elimination. Actually, garbage collection has a special status, as it can always be postponed. We will then identify exponential cut-elimination $\toe$ with linear substitution $\tols$ alone. 

The LSC has a simple meta-theory, and is halfway between traditional calculi with ES---with whom it shares the small-step dynamics---and $\l$-calculus---of which it retains most of the simplicity.

\paragraph{Distilling Abstract Machines.}	Abstract machines implement the traditional approach to ES, by 
	\begin{enumerate}
		\item \emph{Weak Evaluation}: forbidding reduction under abstraction (no rule $\Rew{\l}$ in \refeq{percolations}), 
		\item \emph{Evaluation Strategy}: looking for redexes according to some notion of weak evaluation context $\evctx$,
		\item \emph{Context Representation}: using environments $\env$ (aka lists of substitutions) and stacks $\stack$ (lists of terms) to keep track of the current evaluation context. 
	\end{enumerate}

The \lsc\ distills---\ie\ factorizes---abstract machines. The idea is that one can represent the strategy of an abstract machine by directly plugging the evaluation context in the contextual substitution/exponential rule, obtaining:
\[\begin{array}{rcl@{\sep}ll}
		\evctxp\var\red{\esub\var\tmtwo}			& \toeval & \evctxp{\red{\tmtwo}}\red{\esub\var\tmtwo}\\
\end{array}\]
and factoring out the parts of the machine that just look for the next redex to reduce. By defining $\togen$ as the closure of $\toeval$ and $\tom$ by evaluation contexts $\evctx$, one gets a clean representation of the machine strategy.

The mismatch between the two approaches is in rule $\toap$, that contextually---by nature---cannot be captured.  In order to get out of this
\emph{cul-de-sac}, the very idea of \emph{simulation} of an abstract machine must be refined to that of \emph{distillation}. 

The crucial observation is that the equivalence $\tostruct$ induced by $\toap\cup\togc$ has the same special status of $\togc$, \ie\ it can be postponed without affecting reduction lengths. More abstractly, $\tostruct$ is a \emph{strong bisimulation} with respect to $\togen$, \ie\ it verifies (note \emph{one} step to \emph{one} step, and \emph{viceversa}) 
\ben{\begin{center}
\begin{tabular}{c@{\sep}c@{\sep}c}
 \begin{tikzpicture}[ocenter]
  \node (s) {\normalsize$\tm$};
  \node at (s.center)  [below =0.8*\nodeVerDist](s2) {\normalsize$\tmtwo$};
  \node at (s.center) [right= 0.8*\nodeHorDist](t) {\normalsize$\tmfour$};
%  \node at (s2-|t) [](s1){\normalsize$\tmfive$};

  \node at (s.center)[anchor = center, below=0.3*\nodeVerDist](eq1){\normalsize$\tostruct$};
  %\node at (t.center)[anchor = center, below=0.3*\nodeVerDist](eq2){\normalsize$\tostruct$};
    \draw[-o] (s) to  (t);
	%\draw[-o, dashed] (s2) to  (s1);
\end{tikzpicture} 

&  $\Rightarrow \exists \tmfive$ s.t. & 
 \begin{tikzpicture}[ocenter]
  \node (s) {\normalsize$\tm$};
  \node at (s.center)  [below =0.8*\nodeVerDist](s2) {\normalsize$\tmtwo$};
  \node at (s.center) [right= 0.8*\nodeHorDist](t) {\normalsize$\tmfour$};
  \node at (s2-|t) [](s1){\normalsize$\tmfive$};

  \node at (s.center)[anchor = center, below=0.3*\nodeVerDist](eq1){\normalsize$\tostruct$};
  \node at (t.center)[anchor = center, below=0.3*\nodeVerDist](eq2){\normalsize$\tostruct$};
    \draw[-o] (s) to  (t);
	\draw[-o] (s2) to  (s1);
\end{tikzpicture} 
\end{tabular}
\end{center}
and
\begin{center}
\begin{tabular}{c@{\sep}c@{\sep}c}
 \begin{tikzpicture}[ocenter]
  \node (s) {\normalsize$\tm$};
  \node at (s.center)  [below =0.8*\nodeVerDist](s2) {\normalsize$\tmtwo$};
  \node at (s.center) [right= 0.8*\nodeHorDist](t) {};
  \node at (s2-|t) [](s1){\normalsize$\tmfive$};

  \node at (s.center)[anchor = center, below=0.3*\nodeVerDist](eq1){\normalsize$\tostruct$};
  %\node at (t.center)[anchor = center, below=0.3*\nodeVerDist](eq2){\normalsize$\tostruct$};
    %\draw[-o] (s) to  (t);
	\draw[-o] (s2) to  (s1);
\end{tikzpicture} 

&  $\Rightarrow \exists \tmfour$ s.t. & 
 \begin{tikzpicture}[ocenter]
  \node (s) {\normalsize$\tm$};
  \node at (s.center)  [below =0.8*\nodeVerDist](s2) {\normalsize$\tmtwo$};
  \node at (s.center) [right= 0.8*\nodeHorDist](t) {\normalsize$\tmfour$};
  \node at (s2-|t) [](s1){\normalsize$\tmfive$};

  \node at (s.center)[anchor = center, below=0.3*\nodeVerDist](eq1){\normalsize$\tostruct$};
  \node at (t.center)[anchor = center, below=0.3*\nodeVerDist](eq2){\normalsize$\tostruct$};
    \draw[-o] (s) to  (t);
	\draw[-o] (s2) to  (s1);
\end{tikzpicture} 
\end{tabular}
\end{center}}
Now,  $\tostruct$ can be considered as a \emph{structural equivalence} on the language. Indeed, the strong bisimulation property states that the transformation expressed by $\tostruct$ is irrelevant with respect to $\togen$, in particular $\tostruct$-equivalent terms have $\togen$-evaluations of the same length ending in $\tostruct$-equivalent terms (and this holds even locally). 

Abstract machines then are \emph{distilled}: the logically relevant part of the substitution process is retained by $\togen$ while both the search of the redex $\toap$ and garbage collection $\toneq$ are isolated into the equivalence $\tostruct$. Essentially, $\togen$ captures principal cases of cut-elimination while $\tostruct$ encapsulate the commutative ones (plus garbage collection, corresponding to principal cut-elimination involving weakenings).

\paragraph{Case Studies.} We will analyze along these lines many abstract machines. Some are standard (KAM \cite{DBLP:journals/lisp/Krivine07}, CEK \cite{Felleisen:1986:CEK}, ZINC \cite{Leroy-ZINC}), some are new (MAM, WAM), and of others we provide simpler versions (SECD \cite{LandinSECD}, Lazy KAM \cite{DBLP:journals/lisp/Cregut07,DBLP:conf/ppdp/DanvyZ13}, Sestoft's \cite{Sestoft}). The previous explanation is a sketch of the distillation of the KAM, but the approach applies \emph{mutatis mutandis} to all the other machines, encompassing most realizations of call-by-name, call-by-value, and call-by need evaluation. The main contribution of the paper is indeed a modular \emph{contextual} theory of abstract machines. We start by distilling some standard cases, and then rationally reconstruct and simplify non-trivial machines as the SECD, the lazy KAM, and Sestoft's abstract machine for call-by-need (deemed SAM), by enlightening their mechanisms as different encoding of evaluation contexts, modularly represented in the \lsc. 

\paragraph{Call-by-Need.} Along the way, we show that the contextual (or \emph{at a distance}) approach of the \lsc\ naturally leads to simple machines with just one global environment, as the newly introduced MAM (M for Milner). Such a feature is then showed to be a key ingredient of call-by-need machines, by using it to introduce a new and simple call-by-need machine, the WAM (W for Wadsworth), and then showing how to obtain (simplifications of) the Lazy KAM and the SAM by simple tweaks.

\paragraph{Distillation Preserves Complexity.} It is natural to wonder what is lost in the distillation process. What is the asymptotic impact of distilling machine executions into $\togen$? Does it affect in any way the complexity of evaluation? We will show that \emph{nothing is lost}, as machine executions are only linearly longer than $\togen$. More precisely, they are \emph{bilinear}, \ie\ they are linear in 1) the length of $\togen$, and in 2) the size $\size\tm$ of the starting term $\tm$. In other words, the search of redexes and garbage collection can be safely ignored in quantitative (time) analyses, \ie\ the \lsc\ and $\togen$ provide a complexity-preserving abstraction of abstract machines. While in call-by-name and call-by-value such an analysis follows from an easy local property of machine executions, the call-by-need case is subtler, as such a local property does not hold and bilinearity can be established only via a global analysis.

\paragraph{Linear Logic and Weak Linear Head Reduction.} Beyond the contextual view, our work also unveils a deep connection between abstract machines and \linlogic. The strategies modularly encoding the various machines (generically noted $\togen$ and parametric in a fixed notion of evaluation contexts) are in fact call-by-name/value/need versions of \emph{weak linear head reduction} (\wlhr), a fundamental notion in the theory of \linlogic\ \cite{DBLP:journals/corr/abs-1302-6337,DBLP:journals/tcs/MascariP94,DBLP:conf/lics/DanosHR96,DBLP:conf/cie/EhrhardR06,DBLP:conf/fossacs/Clairambault11}. This insight ---- due to Danos and Regnier for the KAM \cite{Danos04headlinear}---is not ours, but we develop it in a simpler and tighter way, modularly lifting it to many other abstract machines.

\paragraph{Call-by-Name.} The call-by-name case (catching the KAM and the new MAM) is in fact special, as our distillation theorem has three immediate corollaries, following from results about \wlhr\ in the literature:
\begin{enumerate}
	\item \emph{Invariance}: it implies that the length of a KAM/MAM execution is an  invariant time cost model (\ie\ polynomially related to, say, Turing machines, in both directions), given that in \cite{DBLP:conf/rta/AccattoliL12} the same is shown for \wlhr.
	\item \emph{Evaluation as Communication}: we implicitly establish a link between the KAM/MAM and the $\pi$-calculus, given that the evaluation of a term via \wlhr\ is isomorphic to evaluation via Milner's encoding in the $\pi$-calculus \cite{DBLP:journals/corr/abs-1302-6337}.
	\item \emph{Plotkin's Approach}: our study complements the recent \cite{DBLP:conf/popl/AccattoliBKL14}, where it is shown that \wlhr\ is a standard strategy of the \lsc. The two works together provide the lifting to explicit substitutions of Plotkin's approach of relating a machine (the SECD machine in that case, the KAM/MAM in ours) and a calculus (the call-by-value $\l$-calculus and the \lsc, respectively) via a standardization theorem and a standard strategy \cite{DBLP:journals/tcs/Plotkin75}.
\end{enumerate}

\ben{\paragraph{Beyond Abstract Machines.} This paper is just an episode---the one about abstract machines---in the recent \emph{feuilleton} about complexity analysis of functional languages via linear logic and rewriting theory, starring the linear substitution calculus. The story continues in \cite{usefulred} and \cite{valuevariables}. In \cite{usefulred}, the \lsc\ is used to prove that the length of leftmost-outermost $\beta$-reduction is an invariant cost-model for $\l$-calculus (\ie\ it is a measure polynomially related to evaluation in classic computational models like Turing machines or random access machines), solving a long-standing open problem in the theory of $\l$-calculus. Instead, \cite{valuevariables} studies the asymptotic  number of exponential steps (for $\togen$) in terms of the number of multiplicative steps, in the call-by-name/value/need \lsc. Via the results presented here, \cite{valuevariables} establishes a polynomial relationship between the exponential and the multiplicative transitions of abstract machines, complementing our work and implying that distillation can be pushed forward, forgetting exponential steps too.

\paragraph{Related Work.} Beyond the already cited works, Danvy and coauthors have studied abstract machines in a number of works. In some of them, they shows how to extract a functional evaluator from an abstract machine via a sequence of transformations (closure conversion, CPS, and defunctionalization) \cite{DBLP:conf/ppdp/AgerBDM03,DBLP:journals/ipl/AgerDM04,DBLP:conf/ifl/Danvy04}. Such a study is orthogonal in spirit to what we do here. The only point of contact is the \emph{rational deconstruction of the SECD} in \cite{DBLP:conf/ifl/Danvy04}, that is something that we also do, but in a different way. Another sequence of works studies the relationship between abstract machines and calculi with ES \cite{DBLP:journals/tcs/BiernackaD07,DBLP:journals/tocl/BiernackaD07,DBLP:conf/ppdp/DanvyZ13}, and it is clearly closer to our topic, except that 1) \cite{DBLP:journals/tcs/BiernackaD07,DBLP:journals/tocl/BiernackaD07} follow the traditional (rather than the contextual) approach to ES; 2) none of these works deals with complexity analysis nor with linear logic. On the other hand, \cite{DBLP:journals/tocl/BiernackaD07} provides a deeper analysis of Leroy's ZINC machine, as ours does not account for the avoidance of needless closure creations that is a distinct feature of the ZINC. Last, what here we call \emph{commutative transitions} essentially corresponds to what Danvy and Nielsen call \emph{decompose} phase in \cite{Danvy04refocusingin}. 

The call-by-need calculus we use---that is a contextual re-formulation of Maraist, Odersky, and Wadler's calculus \cite{DBLP:journals/jfp/MaraistOW98}---is a novelty of this paper. It is simpler than both Ariola and Felleisen's \cite{DBLP:journals/jfp/AriolaF97} and Maraist, Odersky, and Wadler's calculi because it does not need any re-association axioms. Morally, it is a version with let-bindings (avatars of ES) of Chang and Felleisen's calculus \cite{DBLP:conf/esop/ChangF12}. A similar calculus is used by Danvy and Zerny in \cite{DBLP:conf/ppdp/DanvyZ13}. Another call-by-need machine, with whom we do not deal with, appears in \cite{DBLP:conf/popl/GarciaLS09}.}

\paragraph{Proofs.} Some proofs have been omitted for lack of space. They can be found in the longer version \cite{distillingTR}.
% !TEX root = main.tex
\section{Preliminaries on the Linear Substitution Calculus}
\label{sect:ES-distance}
\emph{Terms and Contexts}. The language of the \emph{weak linear substitution calculus} (\wlsc) is generated by the following grammar:
\begin{center}
$\begin{array}{lll@{\sep\sep\sep}llllllllll}
	\tm,\tmtwo,\tmthree,\tmfour,\tmfive,\tmsix &\grameq& \var\mid \val\mid \tm\tmtwo\mid  \tm\esub\var\tmtwo & \val &\grameq \l\var.\tm
\end{array}
$\end{center}
The constructor $\tm\esub{\var}{\tmtwo}$ is called an \emph{explicit substitution} (of $\tmtwo$ for $\var$ in $\tm$). The usual (implicit) substitution is instead denoted by $\tm\isub{\var}{\tmtwo}$. Both $\l \var. \tm$ and $\tm\esub{\var}{\tmtwo}$ bind $\var$ in $\tm$, with the usual notion of $\alpha$-equivalence. Values, noted $\val$, do not include variables: this is a standard choice in the study of abstract machines.

Contexts are terms with one occurrence of the hole $\ctxhole$, an additional constant. We will use many different contexts. The most general ones will be \emph{weak contexts} $\wctx$ (\ie\ not under abstractions), which are defined by:
\begin{center}
$\begin{array}{lll}
	\wctx,\wctxtwo &\grameq& \ctxhole\mid \wctx\tmtwo \mid \tm\wctx\mid  \wctx\esub\var\tmtwo \mid \tm\esub\var\wctx
\end{array}
$\end{center}
The \emph{plugging} $\wctxp\tm$ (resp. $\wctxp\wctxtwo$) of a term $\tm$ (resp. context $\wctxtwo$) in a context $\wctx$ is defined as $\ctxholep\tm\defeq\tm$ (resp. $\ctxholep\wctxtwo\defeq\wctxtwo$), $(\wctx\tm)\ctxholep\tmtwo\defeq \wctx\ctxholep\tmtwo\tm$ (resp. $(\wctx\tm)\ctxholep\wctxtwo\defeq \wctx\ctxholep\wctxtwo\tm$), and so on. The set of free variables of a term $\tm$ (or context $\wctx$) is denoted by $\fv\tm$ (resp. $\fv\wctx$). Plugging in a context may capture free variables (replacing holes on the left of substitutions). These notions will be silently extended to all the contexts used in the paper.

\emph{Rewriting Rules}. On the above terms, one may define several variants of the \lsc\  by considering two elementary rewriting rules, \emph{distance-$\beta$} (\db) and \emph{linear substitution} (\lssym), each one coming in two variants, call-by-name and call-by-value (the latter variants being abbreviated by \dbv\ and \lsvsym), and pairing them in different ways and with respect to different evaluation contexts.

The rewriting rules rely in multiple ways on contexts. We start by defining \emph{substitution contexts}, generated by
\begin{center}
$\begin{array}{lllllllllllll}
	\sctx &\grameq  \ctxhole \mid \sctx\esub{\var}{\tm}.
\end{array}
$\end{center}
A term of the form $\sctxp\val$ is an \emph{answer}. Given a family of contexts $\ctx$, the two variants of the elementary rewriting rules, also called \emph{root rules}, are defined as follows:
\begin{center}
$\begin{array}{rcl}
	\sctxp{\l\var.\tm}\tmtwo &\rtodb &\sctxp{\tm\esub\var\tmtwo}\\
	\sctxp{\l\var.\tm}\sctxtwop{\val}&\rtodbv &\sctxp{\tm\esub\var{\sctxtwop{\val}}} \\
	\ctxp\var\esub\var\tmtwo &\rtols &\ctxp\tmtwo\esub\var\tmtwo\\
	\ctxp\var\esub\var{\sctxp{\val}} &\rtolsv &\sctxp{\ctxp\val\esub\var\val}
\end{array}
$\end{center}
In the linear substitution rules, we assume that $\var\in\fv{\ctxp\var}$, \ie, the context $\ctx$ does not capture the variable $\var$, and we also silently work modulo $\alpha$-equivalence to avoid variable capture in the rewriting rules. Moreover, we use the notations $\rtolsc{\ctx}$ and $\rtolsvc{\ctx}$ to specify the family of contexts used by the rules, with $\ctx$ being the meta-variable ranging over such contexts.

All of the above rules are \emph{at a distance} (or \emph{contextual}) because their definition involves contexts. Distance-$\beta$ and linear substitution correspond, respectively, to the so-called \emph{multiplicative} and \emph{exponential} rules for cut-elimination in \pns. The presence of contexts is how locality on \pns\ is reflected on terms.

\ben{The rewriting rules decompose the usual small-step semantics for $\l$-calculi, by substituting one occurrence at the time, and only when such an occurrence is in evaluation position. We emphasise this fact saying that we adopt a \emph{micro-step semantics}.}

A linear substitution calculus is defined by a choice of root rules, \ie, one of $\db/\dbv$ and one of $\lssym/\lsvsym$, and a family of \emph{evaluation contexts}. The chosen distance-$\beta$ (resp.\ linear substitution) root rule is generically denoted by $\rtogenm$ (resp.\ $\rtogene$). If $\evctx$ ranges over a fixed notion of evaluation context, the context-closures of the root rules are denoted by $\togenm\defeq\evctxp{\rtogenm}$ and $\togene\defeq\evctxp{\rtogene}$, where $\msym$ (resp. $\esym$) stands for \emph{multiplicative} (\emph{exponential}). The rewriting relation defining the calculus is then $\togen\defeq\togenm\cup\togene$.

\begin{table*}[t]
	\begin{center}
		{\setlength{\tabcolsep}{0.4em}
		\begin{tabular}{|l|l|c|c|c|c|}
			\hline
			Calculus & Evaluation contexts & $\rtogenm$ & $\rtogene$ & $\togenm$ & $\togene$ \\
			\hline
			$\wlscname$ & $\whctx \grameq \ctxhole\mid \whctx \tm \mid\whctx\esub{\var}{\tm}$ & $\rtodb$ & $\rtolsc{\whctx}$ &$\whctxp{\rtodb}$& $\whctxp{\rtolsc{\whctx}}$\\
			$\wlscvaluelr$ & $\cbvctx \grameq \ctxhole\mid \cbvctx\tm\mid \sctxp\val\cbvctx\mid \cbvctx\esub{\var}{\tm}$ & $\rtodbv$ & $\rtolsvc{\cbvctx}$ &$\cbvctxp{\rtodb}$& $\cbvctxp{\rtolsc{\cbvctx}}$\\
			$\wlscvaluerl$ & $\scbvctx \grameq \ctxhole\mid \scbvctx\sctxp\val\mid \tm\scbvctx\mid \scbvctx\esub{\var}{\tm}$ & $\rtodbv$ & $\rtolsvc{\scbvctx}$ &$\scbvctxp{\rtodb}$& $\scbvctxp{\rtolsc{\scbvctx}}$\\
			$\wlscneed$ & $\cbndctx \grameq \ctxhole\mid \cbndctx\tm\mid \cbndctx\esub\var\tm\mid \cbndctxtwop\var\esub\var\cbndctx$ & $\rtodb$ & $\rtolsvc{\cbndctx}$ &$\cbndctxp{\rtodb}$& $\cbndctxp{\rtolsc{\cbndctx}}$\\
			\hline
		\end{tabular}}
	\end{center}
	\caption{The four linear substitution calculi.}
	\label{tab:Calculi}
\end{table*}
\subsection{Calculi}
We consider four calculi, noted $\wlscname$, $\wlscvaluelr$, $\wlscvaluerl$, and $\wlscneed$, and defined in \reftab{Calculi}. They correspond to four standard evaluation strategies for functional languages. We are actually slightly abusing the terminology, because---as we will show---they are \emph{deterministic} calculi and thus should be considered as strategies. Our abuse is motivated by the fact that they are not strategies in the same calculus. The essential property of all these four calculi is that they are deterministic, because they implement a reduction strategy.
\begin{proposition}[Determinism]
	\label{prop:GenDet}
	\label{prop:CbNDet}
	\label{prop:CbvCekDet}
	\label{prop:CbvLamDet}
	\label{prop:CbNeedDet}
	The reduction relations of the four calculi of \reftab{Calculi} are deterministic: in each calculus, if $\evctx_1,\evctx_2$ are evaluation contexts and if $\tmfour_1,\tmfour_2$ are redexes (\ie, terms matching the left hand side of the root rules defining the calculus), $\evctx_1\ctxholep{\tmfour_1}=\evctx_2\ctxholep{\tmfour_2}$ implies $\evctx_1=\evctx_2$ and $\tmfour_1=\tmfour_2$, so that there is at most one way to reduce a term.
\end{proposition}
\begin{proof}
	\withproofs{
	See \refsect{determ-proofs} in the appendix (page \pageref{sect:determ-proofs}).}
	\withoutproofs{See \cite{distillingTR}.}
	%See \refsect{ProofCbNDet} for call-by-name, \refsect{ProofCbvCekDet} for left-to-right call-by-value, \refsect{ProofCbvLamDet} for right-to-left call-by-value and \refsect{ProofCbNeedDet} for call-by-need.
\end{proof}

\paragraph{Call-by-Name (CBN).} The evaluation contexts for $\wlscname$ are called \emph{weak head contexts} and\ben{---when paired with micro-step evaluation---}implement a strategy known as \emph{weak linear head reduction}. The original presentation of this strategy does not use explicit substitutions \cite{DBLP:journals/tcs/MascariP94,Danos04headlinear}. The presentation in use here has already appeared in \cite{DBLP:journals/corr/abs-1302-6337,DBLP:conf/popl/AccattoliBKL14} (see also \cite{DBLP:conf/rta/Accattoli12,DBLP:conf/rta/AccattoliL12}) as the weak head strategy of the \emph{linear substitution calculus} (which is obtained by considering \emph{all} contexts as evaluation contexts), and it avoids many technicalities of the original one. In particular, its relationship with the KAM is extremely natural, as we will show.

\ben{Let us give some examples of evaluation. Let $\delta \defeq \l\var.(\var\var)$ and consider the usual diverging term $\Omega \defeq \delta\delta$. In $\wlscname$ it evaluates---diverging---as follows:
\[\begin{array}{lclcccccccc}
	\delta\delta	& \togenm	& (\var\var)\esub\var\delta			& \togene	&\\
						&				&	(\delta\var)\esub\var\delta		& \togenm	&\\
						&				&	(\vartwo\vartwo)\esub\vartwo\var\esub\var\delta		& \togene	&\\
						&				&	(\var\vartwo)\esub\vartwo\var\esub\var\delta			& \togene	&\\
						&				&	(\delta\vartwo)\esub\vartwo\var\esub\var\delta			& \togenm	&\\
						&				&	(\varthree\varthree)\esub\varthree\vartwo\esub\vartwo\var\esub\var\delta			& \togene	&\ldots
\end{array}\]

Observe that according to our definitions both $\l\var.\Omega$ and $\var\Omega$ are $\togen$-normal for $\wlscname$, because evaluation does not go under abstractions, nor on the right of a variable (but terms like $\var\Omega$ will be forbidden, as we will limit ourselves to closed terms). Now let show the use of the context $\sctx$ in rule $\togenm$. Let $I \defeq \l \vartwo.\vartwo$ and $\tau \defeq (\l\varthree.\delta) I$, and  consider the following variation over $\Omega$, where rule $\togenm$ is applied with $\sctx \defeq \ctxhole\esub\varthree I$:
\[\begin{array}{lclcccccccc}
	\tau\tau	& \togenm	& \delta\esub\varthree I	\tau		& \togenm	&	(\var\var)\esub\var\tau\esub\varthree I	& \togene & \ldots\\
\end{array}\]
}

\paragraph{Call-by-Value (CBV).} For call-by-value calculi, \emph{left-to-right} ($\wlscvaluelr$) and \emph{right-to-left} ($\wlscvaluerl$) refer to the evaluation order of applications, \ie\ they correspond to \emph{operator first} and \emph{argument first}, respectively. The two calculi we consider here can be seen as strategies of a micro-step variant of the \emph{value substitution calculus}, the (small-step) call-by-value calculus at a distance introduced and studied in \cite{DBLP:conf/flops/AccattoliP12}. 

\ben{As an example, we consider again the evaluation of $\Omega$. In $\wlscvaluelr$ it goes as follows:
\[	\begin{array}{lllllll}
		\delta\delta & \togenm & (\var_1\var_1)\esub{\var_1}\delta & \togene &\\
		&& (\delta\var_1)\esub{\var_1}\delta & \togene\\
		&& (\delta\delta)\esub{\var_1}\delta & \togenm\\
		&& (\var_2\var_2)\esub{\var_2}{\delta}\esub{\var_1}\delta & \togene\\
		&& (\delta\var_2)\esub{\var_2}{\delta}\esub{\var_1}\delta & \togene \ldots
		\end{array}\]
While in $\wlscvaluerl$ it takes the following form:
\[	\begin{array}{lllllll}
		\delta\delta & \togenm & (\var_1\var_1)\esub{\var_1}\delta & \togene &\\
		&& (\var_1\delta)\esub{\var_1}\delta & \togene\\
		&& (\delta\delta)\esub{\var_1}\delta & \togenm\\
		&& (\var_2\var_2)\esub{\var_2}{\delta}\esub{\var_1}\delta & \togene\\
		&& (\var_2\delta)\esub{\var_2}{\delta}\esub{\var_1}\delta & \togene\ldots
		\end{array}\]
Note that the CBV version of $\togenm$ and $\togene$ employ substitution contexts $\sctx$ in a new way. An example of their use is given by the term $\tau\tau$ consider before for CBN. For instance, in $\wlscvaluelr$ we have:
\[\begin{array}{lclcccccccc}
	\tau\tau	& \togenm	& \delta\esub\varthree I	\tau		& \togenm	&	\\
				&				&  \delta\esub\varthree I	(\delta\esub\varthree I)	& \togenm & \\
				&				&  (\var\var)\esub\var{\delta\esub\varthree I}\esub\varthree I & \togene\\
				&				&  (\delta\var)\esub\var{\delta\esub\varthree I}\esub\varthree I &\ldots\\
\end{array}\]}

\paragraph{Call-by-Need (CBNeed).} The call-by-need calculus $\wlscneed$\ is a novelty of this paper, and can be seen either as a version at a distance of the calculi of \cite{DBLP:journals/jfp/MaraistOW98,DBLP:journals/jfp/AriolaF97} or as a version with explicit substitution of the one in \cite{DBLP:conf/esop/ChangF12}. It fully exploits the fact that the two variants of the root rules may be combined: the $\beta$-rule is call-by-name, which reflects the fact that, operationally, the strategy is \emph{by name}, but substitution is call-by-value, which forces arguments to be evaluated before being substituted, reflecting the \emph{by need} content of the strategy. Please note the definition of CBNeed evaluation contexts $\cbndctx$ in \reftab{Calculi}. They extend the weak head contexts for call-by-name with a clause ($\cbndctxtwop\var\esub\var\cbndctx$) turning them into \emph{hereditarily weak head contexts}. This new clause is how sharing is implemented by the reduction strategy. The general (non-deterministic) calculus is obtained by closing the root rules by \emph{all} contexts, but its study is omitted. What we deal with here can be thought as its standard strategy (stopping on a sort of weak head normal form).

\ben{Let us show, once again, the evaluation of $\Omega$. 
\[ \begin{array}{lllllll}
		\delta\delta & \togenm & (\var_1\var_1)\esub{\var_1}\delta & \togene &\\
		&& (\delta\var_1)\esub{\var_1}\delta & \togenm\\
		&& (\var_2\var_2)\esub{\var_2}{\var_1}\esub{\var_1}\delta & \togene\\
		&& (\var_2\var_2)\esub{\var_2}{\delta}\esub{\var_1}\delta & \togene\\
		&& (\delta\var_2)\esub{\var_2}{\delta}\esub{\var_1}\delta & \togenm\\
		&& (\var_3\var_3)\esub{\var_3}{\var_2}\esub{\var_2}{\delta}\esub{\var_1}\delta & \togene\\
		&& (\var_3\var_3)\esub{\var_3}{\delta}\esub{\var_2}{\delta}\esub{\var_1}\delta & \togene\\
		&& (\delta\var_3)\esub{\var_3}{\delta}\esub{\var_2}{\delta}\esub{\var_1}\delta & \togenm &\ldots\\
\end{array}\]

Note the difference with CBN: hereditarily weak evaluation contexts allow the micro-step substitution rule to replace variable occurrences in explicit substitutions.

As shown by the evaluation of $\Omega$, the various calculi considered in this paper not only select different $\beta$-redexes, they are also characterised by different substitution processes. Such processes are the object of a detailed analysis in the companion paper \cite{valuevariables}.}

\begin{figure*}[t]
	\begin{center}
	$\begin{array}{rlll@{\hspace*{0.25cm}}|@{\hspace*{0.25cm}}rlll}
		\tm\esub{\var}{\tmtwo} &\tostructgc&  \tm&\mbox{if $\var\notin\fv{\tm}$}&
		%\tm\esub{\var}{\sctxp{\val}} &\tostructgcv&  \tm&\mbox{if $\var\notin\fv{\tm}$}\\
		\tm\esub{\var}{\tmtwo} &\tostructdup&  \varsplit{\tm}{\var}{\vartwo}\esub{\var}{\tmtwo}\esub{\vartwo}{\tmtwo}\\
		\tm\esub{\var}{\tmtwo}\esub{\vartwo}{\tmthree} &\tostructcom& \tm\esub{\vartwo}{\tmthree}\esub{\var}{\tmtwo}&\mbox{if $\vartwo\notin\fv{\tmtwo}$ \ben{and $\var\notin\fv{\tmthree}$}}&
		(\tm\tmthree)\esub{\var}{\tmtwo} &\tostructap&  \tm\esub{\var}{\tmtwo}\tmthree\esub{\var}{\tmtwo}\\
		\tm\esub{\var}{\tmtwo}\esub{\vartwo}{\tmthree} &\tostructes& \tm\esub{\var}{\tmtwo\esub{\vartwo}{\tmthree}} & \mbox{if $\vartwo\not\in\fv{\tm}$} &
		(\tm\tmthree)\esub\var\tmtwo &\tostructapl& \tm\esub\var\tmtwo\tmthree & \textrm{if }\var\not\in\fv\tmthree 
		 
	\end{array}$
	\end{center}
	\caption{Axioms for structural equivalences. In $\tostructdup$, $\varsplit{\tm}{\var}{\vartwo}$ denotes a term obtained from $\tm$ by renaming some (possibly none) occurrences of $\var$ as $\vartwo$.}
	\label{fig:StructEq}
\end{figure*}

\paragraph{Structural equivalence.} Another common feature of the four calculi is that they come with a notion of \emph{structural equivalence}, denoted by $\eqstruct$. Consider \reffig{StructEq}. For call-by-name and call-by-value calculi, $\eqstruct$ is defined as the smallest equivalence relation containing the closure by weak contexts of $\alphaequiv\cup\tostructgc\cup\tostructdup\cup\tostructap\cup\tostructcom\cup\tostructes$ where $\alphaequiv$ is $\alpha$-equivalence. Call-by-need evaluates inside some substitutions (\cben{but not any substitution}{those hereditarily substituting on the head}) and thus axioms as $\tostructdup$ and $\tostructap$ are too strong. Therefore, the structural equivalence  for call-by-need\ben{, noted $\eqstructneed$,} is the one generated by $\tostructapl\cup\tostructcom\cup\tostructes$. 

 Structural equivalence represents the fact that certain manipulations on explicit substitutions are computationally irrelevant, in the sense that they yield behaviorally equivalent terms. Technically, it is a \emph{strong bisimulation}:
\begin{proposition}[$\tostruct$ is a Strong Bisimulation]
	\label{prop:GenStrongBisim}
	\label{prop:strong-bis}
	Let $\togenm$, $\togene$ and $\eqstruct$ be the reduction relations and the structural equivalence relation of any of the calculi of \reftab{Calculi}, and let $\mathtt{x}\in\set{\mathtt{m},\mathtt{e}}$. Then, $\tm\eqstruct\tmtwo$ and $\tm\togenx\tmp$ implies that there exists $\tmtwop$ such that $\tmtwo\togenx\tmtwop$ and $\tmp\eqstruct\tmtwop$.
\end{proposition}
\begin{proof}
	\withproofs{See \refsect{strong-bisim-proofs} of the appendix (page \pageref{sect:strong-bisim-proofs}).}
		\withoutproofs{See \cite{distillingTR}.}
\end{proof}

The essential property of strong bisimulations is that they can be postponed. In fact, it is immediate to prove the following, which holds for all four calculi:
\begin{lemma}[$\tostruct$ Postponement]
	\label{l:postponement}
	If $\tm\mathrel{(\togenm\cup\togene\cup\eqstruct)^*}\tmtwo$ then $\tm\mathrel{(\togenm\cup\togene)^*\eqstruct}\tmtwo$ and the number of $\togenm$ and $\togene$ steps in the two reduction sequences is exactly the same.
\end{lemma}

In the simulation theorems for machines with a global environment (see \refsect{mam} and \refsect{need}) we will also use the following commutation property between substitutions and evaluation contexts via the structural equivalence of every evaluation scheme, proved by an easy induction on the actual definition of evaluation contexts.

\begin{lemma}[ES Commute with Evaluation Contexts via $\tostruct$] 
	\label{l:ev-comm-struct}
	For every evaluation scheme let $\ctx$ denote an evaluation context s.t. $\var\notin\fv\ctx$ and $\tostruct$ be its structural equivalence. Then $\ctxp{\tm}\esub{\var}{\tmtwo} \eqstruct \ctxp{\tm\esub{\var}{\tmtwo}}$.
\end{lemma}
% !TEX root = main.tex
\section{Preliminaries on Abstract Machines.} 
\label{sect:prel-am}
\emph{Codes.} All the abstract machines we will consider execute pure \lat s. In our syntax, these are nothing but terms \emph{without explicit substitutions}. Moreover, while for calculi we work implicitly modulo $\alpha$, for machines we will \emph{not} consider terms up to $\alpha$, as the handling of $\alpha$-equivalence characterizes different approaches to abstract machines. To stress these facts, we use the metavariables $\code,\codetwo,\codethree,\codefour$ for pure \lat s (not up to $\alpha$) and $\codeval$ for pure values. 

\emph{States}. A machine state $\state$ will have various components, of which the first will always be \emph{the code}, \ie\ a pure $\l$-term $\code$. The others (\emph{environment}, \emph{stack}, \emph{dump}) are all considered as lists, whose constructors are the empty list $\stempty$ and the concatenation operator $\cons$. A state $\state$ of a machine is \emph{initial} if its code $\code$ is closed (\ie, $\fv\code=\emptyset$) and all other components are empty. An \emph{execution} $\exec$ is a sequence of transitions of the machine $\state_0\to^\ast\state$ from an initial state $\state_0$. In that case, we say that $\state$ is a \emph{reachable state}, and if $\code$ is the code of $\state_0$ then $\code$ is the \emph{initial code} of $\state$. 

\emph{Invariants}. For every machine our study will rely on a lemma about some \emph{dynamic invariants}, \ie\  some properties of the reachable states that are stable by executions. The lemma  is always proved by a straightforward induction on the length of the execution and \emph{the proof is omitted}.

\emph{Environments and Closures}. There will be two types of machines, those with many \emph{local environments} and those with just one \emph{global environment}. 
Machines with local environments are based on the mutually recursive definition of \emph{closure} (ranged over by $\clos$) and \emph{environment} ($\env$):
\begin{center}$\begin{array}{rcl@{\sep}rcl@{\sep}rcl}
	\clos &\grameq& (\code,\env)  &
	\env &\grameq& \stempty\mid \esub{\var}{\clos}\cons\env
\end{array}$\end{center}
Global environments are defined by $\genv \grameq \stempty \mid \esub\var\code\cons\genv$, and global environment machines will have just one global closure $(\code,\genv)$.

\emph{Well-Named and Closed Closures}. The explicit treatment of $\alpha$-equivalence, is based on particular representatives of $\alpha$-classes defined via the notion of support. 
The \emph{support} $\Delta$ of codes, environments, and closures is defined by:
\begin{itemize}
	\item $\support\code$ is the \emph{multi}set of its bound names (\eg\ $\support{\l\var.\l \vartwo. \l\var.(\varthree \var)}\\ = [\var,\var,\vartwo]$). 
	\item $\support\env$ is the \emph{multi}set of names captured by $\env$ (for example $\support{\esub\var{\clos_1}\esub\vartwo{\clos_2}\esub\var{\clos_3}}=[\var,\var,\vartwo]$), and similarly for $\support\genv$. 
	\item $\support{\code,\env} \defeq \support\code+ \support\env$ and $\support{\code,\genv} \defeq \support\code+ \support\genv$.
\end{itemize}

A code/environment/closure is \emph{well-named} if its support is a set (\ie\ a multiset with no repetitions). Moreover, a closure $(\code,\env)$ (resp.\ $(\tm,\genv)$) is \emph{closed} if $\fv{\code}\subseteq \support\env$ (resp.\ $\fv{\code}\subseteq \support\genv$).

% !TEX root = main.tex
\section{Distilleries}
\label{sect:theory}
This section presents an abstract, high-level view of the relationship between abstract machines and linear substitution calculi, via the notion of \emph{distillery}. 
%\begin{table}[t]
%	\begin{center}
%		{\setlength{\tabcolsep}{1.5em}
%		\begin{tabular}{|l|l|}
%			\hline
%			\emph{\deff{Calculus}} & \emph{\deff{Abstract Machine}} \\
%			\hline
%			Call-by-Name & KAM\\
%			Call-by-Value (Left-to-Right) & CEK, Split CEK\\
%			Call-by-Value (Right-to-Left) & LAM\\
%			Call-by-Need & WAM, Merged WAM\\
%			\hline
%		\end{tabular}}
%	\end{center}
%	\caption{Correspondence between calculi of \reftab{Calculi} and abstract machines.}
%	\label{tab:Correspondence}
%\end{table}

\begin{definition}
A \emph{distillery} $\distil = (\mach, \calculus, \tostruct,\decode{{ }\cdot{ }})$ is given by:
\begin{enumerate}
	\item An \emph{abstract machine} $\mach$, given by
	\begin{enumerate}
		\item a deterministic labeled transition system $\tomach$ on states $\state$;
		\item a distinguished class of states called \emph{initials} (in bijection with closed $\l$-terms, and from which applying $\tomach$ one obtains the \emph{reachable} states);
		\item \ben{a partition of the labels of the transition system $\tomach$ as:
		\begin{itemize}
			\item \emph{commutative} transitions, noted $\tomacha$;
			\item \emph{principal} transitions, in turn partitioned into
			\begin{itemize}
				\item \emph{multiplicative} transitions, denoted by $\tomachm$;
				\item  \emph{exponential} transitions, denoted by $\tomache$;
			\end{itemize}			
		\end{itemize}
		}
	\end{enumerate}
	
		\item a \emph{linear substitution calculus} $\calculus$ given by a pair $(\togenm, \togene)$ of rewriting relations on terms with ES;

	 \item a \emph{structural equivalence} $\eqstruct$ on terms s.t. it is a strong bisimulation with respect to $\togenm$ and $\togene$;

	\item a \emph{distillation} $\decode{{ }\cdot{ }}$, \ie\ a decoding function from states to terms, s.t. on reachable states:
	\begin{itemize}
		\item \emph{Commutative}: $\state\tomacha\statetwo$ implies $\decode\state\eqstruct\decode\statetwo$.
		\item \emph{Multiplicative}: $\state\tomachm\statetwo$ implies $\decode\state\togenm\eqstruct\decode\statetwo$;
		\item \emph{Exponential}: $\state\tomache\statetwo$ implies $\decode\state\togene\eqstruct\decode\statetwo$;

	\end{itemize}	
\end{enumerate}
\end{definition}

Given a distillery, the simulation theorem holds abstractly. Let $\size\exec$ (resp. $\size\deriv$), $\sizem\exec$ (resp. $\sizem\deriv$), $\sizee\exec$ (resp. $\sizee\deriv$), and $\sizep\exec$ denote the number of unspecified, multiplicative, exponential, and principal steps in an execution (resp. derivation).

\begin{theorem}[Simulation]
	\label{tm:GenSim}
	Let $\distil$ be a distillery. Then for every execution $\exec:\state\tomach^*\statetwo$ there is a derivation $\deriv:\decode\state\togen^*\tostruct\decode\statetwo$ s.t. $\sizem\exec=\sizem\deriv$, $\sizee\exec=\sizee\deriv$, and $\sizep\exec=\size\deriv$.
\end{theorem}

\begin{proof}
By induction on $\size\exec$ and by the properties of the decoding, it follows that there is a derivation $\derivtwo:\decode\state(\togen\tostruct)^*\decode\statetwo$  s.t. the number $\sizep\exec=\size\derivtwo$. The witness $\deriv$ for the statement is obtained by applying the postponement of strong bisimulations (\reflemma{postponement}) to $\derivtwo$.
\end{proof}

\paragraph{Reflection.} Given a distillery, one would also expect that reduction in the calculus is reflected in the machine. This result in fact requires two additional abstract properties.

\begin{definition}[Reflective Distillery]
A distillery is \emph{reflective} when:
\begin{description}
	\item[Termination:] 
	$\tomacha$ \emph{terminates} (on reachable states); hence, by determinism, every state $\state$ has a unique \emph{commutative normal form} $\admnf\state$;
	\item[Progress:] if $\state$ is reachable, $\admnf{\state}=\state$ and $\decode\state\togen_x\tm$ with $\mathtt{x}\in\set{\mathtt{m},\mathtt{e}}$, then there exists $\statetwo$  such that $\state\tomachx\statetwo$, \ie, $\state$ is not final.
\end{description}
\end{definition}

Then, we may prove the following reflection of steps in full generality:
\begin{proposition}[Reflection]
	\label{prop:GenCorrectness}
	Let $\distil$ be a reflective distillery, $\state$ be a reachable state, and $x\in\set{\mulsym,\expsym}$. Then, $\decode\state\togenx\tmtwo$ implies that there exists a state $\statetwo$ s.t. $\admnf\state\tomachx\statetwo$ and $\decode\statetwo\eqstruct\tmtwo$.
\end{proposition}

In other words, every rewriting step on the calculus can be also performed on the machine, up to commutative transitions. 
\proof
	The proof is by induction on the number $n$ of transitions leading from $\state$ to $\admnf{\state}$.
	\begin{itemize}
		\item \emph{Base case} $n=0$: by the progress property, we have $\state\Rew{\mathtt x'}\statetwo$ for some state $\statetwo$ and $\mathtt x'\in\set{\mathtt{m},\mathtt{e}}$. By \refth{GenSim}, we have $\decode\state\multimap_{\mathtt x'}\tmtwo'\eqstruct\decode{\state'}$ and we may conclude because $\mathtt x'=\mathtt x$ and $\tmtwo'=\tmtwo$ by determinisim of the calculus (\refprop{GenDet}).
		\item \emph{Inductive case} $n>0$: by hypothesis, we have $\state\tomacha\state_1$. By \refth{GenSim}, $\decode\state\eqstruct\decode{\state_1}$. The hypothesis and the strong bisimulation property (\refprop{GenStrongBisim}) then give us $\decode{\state_1}\togenx\tmtwo_1\eqstruct\tmtwo$. But the induction hypothesis holds for $\state_1$, giving us a state $\state'$ such that $\admnf{\state_1}\tomachx\state'$ and $\decode{\state'}\eqstruct\tmtwo_1\eqstruct\tmtwo$. We may now conclude because $\admnf\state=\admnf{\state_1}$.\qed
	\end{itemize}

The reflection can then be extended to a reverse simulation.

\begin{corollary}[Reverse Simulation]
Let $\distil$ be a reflective distillery and $\state$ an initial state. Given a derivation $\deriv:\decode\state\togen^*\tm$ there is an execution $\exec:\state\tomach^*\statetwo$ s.t. $\tm\tostruct\decode\statetwo$ and $\sizem\exec=\sizem\deriv$, $\sizee\exec=\sizee\deriv$, and $\sizep\exec=\size\deriv$.
\end{corollary}

\begin{proof}
By induction on the length of $\deriv$, using \refprop{GenCorrectness}.
\end{proof}

In the following sections we shall introduce abstract machines and distillations for which we will prove that they form reflective distilleries with respect to the calculi of \refsect{ES-distance}. For each machine we will prove 1) that the decoding is in fact a distillation, and 2) the progress property.  \emph{We will instead assume the termination property}, whose proof is delayed to the quantitative study of the second part of the paper, where we will actually prove stronger results, giving explicit bounds.
% !TEX root = main.tex
\section{Call-by-Name: the KAM}
\label{s:kam}
\renewcommand{\evctx}{\whctx}
\renewcommand{\evctxtwo}{\whctxtwo}
\renewcommand{\evctxp}[1]{\whctxp{#1}}
\renewcommand{\evctxtwop}[1]{\whctxtwop{#1}}

The Krivine Abstract Machine (KAM) is the simplest machine studied in the paper. A KAM \emph{state} ($\state$) is made out of a closure and of a \emph{stack} ($\stack$):
\begin{center}$\begin{array}{rcl@{\sep\sep\sep}rcl@{\sep}rcl}
	\stack &\grameq& \stempty\mid \clos\cons\stack &
	\state  &\grameq& (\clos,\stack)
\end{array}$\end{center}
For readability, we will use the notation $\kamstate{\code}{\env}{\stack}$ for a state $(\clos,\stack)$ where $\clos = (\code,\env)$. The transitions of the KAM then are:
\begin{center}${\setlength{\arraycolsep}{0.9em}
\begin{array}{c|c|ccc|c|c}
	\code\codetwo&\env&\stack
	&\tomacha&
	\code&\env&(\codetwo,\env)\cons\stack
	\\
	\l\var.\code&\env&\clos\cons\stack
	&\tomachm&
	\code&\esub\var\clos\cons\env&\stack
	\\
	\var&\env&\stack
	&\tomache&
	\code&\envtwo&\stack
\end{array}}$\end{center}
where $\tomache$ takes place only if $\env=\envthree\cons\esub{\var}{(\code,\envtwo)}\cons\envfour$.

A key point of our study is that environments and stacks rather immediately become contexts of the \lsc, through the following decoding:
\begin{center}
$\begin{array}{rcl@{\sep\sep}rcl}
	\decode{\stempty} & \defeq & \ctxhole&
	\decode{\esub{\var}{\clos}\cons\env} & \defeq & \decode{\env}\ctxholep{\ctxhole\esub{\var}{\decode{\clos}}}\\
	\decode{(\code,\env)} & \defeq & \decode{\env}\ctxholep{\code} &
	\decode{\clos\cons\stack} & \defeq & \decode{\stack}\ctxholep{\ctxhole\decode{\clos}}\\
	\decode{\kamstate{\code}{\env}{\stack}} &\defeq &\decstackp{\decenvp{\code}}
\end{array}$
\end{center}

The decoding satisfies the following static properties, shown by easy inductions on the definition.

\begin{lemma}[Contextual Decoding]
	\label{l:kam-ctx-dec}
\ben{	Let $\env$ be an environment and $\pi$ be a stack of the KAM. Then }$\decode{\env}$ is a substitution context, and both $\decode{\stack}$ and $\decode{\stack}\ctxholep{\decode{\env}}$ are evaluation contexts.
\end{lemma}

Next, we need the dynamic invariants of the machine.

\begin{lemma}[KAM Invariants]
	\label{l:kam-prop}
	Let $\state=\kamstate{\codetwo}\env\stack$ \invariantshyp{KAM}. Then:
	\begin{enumerate}
		\item \label{p:kam-prop-one}\emph{Closure}: \closprop;
		\item \label{p:kam-prop-two}\emph{Subterm}: \subprop.
		\item \label{p:kam-prop-three}\emph{Name}: \lwnameprop.
		\item \label{p:kam-prop-four}\emph{Environment Size}: \envprop.
	\end{enumerate}
\end{lemma}

\emph{Abstract Considerations on Concrete Implementations}. The name invariant is the abstract property that allows to avoid $\alpha$-equivalence in KAM executions. In addition, forbidding repetitions in the support of an environment, it allows to bound the length of any environment with the names in $\code$, \ie\ with $\size\code$. This fact is important, as the static bound on the size of environments guarantees that $\tomache$ and $\tomacha$---the transitions looking-up and copying environments---can be implemented (independently of the chosen concrete representation of terms) in at worst linear time in $\size\tm$, so that an execution $\exec$ can be implemented in $O(\size\exec\cdot\size\tm)$. The same will hold for every machine with local environments.

The previous considerations are based on the name and environment size invariants. The closure invariant is used in the progress part of the next theorem, and the subterm invariant is used in the quantitative analysis in \refsect{complexity} (\refth{value-name-glob-bilin}), subsuming the termination condition of reflective distilleries.

\begin{theorem}[KAM Distillation]
	\label{tm:kam-sim}
	$(\mbox{KAM},\wlscname,\tostruct,\decode{{ }\cdot{ }})$ \distillationStatement:
	\begin{enumerate}
		\item \emph{Commutative}: if $\state\tomacha\statetwo$ then $\decode{\state}\eqstruct\decode{\statetwo}$.
		\item \emph{Multiplicative}: if $\state\tomachm\statetwo$ then $\decode{\state}\towhldb\decode{\statetwo}$;
		\item \emph{Exponential}: if $\state\tomache\statetwo$ then $\decode{\state}\towhlls\eqstruct\decode{\statetwo}$;
	\end{enumerate}
\end{theorem}

%%%%% PROOF %%%%%%
 % !TEX root = ../main.tex
\proof
\emph{Properties of the decoding}:
\begin{enumerate}
\item	\emph{Commutative}. We have $\kamstate{\code\codetwo}{\env}{\stack}\ \tomacha\ \kamstate{\code}{\env}{(\codetwo,\env)\cons\stack}$, and:
	\begin{center}
		$\begin{array}{ccccccccccc}
			\decode{\kamstate{\code\codetwo}{\env}{\stack}}
			&=&
			\decstackp{\decenvp{\code\codetwo}}\\
			 &\tostructap^*& \decstackp{\decenvp{\code}\decenvp{\codetwo}}
			 &=&\decode{\kamstate{\code}{\env}{(\codetwo,\env)\cons\stack}}
		\end{array}$
	\end{center}
	
\item \emph{Multiplicative}.  $\kamstate{\l\var.\code}{\env}{\clos\cons\stack}\ \tomachm\
\kamstate{\code}{\esub\var\clos\cons\env}{\stack}$, and
\begin{center}$\begin{array}{rcl}
		\decode{\kamstate{\l\var.\code}{\env}{\clos\cons\stack}}&=&
		\decstackp{\decenvp{\l\var.\code}\decclos} \\
		&\towhldb&\decstackp{\decenvp{\code \esub{\var}{\decclos}}}\\
		&=&\decode{\kamstate{\code}{\esub\var\clos\cons\env}{\stack}}
	\end{array}$\end{center}
The rewriting step can be applied because by contextual decoding (\reflemma{kam-ctx-dec}) it takes place in an evaluation context.
	
\item	\emph{Exponential}. $\kamstate{\var}{\envtwo\cons\esub{\var}{(\code,\env)}\cons\envthree}{\stack}\ \tomache\ \kamstate{\code}{\env}{\stack}$, and
	\begin{center}${\setlength{\arraycolsep}{3pt}\begin{array}{rcl}
		\decode{\kamstate{\var}{\envtwo\cons\esub{\var}{(\code,\env)}\cons\envthree}{\stack}}
		&=&
		\decstackp{\decenvthreep{\decenvtwop{\var}\esub{\var}{\decenvp{\code}}}}\\
		 &\towhlls&\decstackp{\decenvthreep{\decenvtwop{\decenvp{\code}}\esub{\var}{\decenvp{\code}}}}\\
		 &\tostructgc^*& \decstackp{\decenvp{\code}}\\
		&=&\decode{\kamstate{\code}{\env}{\stack}}
	\end{array}}$\end{center}
	Note that $\envthreep{\envtwop{\envp{\code}}\esub{\var}{\envp{\code}}}\tostructgc^*\envp{\code}$ holds because $\envp{\code}$ is closed by point \ref{p:kam-prop-one} of \reflemma{kam-prop}, and so all the substitutions around it can be garbage collected.

	\end{enumerate}
	
\noindent    \emph{Termination}. Given by (forthcoming) \refth{value-name-glob-bilin} (future proofs of distillery theorems will omit termination). 
    
\noindent \emph{Progress}. Let $\state=\kamstate\code\env\stack$ be a commutative normal form s.t. $\decode\state\togen\tmtwo$. If $\code$ is 
	\begin{itemize}
	    \item \emph{an application $\codetwo\codethree$}. Then a $\tomacha$ transition applies and $\state$ is not a commutative normal form, absurd.
		\item \emph{an abstraction} $\l\var.\codetwo$: if $\stack=\stempty$ then $\decode\state=\decenvp{\l\var.\codetwo}$, which is $\togen$-normal, absurd. Hence, a $\tomachm$ transition applies.
				\item \emph{a variable} $\var$: by point 1 of \reflemma{kam-prop}.\ref{p:kam-prop-one}, we must have $\env=\env'\cons\esub\var\clos\cons\env''$, so a $\tomache$ transition applies;
\qed
	\end{itemize}
%%%%%%%%%%%%%%%%

% !TEX root = main.tex
\section{Call-by-Value: the CEK  and the LAM}
\renewcommand{\evctx}{\cbvctx}
\renewcommand{\evctxtwo}{\cbvctxtwo}
\renewcommand{\evctxp}[1]{\cbvctxp{#1}}
\renewcommand{\evctxtwop}[1]{\cbvctxtwop{#1}}
\ben{Here we deal with two adaptations to call-by-value of the KAM, namely Felleisen and Friedman's CEK machine \cite{Felleisen:1986:CEK} (without control operators), and a variant, deemed \emph{Leroy abstract machine} (LAM). They differ on how they behave with respect to applications: the CEK implements left- to-right call-by-value, \ie\ it first evaluates the function part, the LAM gives instead precedence to arguments, realizing right-to-left call-by-value. The LAM owes its name to Leroy's ZINC machine \cite{Leroy-ZINC}, that implements right-to-left call-by-value evaluation. We introduce a new name because the ZINC is a quite more sophisticated machine than the LAM: it has a separate sets of instructions to which terms are compiled, it handles arithmetic expressions, and it avoids needless closure creations in a way that it is not captured by the LAM.  

The states of the CEK and the LAM} have the same shape of those of the KAM, \ie\ they are given by a closure plus a stack. The difference is that they use \emph{call-by-value stacks}, whose elements are labelled either as \emph{functions} or \emph{arguments}, so that the machine may know whether it is launching the evaluation of an argument or it is at the end of such an evaluation. They are re-defined and decoded by ($\clos$ is a closure):
\begin{center}
$\begin{array}{rcl@{\sep\sep}rclllll}
	\stack & \grameq & \stempty \mid \fnst{\clos}\cons\stack \mid \argst{\clos}\cons\stack&
	\decode{\stempty} & \defeq & \ctxhole&\\
&&&	\decode{\fnst{\clos}\cons\stack} & \defeq & \decode{\stack}\ctxholep{\decode\clos\ctxhole}\\
&&&	\decode{\argst{\clos}\cons\stack} & \defeq & \decode{\stack}\ctxholep{\ctxhole\decode\clos}
\end{array}$
\end{center}
The states of both machines are decoded exactly as for the KAM, \ie\ $\decode{\lamstate{\code}{\env}{\stack}}\defeq\decstackp{\decenvp{\code}}$. 

\subsection{Left-to Right Call-by-Value: the CEK machine.}
The transitions of the CEK are:
\begin{center}${\setlength{\arraycolsep}{0.4em}
\begin{array}{c|c|ccc|c|c}
	\code\codetwo&\env&\stack
	&\tomachaone&
	\code&\env&\argst{\codetwo,\env}\cons\stack
	\\
	\codeval&\env&\argst{\codetwo,\envtwo}\cons\stack
	&\tomachatwo&
	\codetwo&\envtwo&\fnst{\codeval,\env}\cons\stack
	\\
	\codeval&\env&\fnst{\l\var.\code,\envtwo}\cons\stack
	&\tomachm&
	\code&\esub{\var}{(\codeval,\env)}\cons\envtwo&\stack
	\\
	\var&\env&\stack
	&\tomache&
	\code&\envtwo&\stack
\end{array}}$\end{center}
where $\tomache$ takes place only if $\env=\envthree\cons\esub{\var}{(\code,\envtwo)}\cons\envfour$.

While one can still statically prove that environments decode to substitution contexts, to prove that $\decstack$ and $\decstackp\decenv$ are evaluation contexts we need the dynamic invariants of the machine.

\begin{lemma}[CEK Invariants]
	\label{l:cek-prop} % \reflemma{cek-prop}.\refp{cek-prop-four}
	Let $\state=\cekstate\codetwo\env\stack$ \invariantshyp{CEK}. Then:
	\begin{enumerate}
		\item \label{p:cek-prop-one}\emph{Closure}: \closprop;
		\item \label{p:cek-prop-two}\emph{Subterm}: \subprop;
		\item \label{p:cek-prop-three}\emph{Value}: any code in $\env$ is a value and, for every element of $\stack$ of the form $\fnst{\codetwo,\envtwo}$, $\codetwo$ is a value;
		\item \label{p:cek-prop-four}\emph{Contextual Decoding}: $\decstack$ and $\decstackp\decenv$ are left-to-right call-by-value evaluation contexts.
		\item \label{p:cek-prop-five}\emph{Name}: \lwnameprop.
		\item \label{p:cek-prop-six}\emph{Environment Size}: \envprop.

	\end{enumerate}
\end{lemma}

% \subsubsection{The decoding.}
We have everything we need:

\begin{theorem}[CEK Distillation]
    \label{tm:cek-simulation}
		$(\mbox{CEK},\wlscvaluelr,\tostruct,\decode{{ }\cdot{ }})$ \distillationStatement:
	\begin{enumerate}
		\item \emph{Commutative 1}: if $\state\tomachaone\statetwo$ then $\decode{\state}\eqstruct\decode{\statetwo}$;
		\item \emph{Commutative 2}: if $\state\tomachatwo\statetwo$ then $\decode{\state}=\decode{\statetwo}$.
		\item \emph{Multiplicative}: if $\state\tomachm\statetwo$ then $\decode{\state}\towhlcekdb\decode{\statetwo}$;
		\item \emph{Exponential}: if $\state\tomache\statetwo$ then $\decode{\state}\towhlcekls\eqstruct\decode{\statetwo}$;
	\end{enumerate}
\end{theorem}

%%% PROOF %%%
    % !TEX root = ../main.tex
\proof
\emph{Properties of the decoding}:
in the following cases, evaluation will always takes place under a context that
by \reflemma{cek-prop}.\ref{p:cek-prop-four} will be a
left-to-right call-by-value evaluation context, and similarly structural equivalence will alway be used in a weak context, as it should be.

\begin{enumerate}
\item \caselight{Commutative 1}. We have $\cekstate{\code\codetwo}{\env}{\stack}\ \tomachaone\ \cekstate{\code}{\env}{\argst{\codetwo,\env}\cons\stack}$, and:
	\begin{center}
	$\begin{array}{rclclccc}
		\decode{\cekstate{\code\codetwo}{\env}{\stack}}
		&=&	\decstackp{\decenvp{\code\codetwo}} &\tostructap^*&\\
&&\normalsize\decstackp{\decenvp{\code}\decenvp{\codetwo}}
&=&\decode{\cekstate{\code}{\env}{\argst{\codetwo,\env}\cons\stack}}
	\end{array}$
	\end{center}

\item \caselight{Commutative 2}. We have 
$\cekstate{\codeval}{\env}{\argst{\codetwo,\envtwo}\cons\stack}
\ \tomachatwo\ 
\cekstate{\codetwo}{\envtwo}{\fnst{\codeval,\env}\cons\stack}$, and:
	\begin{center}
	$\begin{array}{rclccccc}
		\decode{\cekstate{\codeval}{\env}{\argst{\codetwo,\envtwo}\cons\stack}}
		&=&	\decstackp{\decenvp{\codeval}\decenvtwop{\codetwo}} 
		&=&\\
		&&\decode{\cekstate{\codetwo}{\envtwo}{\fnst{\codeval,\env}\cons\stack}}
	\end{array}$
	\end{center}

\item \caselight{Multiplicative}. We have 
$\cekstate{\codeval}{\env}{\fnst{\l\var.\code,\envtwo}\cons\stack}
\ \tomachm\ 
\cekstate{\codetwo}{\esub{\var}{(\codeval,\env)}\cons\envtwo}{\stack}$, and:
	\begin{center}
	$\begin{array}{rclccccc}
		\decode{\cekstate{\codeval}{\env}{\fnst{\l\var.\code,\envtwo}\cons\stack}}
		&=&	\decstackp{\decenvtwop{\l\var.\code}\decenvp{\codeval}} &\towhlcekdb\\
		&&\decstackp{ \decenvtwop{\code\esub{\var}{\decenvp{\codeval}}} }
		&=&\\
		&&\decode{\cekstate{\code}{\esub{\var}{(\codeval,\env)}\cons\envtwo}{\stack}}
	\end{array}$
	\end{center}

\item \caselight{Exponential}.
Let $\env=\envthree\cons\esub{\var}{(\code,\envtwo)}\cons\envfour$. We have 
$\cekstate{\var}{\env}{\stack}
\ \tomache\
\cekstate{\code}{\envtwo}{\stack}$, and:
	\begin{center}
	$\begin{array}{rclclccccc}
		\decode{\cekstate{\var}{\env}{\stack}}
		&=&	\stackp{\envp{\var}} &=\\
		&&\decstackp{\decenvfourp{\decenvthreep{\var}\esub{\var}{\decenvtwop{\code}}}}&\towhlcekls\\
		&&\decstackp{\decenvfourp{\decenvtwop{\decenvthreep{\code}\esub{\var}{\code}}}}
		&\tostructgc^*&\\
		&&\decstackp{\decenvtwop{\code}}&=&\decode{\cekstate{\code}{\envtwo}{\stack}}
	\end{array}$
	\end{center}
We can apply $\towhlcekls$ since by \reflemma{cek-prop}.\ref{p:cek-prop-three},
$\code$ is a value.
We also use that by \reflemma{cek-prop}.\ref{p:cek-prop-one},
$\decenvtwop{\code}$ is a closed term to ensure that $\decenvthree$ and $\decenvfour$
can be garbage collected.
\end{enumerate}

\noindent \emph{Progress}.	 Let $\state=\kamstate\code\env\stack$ be a commutative normal form s.t. $\decode\state\togen\tmtwo$. If $\code$ is 
	\begin{itemize}
		    \item \emph{an application $\codetwo\codethree$}. Then a $\tomachaone$ transition applies and $\state$ is not a commutative normal form, absurd.
		\item \emph{an abstraction} $\codeval$: by hypothesis, $\stack$ cannot be of the form $\argst{\clos}\cons\stacktwo$. Suppose it is equal to $\stempty$. We would then have $\decode{\state}=\decenvp{\codeval}$, which is a call-by-value normal form, because $\decenv$ is a substitution context. This would contradict our hypothesis, so $\stack$ must be of the form $\fnst{\codetwo,\envtwo}\cons\stacktwo$. By point \ref{p:cek-prop-three} of \reflemma{cek-prop}, $\codetwo$ is an abstraction, hence a $\tomachm$ transition applies.
		\item \emph{a variable} $\var$: by point \ref{p:cek-prop-one} of \reflemma{cek-prop}, $\env$ must be of the form $\envtwo\cons\esub{\var}{\clos}\cons\envthree$, so a $\tomache$ transition applies;
\qed
	\end{itemize}
%%%%%%%%%%
% !TEX root = main.tex
\subsection{Right-to-Left Call-by-Value: the Leroy Abstract Machine}
\renewcommand{\evctx}{\scbvctx}
\renewcommand{\evctxtwo}{\scbvctxtwo}
\renewcommand{\evctxp}[1]{\scbvctxp{#1}}
\renewcommand{\evctxtwop}[1]{\scbvctxtwop{#1}}

The transitions of the LAM are:
\begin{center}
${\setlength{\arraycolsep}{0.65em}
\begin{array}{c|c|ccc|c|c}
    \code\codetwo&\env&\stack
    &\tomachaone&
    \codetwo&\env&\fnst{\code,\env}\cons\stack
    \\
    \codeval&\env&\fnst{\code,\envtwo}\cons\stack
    &\tomachatwo&
    \code&\envtwo&\argst{\codeval,\env}\cons\stack
    \\
    \l\var.\code&\env&\argst{\clos}\cons\stack
    &\tomachm&
    \code&\esub{\var}{\clos}\cons\env&\stack
    \\
    \var&\env&\stack
    &\tomache&
    \code&\envtwo&\stack
    \\
\end{array}
}$\end{center}
where $\tomache$ takes place only if $\env=\envthree\cons\esub{\var}{(\code,\envtwo)}\cons\envfour$.

We omit all the proofs (that can be found \withproofs{in the appendix, page \pageref{ss:lam-proofs}}\withoutproofs{in \cite{distillingTR}}) because they are minimal variations on those for the CEK. 

\begin{lemma}[LAM Invariants]
    \label{l:lam-prop}
    Let $\state=\lamstate\codetwo\env\stack$ \invariantshyp{LAM}. Then:
    \begin{enumerate}
        \item \label{p:lam-prop-one}\emph{Closure}: \closprop;
        \item \label{p:lam-prop-two}\emph{Subterm}: \subprop;
        \item \label{p:lam-prop-three}\emph{Value}: any code in $\env$ is a value and, for every element of $\stack$ of the form $\argst{\codetwo,\envtwo}$, $\codetwo$ is a value;
        \item \label{p:lam-prop-four}\emph{Contexts Decoding}: $\decstack$ and $\decstackp\decenv$ are right-to-left call-by-value evaluation contexts.
    		\item \label{p:lam-prop-five}\emph{Name}: \lwnameprop.
		\item \label{p:lam-prop-six}\emph{Environment Size}: \envprop.

    \end{enumerate}
\end{lemma}

\begin{theorem}[LAM Distillation]
    \label{tm:lam-simulation}
	$(\mbox{LAM},\wlscvaluerl,\tostruct,\decode{{ }\cdot{ }})$ \distillationStatement:
    \begin{enumerate}
        \item \emph{Commutative 1}: if $\state\tomachaone\statetwo$ then $\decode{\state}\eqstruct\decode{\statetwo}$;
        \item \emph{Commutative 2}: if $\state\tomachatwo\statetwo$ then $\decode{\state}=\decode{\statetwo}$.
        \item \emph{Multiplicative}: if $\state\tomachm\statetwo$ then $\decode{\state}\towhllamdb\decode{\statetwo}$;
        \item \emph{Exponential}: if $\state\tomache\statetwo$ then $\decode{\state}\towhllamls\eqstruct\decode{\statetwo}$;
    \end{enumerate}
\end{theorem}
% !TEX root = main.tex
\section{Towards Call-by-Need: the MAM and the Split CEK }
In this section we study two further machines:
\begin{enumerate}
	\item \emph{The Milner Abstract Machine (MAM)}, that is a variation over the KAM with only one global environment and without the concept of closure. Essentially, it unveils the content of distance rules at the machine level. 
	\item \emph{The Split CEK (SCEK)}, obtained disentangling the two uses of the stack (for arguments and for functions) in the CEK. The split CEK can be seen as a simplification of Landin's SECD machine \cite{LandinSECD}.
\end{enumerate}
The ideas at work in these two case studies will be combined in the next section, obtaining a new simple call-by-need machine.

\begin{figure*}
\begin{center}${\setlength{\arraycolsep}{1em}
\begin{array}{c|c|c|ccc|c|c|c}
	\code\codetwo&\env&\stack&\fstack
	&\tomachaone&
	\code&\env&(\codetwo,\env)\cons\stack&\fstack 
	\\	
	\codeval&\env&(\code,\envtwo)\cons\stack&\fstack
	&\tomachatwo&
	\code&\envtwo&\stempty&((\codeval,\env),\stack)\cons\fstack
	\\
	\codeval&\env&\stempty&((\l\var.\code,\envtwo),\stack)\cons\fstack
	&\tomachm&
	\code&\esub\var{(\codeval,\env)}\cons\envtwo&\stack&\fstack
	\\
	\var&\env\cons\esub{\var}{(\codeval,\envtwo)}\cons\envthree&\stack&\fstack
	&\tomache&
\codeval&\envtwo&\stack &\fstack
\end{array}}$\end{center}
\caption{The Split CEK, aka the revisited SECD\label{fig:secd}.}
\end{figure*}

\subsection{Milner Abstract Machine}
\label{sect:mam}
\renewcommand{\evctx}{\whctx}
\renewcommand{\evctxtwo}{\whctxtwo}
\renewcommand{\evctxp}[1]{\whctxp{#1}}
\renewcommand{\evctxtwop}[1]{\whctxtwop{#1}}

The linear substitution calculus suggests the design of a simpler version of the KAM, the Milner Abstract Machine (MAM), that avoids the concept of closure. At the language level, the idea is that, by repeatedly applying the axioms $\tostructdup$ and $\tostructap$ of the structural equivalence, explicit substitutions can be folded and brought \emph{outside}. At the machine level, the local environments in the closures are replaced by just one global environment that closes the code and the stack, as well as the global environment itself. 

Of course, naively turning to a global environment breaks the well-named invariant of the machine. This point is addressed using an $\alpha$-renaming in the variable transition, \ie\ when substitution takes place. Here we employ the global environments $\genv$ of \refsect{prel-am} and we redefine stacks as $\stack \grameq \stempty\mid \code\cons\stack$. A state of the MAM is given by a code $\code$, a stack $\stack$ and a global environment $\genv$. Note that the code and the stack together now form a code. 

The transitions of the MAM are:
\begin{center}
$\begin{array}{c|c|ccc|c|cccccccccc}
\code\codetwo&\stack&\genv&\tomacha&
\code&\codetwo\cons \stack&\genv
\\
\l\var.\code&\codetwo\cons \stack&\genv&\tomachm&
\code&\stack&\esub{\var}{\codetwo} \cons \genv
\\
\var&\stack&\genv&\tomache&
\rename{\code}&\stack&\genv
\end{array}$
\end{center}
where $\tomache$ takes place only if $\genv=\genvthreep{\genvtwo\esub{\var}{\code}}$ and $\rename{\code}$ is a well-named code $\alpha$-equivalent to $\code$ and s.t. any bound name in $\rename{\code}$ is fresh with respect to those in $\stack$ and $\genv$\footnote{The well-named invariant can be restored also in another way. One can simply substitute $\code$ (instead of $\rename{\code}$) but modify $\tomachm$ as follows (with $\vartwo$ fresh):
$$\begin{array}{c|c|ccc|c|cccccccccc}
\l\var.\code&\codetwo\cons \stack&\genv&\tomachm&
\code\isub\var\vartwo&\stack&\esub{\vartwo}{\codetwo} \cons \genv
\end{array}$$}.

The decoding of a MAM state $\mamstate{\code}{\stack}{\genv}$ is similar to the decoding of a KAM state, but the stack and the environment context are applied in reverse order (this is why stack and environment in MAM states are swapped with respect to KAM states):
\begin{center}$\begin{array}{rcl@{\sep\sep}rcl}
	\decode{\stempty}&\defeq&\ctxhole&
	\decode{\esub{\var}{\code}\cons\genv}  &\defeq&\decode{\genv}\ctxholep{\ctxhole\esub{\var}{\code}}\\
	\decode{\code\cons\stack}  & \defeq & \decode{\stack}\ctxholep{\ctxhole\code}&
	\decode{\mamstate{\code}{\stack}{\genv}} & \defeq & \decode{\genv}\ctxholep{\decode{\stack}\ctxholep{\code}}
\end{array}$\end{center}
To every MAM state $\mamstate{\code}{\stack}{\genv}$ we associate the pair $(\code\decstack,\genv)$ and call it the \emph{global closure} of the state.

As for the KAM, the decoding of contexts can be done statically, \ie\ it does not need dynamic invariants.

\begin{lemma}[Contextual Decoding]
\ben{Let $\genv$ be a global environment and $\pi$ be a stack of the MAM. Then }	$\decode{\genv}$ is a substitution context, and both $\decode{\stack}$ and $\decode{\stack}\ctxholep{\decode{\genv}}$ are evaluation contexts.
\end{lemma}

For the dynamic invariants we need a different notion of closed closure.

  \begin{definition}
  Given a global environment $\genv$ and a code $\code$, we define by mutual induction
  two predicates \emph{$\pwamclosed{\genv}$} and \emph{$\pwamclosed{(\code,\genv)}$} as follows:
  $$
  \begin{array}{rll}
  && \pwamclosed{\stempty} \\
  \pwamclosed{(\code,\genv)} & \implies & \pwamclosed{\esub{\var}{\tm}\cons\genv}\\
%  \pwamclosed{\env}         & \implies & \pwamclosed{\esub{\var}{\pwammark}\cons\env}\\
  \text{$\fv{\code}\subseteq\support\genv$}
  \land
  \pwamclosed{\genv}   & \implies & \pwamclosed{(\code,\genv)} \\
  \end{array}
  $$
  \end{definition}

The dynamic invariants are:

\begin{lemma}[MAM invariants]
\label{l:mam-prop}
Let $\state = \mamstate{\codetwo}{\stack}{\genv}$ \globinvariantshyp{MAM}. Then:
\begin{enumerate}
\item \label{p:mam-prop-one} {\em Global Closure:} \globclosprop{$(\code\decstack,\genv)$};
\item \label{p:mam-prop-two} {\em Subterm:} \subprop;
\item \label{p:mam-prop-three} {\em Names:} \gwnameprop;
\item \label{p:mam-prop-env} {\em Environment Size:} \genvprop.
\end{enumerate}
\end{lemma}
\emph{Abstract Considerations on Concrete Implementations}. Note the new environment size invariant, whose bound is laxer than for local environment machines. Let $\exec$ be a execution of initial code $\code$. If one implements $\tomache$ looking for $\var$ in $\genv$ sequentially, then each $\tomache$ transition has cost $\sizem\exec$ (more precisely, linear in the number of preceding $\tomachm$ transitions) and the cost of implementing $\exec$ is easily seen to become quadratic in $\size\exec$. An efficient implementation would then employ a representation of codes such that variables are pointers, so that looking for $\var$ in $\genv$ takes constant time. The name invariant guarantees that variables can indeed taken as pointers, as there is no name clash. Note that the cost of a $\tomache$ transition is not constant, as the renaming operation actually makes $\tomache$ linear in $\size\tm$ (by the subterm invariant). So, assuming a pointer-based representation, $\exec$ can be implemented in time $O(\size\exec\cdot\size\code)$, as for local machines, and the same will hold for every global environment machine.

\begin{theorem}[MAM Distillation]
\label{tm:mam-sim}
	$(\mbox{MAM},\wlscname,\tostruct,\decode{{ }\cdot{ }})$ \distillationStatement:
	\begin{enumerate}
		\item \emph{Commutative}: if $\state\tomacha\statetwo$ then $\decode{\state}=\decode{\statetwo}$;
		\item \emph{Multiplicative}: if $\state\tomachm\statetwo$ then $\decode{\state}\towhldb\tostruct\decode{\statetwo}$;
		\item \emph{Exponential}: if $\state\tomache\statetwo$ then $\decode{\state}\towhlls\alphaequiv\decode{\statetwo}$.
	\end{enumerate}
\end{theorem}

\proof
\emph{Properties of the decoding} (\emph{progress} is as for the KAM):
\begin{enumerate}
\item \emph{Commutative}. In contrast to the KAM, $\tomacha$ gives a true identity:
\begin{center}
$\begin{array}{lllll}
\decode{\mamstate{\code\codetwo}{\stack}{\genv}}
&=&
\decgenvp{\decstackp{\code\codetwo}} &=&\decode{\mamstate{\code}{\codetwo\cons\stack}{\genv}}
\end{array}$
\end{center}
\item \emph{Multiplicative}. Since substitutions and evaluation contexts commute via $\tostruct$ (\reflemma{ev-comm-struct}), $\tomachm$ maps to:
\begin{center}
$\begin{array}{rcllll}
\decode{\mamstate{\l\var.\code}{\codetwo\cons\stack}{\genv}}&=&
\decgenvp{\decstackp{(\l\var.\code)\codetwo}} &\towhldb\\
&& \decgenvp{\decstackp{\code \esub{\var}{\codetwo}}}&\tostruct_{Lem. \ref{l:ev-comm-struct}}\\
&&\decgenvp{\decstackp{\code}\esub{\var}{\codetwo}}&=&\\
&&\decode{\mamstate{\code}{\stack}{\esub{\var}{\codetwo}\cons\genv}}
\end{array}$
\end{center}

\item \emph{Exponential}. The erasure of part of the environment of the KAM is replaced by an explicit use of $\alpha$-equivalence:
\begin{center}
$\begin{array}{lllll}
\decode{\mamstate{\var}{\stack}\genv\cons\esub{\var}{\codetwo}\cons\genvtwo}
&=&
\decgenvtwop{\decgenvp{\decstackp{\var}}\esub{\var}{\codetwo}} &\towhlls\\
&&\decgenvtwop{\decgenvp{\decstackp{\codetwo}}\esub{\var}{\codetwo}} &\alphaequiv&\\
&&
\decgenvtwop{\decgenvp{\decstackp{\rename\codetwo}}\esub{\var}{\codetwo}} &=\\
&&\decode{\mamstate{\rename\codetwo}{\stack}{\genv\cons\esub{\var}{\codetwo}\cons\genvtwo}}
\end{array}$\qed
\end{center}
\end{enumerate}

\emph{Digression about $\tostruct$}. Note that in the distillation theorem structural equivalence is used only to commute with stacks. The calculus and the machine in fact form a distillery also with respect to the following simpler notion of structural equivalence. Let $\eqmam$ be the smallest equivalence relation generated by the closure by (call-by-name) evaluation contexts of the axiom $\tostructapl$ in \reffig{StructEq} (page \pageref{fig:StructEq}). The next lemma guarantees that $\eqmam$ is a strong bisimulation (the proof is in \withproofs{the appendix, page \pageref{ss:mam-proofs}}\withoutproofs{\cite{distillingTR}}), and so $\eqmam$ provides another MAM distillery.

\begin{lemma}
$\eqmam$ is a strong bisimulation with respect to $\towhl$.
\end{lemma}

% !TEX root = main.tex
\begin{figure*}
\begin{center}${\setlength{\arraycolsep}{1em}
\begin{array}{c|c|c|ccc|c|c|c}
	\code\codetwo&\stack&\fstack&\genv
	&\tomachaone&
	\code&\codetwo\cons\stack&\fstack&\genv
	\\
	\l\var.\code&\codetwo\cons\stack&\fstack&\genv
	&\tomachm&
	\code&\stack&\fstack&\esub\var\codetwo \cons \genv 
	\\
	\var&\stack &\fstack&\genv_1 \cons \esub\var\code \cons \genv_2
	&\tomachatwo&
	\code&\stempty &(\genv_1,\var,\stack)\cons\fstack&\genv_2
	\\
	\codeval&\stempty&(\genv_1,\var,\stack)\cons\fstack&\genv_2
	&\tomache&
	\rename{\codeval}&\stack&\fstack&\genv_1 \cons \esub\var\codeval \cons \genv_2
\end{array}}$\end{center}
\caption{The Wadsworth Abstract Machine (WAM).\label{fig:need}}
\end{figure*}

\subsection{The Split CEK, or Revisiting the SECD Machine}
For the CEK machine we proved that the stack, that collects both arguments and functions, decodes to an evaluation context (\reflemma{cek-prop}.\ref{p:cek-prop-four}). The new CBV machine in \reffig{secd}, deemed \emph{Split CEK}, has two stacks: one for arguments and one for functions. Both will decode to evaluation contexts. The argument stack is identical to the stack of the KAM, and, accordingly, will decode to an applicative context. Roughly, the function stack decodes to contexts of the form $\evctxp{\val\ctxhole}$. More precisely, an entry of the function stack is a pair $(\clos,\stack)$, where $\clos$ is a closure $(\codeval,\env)$, and the three components $\codeval$, $\env$, and $\stack$ together correspond to the evaluation context $\decode{\stack}\ctxholep{\decode{\env}\ctxholep{\codeval\ctxhole}}$. For the acquainted reader, this new stack corresponds to the \emph{dump} of Landin's SECD machine \cite{LandinSECD}.

Let us explain the main idea. Whenever the code is an abstraction $\codeval$ and the argument stack $\stack$ is non-empty (\ie\ $\stack=\clos\cons\stacktwo$), the machine saves the active closure, given by current code $\codeval$ and environment $\env$, and the tail of the stack $\stacktwo$ by pushing a new entry $((\codeval,\env),\stacktwo)$ on the dump, and then starts evaluating the first closure $\clos$ of the stack. The syntax for dumps then is 
$$\fstack\grameq \stempty\mid(\clos,\stack)\cons\fstack$$
Every dump decodes to a context according to:
\begin{center}
	$\begin{array}{rcl@{\sep\sep}rcl}
		\decode{\stempty} & \defeq & \ctxhole&
		\decode{((\codeval,\env),\stack)\cons\fstack} & \defeq &
		\decode{\fstack}\ctxholep{\decode\stack\ctxholep{\decode\env\ctxholep{\codeval\ctxhole}}}                
	\end{array}$
\end{center}

The decoding of terms, environments, closures, and stacks is as for the KAM. The decoding of states is defined as $\decode{\scekstate{\code}{\env}{\stack}{\fstack}} \defeq 
\decode\fstack\ctxholep{
	\decode{\stack}\ctxholep{
		\decode\env\ctxholep\code
		}
	}$. The proofs for the Split CEK are \withproofs{in the appendix (page \pageref{sect:scek-proofs})}\withoutproofs{in \cite{distillingTR}}.

  \begin{lemma}[Split CEK Invariants]
  \label{l:scek-prop}
  Let $\state = \scekstate{\codetwo}{\env}{\stack}{\fstack}$ \invariantshyp{Split CEK}. Then:
    \begin{enumerate}
    \item \label{p:scek-prop-closure}
       {\em Closure:} \closprop;
    \item \label{p:scek-prop-subterm}
       {\em Subterm:} \subprop;
    \item \label{p:scek-prop-value}
       {\em Value:} the code of any closure in the dump or in any environment in $\state$ is a value;
    \item \emph{Contextual Decoding}: $\decdump$, $\decdumpp\decstack$, and $\decdumpp{\decstackp\decenv}$ are left-to-right call-by-value evaluation context.
    		\item \label{p:scek-prop-five}\emph{Name}: \lwnameprop.
		\item \label{p:scek-prop-six}\emph{Environment Size}: \envprop.

    \end{enumerate}
  \end{lemma}

 \begin{theorem}[Split CEK Distillation]
\label{tm:sCEK-sim}
	$(\emph{Split CEK},\wlscvaluelr,\eqstruct,\decode{{ }\cdot{ }})$ \distillationStatement:
	\begin{enumerate}
		\item \emph{Commutative 1}: if $\state\tomachaone\statetwo$ then $\decode{\state}\eqstruct\decode{\statetwo}$;
		\item \emph{Commutative 2}: if $\state\tomachatwo\statetwo$ then $\decode{\state}\eqstruct\decode{\statetwo}$;
		\item \emph{Multiplicative}: if $\state\tomachm\statetwo$ then $\decode{\state}\towhlcekdb\decode{\statetwo}$;
		\item \emph{Exponential}: if $\state\tomache\statetwo$ then $\decode{\state}\towhlcekls\eqstruct\decode{\statetwo}$.		
	\end{enumerate}
\end{theorem}

% !TEX root = main.tex

\section{Call-by-Need: the \WAM\ and the Merged \WAM}
\label{sect:need}
\ben{In this section we introduce a new abstract machine for call-by-need, deemed \emph{Wadsworth Abstract Machine} (\WAM). The \WAM\ arises very naturally as a reformulation of the $\wlscneed$ calculus of \refsect{ES-distance}. The motivations behind the introduction of a new machine are:
\begin{enumerate}
	\item \emph{Simplicity}: the \WAM\ is arguably simpler than all other CBNeed machines in the literature, in particular its distillation is very natural;
	\item \emph{Factorizing the Distillation of the Lazy KAM and of the SAM}: the study of the \wam\ will be followed by two sections showing how to tweak the \wam\ in order to obtain (simplifications of) two CBNeed machines in the literature, Cregut's Lazy KAM and Sestoft's machine (here called \emph{SAM}). Expressing the Lazy KAM and the SAM as modifications of the \WAM\ helps understanding their design, their distillation (that would otherwise look very technical), and their relationship;
	\item \emph{Modularity of Our Contextual Theory of Abstract Machines}: the \wam\ is obtained by applying to the KAM the following two tweaks:
	\begin{enumerate}
		\item \emph{Dump}: the \wam\ uses the dump-like approach of the Split CEK/SECD to evaluate inside explicit substitutions;
		\item \emph{Global Environments}: the \wam\ uses the global environment approach of the MAM to implement memoization;
	\end{enumerate}	
\end{enumerate}
}
\subsection{The \wam}
The \wam\ is shown in \reffig{need}. \ben{Note that when the code is a variable the transition is now commutative.} The idea is that whenever the code is a variable $\var$ and the environment has the form $\genv_1 \cons \esub\var\code \cons \genv_2$, the machine jumps to evaluate $\code$ saving the prefix of the environment $\genv_1$, the variable $\var$ on which it will substitute the result of evaluating $\code$, and the stack $\stack$. \ben{This is how hereditarily weak head evaluation context are implemented by the \wam.} In \refsect{SAM}, we will present a variant of the WAM that avoids the splitting of the environment saving $\genv_1$ in a dump entry.

The syntax for dumps is 
$$\fstack\grameq \stempty\mid(\genv,\var,\stack)\cons\fstack$$
Every dump stack decodes to a context according to:
\begin{center}
	$\begin{array}{rcl@{\sep}rcl}
		\decode{\stempty} & \defeq & \ctxhole&
		\decode{(\genv,\var,\stack)\cons\fstack} & \defeq &
        \dgenvp{\dfstackp{\dstackp{\var}}}\esub{\var}{\ctxhole}
	\end{array}$
\end{center}

The decoding of terms, environments, and stacks is defined as for the KAM.
The decoding of states is defined by $\decode{\wamstate{\code}{\stack}{\fstack}{\genv}} := \dgenvp{\dfstackp{\dstackp{\code}}}$. The decoding of contexts is static:

\begin{lemma}[Contextual Decoding]
	\ben{Let $\fstack$, $\stack$, and $\genv$ be a dump, a stack, and a global environment of the \wam, respectively. Then} $\decdump$, $\decdumpp\decstack$, $\decgenvp\decdump$, and $\decgenvp{\decdumpp{\decstack}}$ are CBNeed evaluation contexts.
\end{lemma}

Closed closures are defined as for the MAM. Given a state $\state = \scekstate{\code}{\stack}{\fstack}{\genv_0}$ with $\fstack=(\genv_1,\var_1,\stack_1)\cons\ldots\cons(\genv_n,\var_n,\stack_n)$, its closures are $(\decstackp\code,\genv_0)$ and, for $i\in\set{1,\ldots,n}$,
$$(\decode{\stack_i}\ctxholep{\var_i},\genv_{i}\cons\esub{\var_i}{\decode{\stack_{i-1}}\ctxholep{\var_{i-1}}}\cons\ldots\cons\esub{\var_1}{\decstackp\code}\cons\genv_0).$$ 

The dynamic invariants are:

  \begin{lemma}[WAM invariants]
  \label{l:wam-invariants}
  Let $\state = \scekstate{\code}{\stack}{\fstack}{\genv_0}$ \invariantshyp{WAM}, and s.t. $\fstack=(\genv_1,\var_1,\stack_1)\cons\ldots\cons(\genv_n,\var_n,\stack_n)$. Then:
    \begin{enumerate}
    \item \label{p:wam-invariants-globclos}
       {\em Global Closure:} the closures of $\state$ are closed;
    \item \label{p:wam-invariants-subterm}
       {\em Subterm:} \subprop;
    \item \label{p:wam-invariants-value}
       {\em Names:} the closures of $\state$ are well-named.
    \end{enumerate}
  \end{lemma}

For the properties of the decoding function please note that, as defined in \refsect{ES-distance}, the structural congruence $\eqstructneed$ for call-by-need is different from before.

\begin{theorem}[WAM Distillation]
\label{tm:wam-sim}
	$(\emph{WAM},\wlscneed,\eqstructneed,\decode{{ }\cdot{ }})$ \distillationStatement:
	\begin{enumerate}
		\item \emph{Commutative 1}: if $\state\tomachaone\statetwo$ then $\decode{\state}=\decode{\statetwo}$;
		\item \emph{Commutative 2}: if $\state\tomachatwo\statetwo$ then $\decode{\state}=\decode{\statetwo}$;
		\item \emph{Multiplicative}: if $\state\tomachm\statetwo$ then $\decode{\state}\towhlcekdb\eqstructneed\decode{\statetwo}$;
		\item \emph{Exponential}: if $\state\tomache\statetwo$ then $\decode{\state}\towhlcekls\alphaequiv\decode{\statetwo}$.
	\end{enumerate}
\end{theorem}

\proof
  \begin{enumerate}
    \item \emph{Commutative 1}.
%  We have
%    $\wamstate{\code\,\codetwo}{\stack}{\fstack}{\genv}
%     \tomachaone
%     \wamstate{\code}{\codetwo\cons\stack}{\fstack}{\genv}
%    $,
%  and:
    $$
      \decode{ \wamstate{\code\,\codetwo}{\stack}{\fstack}{\genv} } =
      \dgenvp{\dfstackp{\dstackp{\code\,\codetwo}}} =
      \decode{ \wamstate{\code}{\codetwo\cons\stack}{\fstack}{\genv} }
    $$

  \item \emph{Commutative 2}:
%  We have
%    $\wamstate{\var}{\stack}{\fstack}{\genv_1\cons\esub{\var}{\code}\cons\genv_2}
%     \tomachatwo
%     \wamstate{\code}{\stempty}{(\genv_1,\var,\stack)\cons\fstack}{\genv_2}$,
%  and:
    $$
      \begin{array}{llll}
      \decode{ \wamstate{\var}{\stack}{\fstack}{\genv_1\cons\esub{\var}{\code}\cons\genv_2} } & = &
      \dgenvsubtwop{\dgenvsubonep{\dfstackp{\dstackp{\var}}}\esub{\var}{\code}} & = \\
      && \decode{ \wamstate{\code}{\stempty}{(\genv_1,\var,\stack)\cons\fstack}{\genv_2} }
      \end{array}
    $$

  \item \emph{Multiplicative}.
%  We have
%    $\wamstate{\l\var.\code}{\codetwo\cons\stack}{\fstack}{\genv}
%     \tomachm
%     \wamstate{\code}{\stack}{\fstack}{\esub{\var}{\codetwo}\cons\genv}$,
%  and:
    $$
      \begin{array}{llllll}
      \decode{ \wamstate{\l\var.\code}{\codetwo\cons\stack}{\fstack}{\genv} } & = &
      \dgenvp{\dfstackp{\dstackp{(\l\var.\code)\,\codetwo}}} & \towhlcekdb \\
      && \dgenvp{\dfstackp{\dstackp{\code\esub{\var}{\codetwo}}}} & \tostructneed \mbox{\tiny Lem. \ref{l:ev-comm-struct}}& \\
      && \dgenvp{\dfstackp{\dstackp{\code}}\esub{\var}{\codetwo}} & = \\
      && \decode{ \wamstate{\code}{\stack}{\fstack}{\esub{\var}{\codetwo}\cons\genv} }
      \end{array}
    $$
Note that to apply \reflemma{ev-comm-struct} we use the global closure invariant, as $\codetwo$, being on the stack, is closed by $\genv$ and so $\decode \fstack$ does not capture its free variables.
  \item \emph{Exponential}.
%  We have
%    $\wamstate{\codeval}{\stempty}{(\genv_1,\var,\stack)\cons\fstack}{\genv_2}
%     \tomache
%     \wamstate{\rename{\codeval}}{\stack}{\fstack}{\genv_1\cons\esub{\var}{\codeval}\cons\genv_2}
%    $,
%  and:
    $$
      \begin{array}{llll}
      \decode{ \wamstate{\codeval}{\stempty}{(\genv_1,\var,\stack)\cons\fstack}{\genv_2} } & = &
      \dgenvsubtwop{\dgenvsubonep{\dfstackp{\dstackp{\var}}}\esub{\var}{\codeval}} & \towhlcekls \\
      && \dgenvsubtwop{\dgenvsubonep{\dfstackp{\dstackp{\codeval}}}\esub{\var}{\codeval}} & \alphaequiv \\
      && \dgenvsubtwop{\dgenvsubonep{\dfstackp{\dstackp{\rename\codeval}}}\esub{\var}{\codeval}} & = \\
      && \decode{ \wamstate{\rename{\codeval}}{\stack}{\fstack}{\genv_1\cons\esub{\var}{\codeval}\cons\genv_2} }
      \end{array}
    $$

  \end{enumerate}
  
\noindent    \emph{Progress}. Let $\state=\wamstate\code\stack\fstack\genv$ be a commutative normal form s.t. $\decode\state\togen\tmtwo$. If $\code$ is 
    \begin{enumerate}
    \item \emph{an application $\codetwo\codethree$}. Then a $\tomachaone$ transition applies and $\state$ is not a commutative normal form, absurd.

    \item \emph{an abstraction $\val$}. The decoding $\decode\state$ is of the form $\dgenvp{\dfstackp{\dstackp{\val}}}$.
    The stack $\stack$ and the dump $\fstack$ cannot both be empty,
    since then $\decode\state = \dgenvp{\val}$ would be normal.
    So either the stack is empty and a $\tomache$ transition
    applies, or the stack is not empty and a $\tomachm$ transition
    applies.
    
        \item \emph{a variable $\var$}. By \reflemma{wam-invariants}.\ref{p:wam-invariants-globclos}
    it must be bound by $\genv$, so a $\tomachatwo$ transition applies, and $\state$ is not a commutative normal form, absurd.
\qed
    \end{enumerate}

% !TEX root = main.tex
\begin{figure*}
\begin{center}${\setlength{\arraycolsep}{1em}
\begin{array}{c|c|ccc|c|c}
	\code\codetwo&\stack&\genv
	&\tomachaone&
	\code&\argst{\codetwo}\cons\stack&\genv
	\\
	\l\var.\code&\argst{\codetwo}\cons\stack&\genv
	&\tomachm&
	\code&\stack&\esub\var\codetwo \cons \genv
	\\
	\var&\stack&\genv_1 \cons \esub\var\code \cons \genv_2 
	&\tomachatwo&
	\code&\headst{\genv_1,\var}\cons\stack&\genv_2
	\\
	\codeval&\headst{\genv_1,\var}\cons\stack&\genv_2
	&\tomache&
	\rename{\codeval}&\stack&\genv_1 \cons \esub\var\codeval \cons \genv_2

\end{array}}$\end{center}
\caption{The Merged WAM.\label{fig:mwam}}
\end{figure*}

  \begin{figure*}[t]
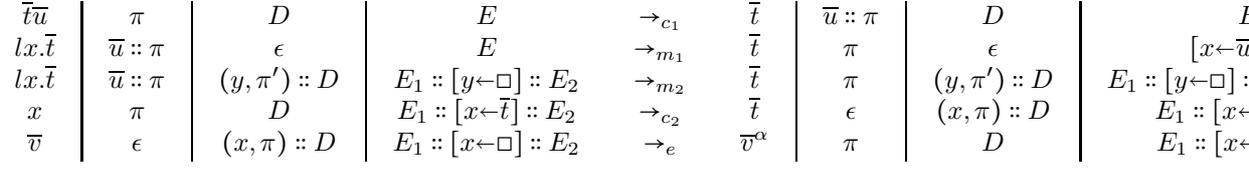

  \begin{center}${\setlength{\arraycolsep}{1em}
  \begin{array}{c|c|c|ccc|c|c|c}
	  \code\codetwo&\stack&\fstack&\genv
	  &\tomachaone&
	  \code&\codetwo\cons\stack&\fstack&\genv
	  \\
	  \l\var.\code&\codetwo\cons\stack&\stempty&\genv
	  &\tomachmone&
	  \code&\stack &\stempty&\esub\var\codetwo \cons \genv
	  \\
	  \l\var.\code&\codetwo\cons\stack&(\vartwo,\stacktwo)\cons\fstack&\genv_1\cons\esub{\vartwo}{\pwammark}\cons\genv_2
	  &\tomachmtwo&
	  \code&\stack &(\vartwo,\stacktwo)\cons\fstack&\genv_1\cons\esub{\vartwo}{\pwammark}\cons\esub{\var}{\codetwo}\cons\genv_2
	  \\
	  \var&\stack &\fstack&\genv_1 \cons \esub\var\code \cons \genv_2
	  &\tomachatwo&
	  \code&\stempty &(\var,\stack)\cons\fstack&\genv_1 \cons \esub\var\pwammark \cons \genv_2
	  \\
	  \codeval&\stempty&(\var,\stack)\cons\fstack&\genv_1 \cons \esub\var\pwammark \cons \genv_2
	  &\tomache&
	  \rename{\codeval}&\stack&\fstack&\genv_1 \cons \esub\var\codeval \cons \genv_2
  \end{array}}$\end{center}
  \caption{The Pointing WAM.\label{fig:PointingWAM}}
  \end{figure*}

\section{The Merged \WAM, or Revisiting the Lazy KAM}
Splitting the stack of the CEK machine in two we obtained a simpler form of the SECD machine. In this section we apply to the \WAM\ the reverse transformation. The result is a machine, deemed \emph{Merged \WAM}, having only one stack and that can be seen as a simpler version of Cregut's lazy KAM \cite{DBLP:journals/lisp/Cregut07} (but we are rather inspired by Danvy and Zerny's presentation in \cite{DBLP:conf/ppdp/DanvyZ13}).

To distinguish the two kinds of objects on the stack we use a marker, as for the CEK and the LAM. Formally, the syntax for stacks is:
$$\stack\grameq \stempty \mid \argst{\code}\cons\stack \mid \headst{\genv,\var}\cons\stack$$
where $\argst{\code}$ denotes a term to be used as an argument (as for the CEK) and $\headst{\genv,\var,\stack}$ is morally an entry of the dump of the WAM, where however there is no need to save the current stack. The transitions of the Merged WAM are in \reffig{mwam}.

The decoding is defined as follows

\begin{center}
	$\begin{array}{rcl@{\sep}rcl}
		\decode{\stempty} & \defeq & \ctxhole\\
		\decode{\esub\var\code \cons \genv} & \defeq & \decode{\genv}\ctxholep{\ctxhole\esub\var\code}\\
		\decode{\headst{\genv,\var}\cons\stack} & \defeq & \decode{\genv}\ctxholep{\decode{\stack}\ctxholep\var}\esub\var\ctxhole\\
		\decode{\argst{\code}\cons\stack} & \defeq & \decode{\stack}\ctxholep{\ctxhole\code}\\
		\decode{\mamstate{\code}{\stack}{\genv}} & \defeq &\decode{\genv}\ctxholep{ \decode{\stack}\ctxholep{\code}}\\
	\end{array}$
\end{center}

\begin{lemma}[Contextual Decoding]
	\ben{Let $\stack$ and $\genv$ be a stack and a global environment of the Merged \wam. Then} $\decstack$ and $\decgenvp\decstack$ are CBNeed evaluation contexts.
\end{lemma}

The dynamic invariants of the Merged WAM are exactly the same of the WAM, with respect to an analogous set of closures associated to a state (whose exact definition is omitted). The proof of the following theorem---almost identical to that of the WAM---is \withproofs{in the appendix (page \pageref{sect:mwam-proofs})}\withoutproofs{in \cite{distillingTR}}.

\begin{theorem}[Merged WAM Distillation]
\label{tm:mwam-sim}
	$(\emph{Merged WAM},\wlscneed,\eqstructneed,\decode{{ }\cdot{ }})$ \distillationStatement:
	\begin{enumerate}
		\item \emph{Commutative 1}: if $\state\tomachaone\statetwo$ then $\decode{\state}=\decode{\statetwo}$;
		\item \emph{Commutative 2}: if $\state\tomachatwo\statetwo$ then $\decode{\state}=\decode{\statetwo}$;
		\item \emph{Multiplicative}: if $\state\tomachm\statetwo$ then $\decode{\state}\towhlcekdb\eqstructneed\decode{\statetwo}$;
		\item \emph{Exponential}: if $\state\tomache\statetwo$ then $\decode{\state}\towhlcekls\alphaequiv\decode{\statetwo}$.
	\end{enumerate}
\end{theorem}

%\claudio{NOTE: the lemma used in the multiplicative transition above also requires a big invariant on the dump}

% !TEX root = main.tex

% lam2 --> machm
% lam1 --> mache
% ap --> machaone
% var --> machatwo

\section{The Pointing WAM, or Revisiting the SAM}
\label{sect:SAM}
In the WAM, the global environment is divided between the environment of the machine and the entries of the dump. On one hand, this choice makes the decoding very natural. On the other hand, one would like to keep the environment in just one place, letting the dump only collect variables and stacks. This is what we do here, exploiting the fact that  variable names can be taken as pointers (see the discussion after the invariants in \refsect{mam}). 

The new machine, called Pointing WAM, is in \reffig{PointingWAM}, and uses a new dummy constant $\pwammark$ for the substitutions whose variable is in the dump. It can be seen as a simpler version of Sestoft's abstract machine \cite{Sestoft}, here called SAM. Dumps and environments are defined by:
\begin{center}
    $\begin{array}{rcl@{\sep}rcl}
        \dump & \grameq & \stempty \mid (\var,\stack)\cons\fstack\\
        \genv& \grameq & \stempty \mid \esub\var\code\cons\genv \mid \esub\var\pwammark\cons\genv
    \end{array}$
\end{center}
A substitution of the form $\esub\var\pwammark$ is \emph{dumped}, and we also say that $\var$ is dumped. 

Note that the variables of the entries in $\dump$ appear in reverse order with respect to the corresponding substitutions in $\genv$. We will show that fact is an invariant, called \emph{duality}.

\begin{definition}[Duality $\genv\dual\dump$]
    Duality $\genv\dual\dump$ between environments and dumps is defined by
    \begin{enumerate}
        \item $\stempty\dual\stempty$;
        \item $\genv\cons\esub\var\code \dual \dump$  if $\genv \dual \dump$;
        \item $\genv\cons\esub\var\pwammark \dual (\var,\stack)\cons\dump$  if $\genv \dual \dump$.
    \end{enumerate}
\end{definition}

Note that in a dual pair the environment is always at least as long as the dump.
A dual pair $\genv\dual\dump$ decodes to a context as follows:

\begin{center}
    $\begin{array}{rcl@{\sep}rcl}
        \decode{(\genv,\stempty)} & \defeq & \decode \genv\\
        \decode{(\genv\cons\esub\var\pwammark,(\var,\stack)\cons\dump)} & \defeq &
                            \decode{(\genv,\dump)}\ctxholep{\decstackp\var}\esub\var{\ctxhole}
                        \\
        \decode{(\genv\cons\esub\var\code,(\vartwo,\stack)\cons\dump)} & \defeq & \decode{(\genv,(\vartwo,\stack)\cons\dump)}\esub\var\code\\
    \end{array}$
\end{center}

The analysis of the Pointing WAM is based on a complex invariant that includes duality plus a generalization of the global closure invariant. We need an auxiliary definition:
  \begin{definition}
  Given an environment $\genv$, we define its {\em slice} $\genv\envslice$ as the
  sequence of substitutions after the rightmost dumped substitution. Formally:
  $$
  \begin{array}{rll}
  \stempty\envslice & := & \stempty \\
  (\genv\cons\esub{\var}{\code})\envslice & := & \genv\envslice\cons\esub{\var}{\code} \\
  (\genv\cons\esub{\var}{\pwammark})\envslice & := & \stempty \\
  \end{array}
  $$
  Moreover, if an environment $\genv$ is of the form $\genv_1\cons\esub{\var}{\pwammark}\cons\genv_2$,
  we define $\genv\envslicevar{\var} := \genv_1\envslice\cons\esub{\var}{\pwammark}\cons\genv_2$.
  \end{definition}

The notion of closed closure with global environment (\refsect{mam}) is extended to dummy constants $\pwammark$ as expected.

  \begin{lemma}[Pointing WAM invariants]
  \label{l:pwam-invariants}
  Let $\state = \scekstate{\code}{\genv}{\stack}{\fstack}$ \invariantshyp{Pointing WAM}. Then:
    \begin{enumerate}
    %\item \label{p:pwam-invariants-globclos}
      % {\em Global Closure:} \globclosprop;
    \item \label{p:pwam-invariants-subterm}
       {\em Subterm:} \subprop;
    \item \label{p:pwam-invariants-value}
       {\em Names:} \gwnameprop.
    \item \label{p:pwam-invariants-dual}
       {\em Dump-Environment Duality:}

         \begin{enumerate}
         \item \label{p:pwam-invariants-dual1} $\pwamclosed{(\dstackp{\code}, \genv\envslice)}$;
         \item \label{p:pwam-invariants-dual2} for every pair $(\var,\stacktwo)$ in $\fstack$, $\pwamclosed{(\dstacktwop{\var},\genv\envslicevar{\var})}$;
         \item \label{p:pwam-invariants-dual3} $\genv\dual\dump$ holds.
         \end{enumerate}

    \item \label{p:pwam-invariants-ctx}
       {\em Contextual Decoding:} $\decode{(\genv,\dump)}$ is a call-by-need evaluation context.
    \end{enumerate}
  \end{lemma}

\begin{proof}
\withproofs{In the appendix, page \pageref{ss:sam-proofs}.}\withoutproofs{See \cite{distillingTR}.}
\end{proof}

The decoding of a state is defined as
$\decode{\pwamstate{\code}{\stack}{\genv}{\dump}} := \decode{(\genv,\dump)}\ctxholep{\dstackp{\code}}$.

\begin{theorem}[Pointing WAM Distillation]
\label{tm:pwam-sim}
	$(\emph{Pointing WAM},\wlscneed,\eqstructneed,\decode{{ }\cdot{ }})$ \distillationStatement:
    \begin{enumerate}
        \item \emph{Commutative 1 \& 2}: if $\state\tomachaone\statetwo$ or $\state\tomachatwo\statetwo$ then $\decode{\state}=\decode{\statetwo}$;
        \item \emph{Multiplicative 1 \& 2}: if $\state\tomachmone\statetwo$ or $\state\tomachmtwo\statetwo$ then $\decode{\state}\towhlcekdb\eqstructneed\decode{\statetwo}$;
        \item \emph{Exponential}: if $\state\tomache\statetwo$ then $\decode{\state}\towhlcekls\alphaequiv\decode{\statetwo}$;
    \end{enumerate}
\end{theorem}

\proof
\emph{Properties of the decoding}:
  \begin{enumerate}
  \item \caselight{Conmutative 1}.
        We have
%        $$
%        \begin{array}{llll}
%          \pwamstate{\code\,\codetwo}{\stack}{\genv}{\fstack} & \tomachaone &
%          \pwamstate{\code}{\codetwo\cons\stack}{\genv}{\fstack}
%        \end{array}
%        $$
%        and:
        $$
        \begin{array}{lllllll}
          \decode{\pwamstate{\code\,\codetwo}{\stack}{\genv}{\fstack}} & = &
          \decode{(\genv,\fstack)}\ctxholep{\dstackp{\code\,\codetwo}} & = 
          & \decode{\pwamstate{\code}{\codetwo\cons\stack}{\genv}{\fstack}}
        \end{array}
        $$
  \item \caselight{Conmutative 2}.
%        We have
%        $$
%        \begin{array}{ll}
%        \pwamstate{\var}{\stack}{\genv_1\cons\esub{\var}{\code}\cons\genv_2}{\fstack} & \tomachatwo \\
%        \pwamstate{\code}{\stempty}{\genv_1\cons\esub{\var}{\pwammark}\cons\genv_2}{(\var,\stack)\cons\fstack}
%        \end{array}
%        $$
        Note that $\genv_2$ has no dumped substitutions, since
        $\genv_1\cons\esub{\var}{\pwammark}\cons\genv_2 \dual (\var,\stack)\cons\fstack$.
        Then:
        $$
        \begin{array}{ll}
        \decode{\pwamstate{\var}{\stack}{\genv_1\cons\esub{\var}{\code}\cons\genv_2}{\fstack}} & = \\
        \dgenvsubtwop{\decode{(\genv_1,\fstack)}\ctxholep{\dstackp{\var}}\esub{\var}{\code}} & = \\
        \decode{\pwamstate{\code}{\stempty}{\genv_1\cons\esub{\var}{\pwammark}\cons\genv_2}{(\var,\stack)\cons\fstack}}
        \end{array}
        $$
  \item \caselight{Multiplicative, empty dump}.
%        We have
%        $\pwamstate{\l\var.\code}{\codetwo\cons\stack}{\genv}{\stempty}
%          \tomachm
%          \pwamstate{\code}{\stack}{\esub{\var}{\codetwo}\cons\genv}{\stempty}$, and:
       $$
        \begin{array}{llllll}
        \decode{\pwamstate{\l\var.\code}{\codetwo\cons\stack}{\genv}{\stempty}} & = &
        \dgenvp{\dstackp{(\l\var.\code)\,\codetwo}} & \towhlcekdb \\
        && \dgenvp{\dstackp{\code\esub{\var}{\codetwo}}} & \tostructapl^* \mbox{\tiny Lem. \ref{l:ev-comm-struct}}\\
        && \dgenvp{\dstackp{\code}\esub{\var}{\codetwo}} & = \\
        && \decode{\pwamstate{\code}{\stack}{\esub{\var}{\codetwo}\cons\genv}{\stempty}}
        \end{array}
        $$
  \item \caselight{Multiplicative, non-empty dump}.
%        We have
%        $$
%        \begin{array}{ll}
%        \pwamstate{\l\var.\code}{\codetwo\cons\stack}{\genv_1\cons\esub{\vartwo}{\pwammark}\cons\genv_2}{(\vartwo,\stacktwo)\cons\fstack} & \tomachm \\
%        \pwamstate{\code}{\stack}{\genv_1\cons\esub{\vartwo}{\pwammark}\cons\esub{\var}{\codetwo}\cons\genv_2}{(\vartwo,\stacktwo)\cons\fstack}
%        \end{array}
%        $$
%        and:
        $$
        \begin{array}{ll}
        \decode{\pwamstate{\l\var.\code}{\codetwo\cons\stack}{\genv_1\cons\esub{\vartwo}{\pwammark}\cons\genv_2}{(\vartwo,\stacktwo)\cons\fstack}} & = \\
        \dgenvsubtwop{\decode{(\genv_1,\dump)}\ctxholep{\dstacktwop{\vartwo}}\esub{\vartwo}{\dstackp{(\l\var.\code)\,\codetwo}}} & \towhlcekdb \\
        \dgenvsubtwop{\decode{(\genv_1,\dump)}\ctxholep{\dstacktwop{\vartwo}}\esub{\vartwo}{\dstackp{\code\esub{\var}{\codetwo}}}} & \tostructneed  \mbox{\tiny Lem. \ref{l:ev-comm-struct}}\\
        \dgenvsubtwop{\decode{(\genv_1,\dump)}\ctxholep{\dstacktwop{\vartwo}}\esub{\vartwo}{\dstackp{\code}}\esub{\var}{\codetwo}} & = \\
        \decode{\pwamstate{\code}{\stack}{\genv_1\cons\esub{\vartwo}{\pwammark}\cons\esub{\var}{\codetwo}\cons\genv_2}{(\vartwo,\stacktwo)\cons\fstack}}
        \end{array}
        $$
  \item \caselight{Exponential}.
        $$
        \begin{array}{ll}
         \decode{\pwamstate{\codeval}{\stempty}{\genv_1\cons\esub{\var}{\pwammark}\cons\genv_2}{(\var,\stack)\cons\fstack}} & = \\
         \dgenvsubtwop{\decode{(\genv_1,\dump)}\ctxholep{\dstackp{\var}}\esub{\var}{\val}} & \towhlcekls \\
         \dgenvsubtwop{\decode{(\genv_1,\dump)}\ctxholep{\dstackp{\val}}\esub{\var}{\val}} & \alphaequiv \\
         \dgenvsubtwop{\decode{(\genv_1,\dump)}\ctxholep{\dstackp{\rename\val}}\esub{\var}{\val}} & = \\
         \decode{\pwamstate{\rename\codeval}{\stack}{\genv_1\cons\esub{\var}{\codeval}\cons\genv_2}{\fstack}}
        \end{array}
        $$
  \end{enumerate}

\noindent \emph{Progress}.
  Let $\state = \pwamstate{\code}{\stack}{\genv}{\dump}$ be a commutative normal form s.t. $\decode\state\togen\tmtwo$. If $\code$ is 
\begin{itemize}
		\item \emph{an application $\codetwo\codethree$}. Then a $\tomachaone$ transition applies and $\state$ is not a commutative normal form, absurd.

	\item \emph{a variable $\var$}. By the machine invariant, $\var$ must be bound by $\genv\envslice$.
  So $\genv = \genv_1\cons\esub{\var}{\codetwo}\cons\genv_2$, a $\tomachatwo$ transition applies, and $\state$ is not a commutative normal form, absurd.
  
	\item   \emph{an abstraction $\codeval$}. Two cases:
  \begin{itemize}
  \item \caselight{The stack $\stack$ is empty}.
        The dump $\fstack$ cannot be empty, since if $\fstack = \stempty$
        we have that $\decode\state = \denvp{\codeval}$ is normal. So $\fstack = (\var,\stacktwo)\cons\fstacktwo$. By duality,
        $\genv = \genv_1\cons\esub{\var}{\pwammark}\cons\genv_2$ and a $\tomache$
        transition applies.
  \item \caselight{The stack $\stack$ is non-empty}.
        If the dump $\fstack$ is empty, the first case of $\tomachm$ applies.
        If $\fstack = (\var,\stacktwo)\cons\fstacktwo$, by duality
        $\genv = \genv_1\cons\esub{\var}{\pwammark}\cons\genv_2$ and the second case of
        $\tomachm$ applies.\qed
  \end{itemize}

\end{itemize}

% !TEX root = main.tex
\section{Distillation Preserves Complexity}
\label{sect:complexity}
Here, for every abstract machine we bound the number of commutative steps $\sizeadm{\exec}$ in an execution $\exec$ in terms of 
\begin{enumerate}
\item the number of principal steps $\sizelog{\exec}$,
\item the size $\size\code$ of the initial code $\code$. 
\end{enumerate}
The analysis only concerns the machines, but via the distillation theorems it expresses the length of the machine executions as a linear function of the length of the distilled derivations in the calculi. For every distillery, we will prove that the relationship is linear in both parameters, namely $\sizeadm{\exec}=O(\size{\code}\cdot\sizelog{\exec})$ holds. 

\begin{definition}
Let $\mach$ be a distilled abstract machine  and $\exec:\state\tomach^*\statetwo$ be an execution of initial code $\code$. $\mach$ is 
\begin{enumerate}
	\item \deff{Globally bilinear} if $\sizeadm{\exec}=O((\size{\code}+1)\cdot\sizelog{\exec})$.
	
	\item \deff{Locally linear} if whenever $\statetwo\tomacha^k\statethree$ then $k =O(\size{\code})$.
	\end{enumerate}
\end{definition}

%The following remark, proved by an easy inspection of the machines, will be used in the analysis:
%
%\begin{lemma}
%\label{l:last-principal}
%Let $\exec$ be an execution to a final state for any of the machines in the paper. The last transition of $\exec$, if any, is principal.
%\end{lemma}

The next lemma shows that local linearity is a sufficient condition for global bilinearity.

\begin{proposition}[Locally Linear $\Rightarrow$ Globally Bilinear]
	\label{prop:quant-anal}
	Let $\mach$ be a locally linear distilled abstract machine, and $\exec$ an execution of initial code $\code$. Then $\mach$ is globally bilinear.
\end{proposition}
\begin{proof}
 The execution $\rho$ writes uniquely as $\tomacha^{k_1}\tomachnoa^{h_1}\ldots\tomacha^{k_m}\tomachnoa^{h_m}$. By hypothesis $k_i=O(\size{\code})$ for every $i\in\set{1,\ldots,m}$. From $m\leq \sizelog{\exec}$ follows that $\sizeadm{\exec}=O(\size{\code}\cdot\sizelog{\exec})$. We conclude with $\size{\exec}=\sizelog{\exec}+\sizeadm{\exec}= \sizelog{\exec}+ O(\size{\code}\cdot\sizelog{\exec})= O((\size{\code}+1)\cdot\sizelog{\exec})$.
\end{proof}

Call-by-name and call-by-value machines are easily seen to be locally linear, and thus globally bilinear.

\begin{theorem}
	\label{tm:value-name-glob-bilin}
	KAM, MAM, CEK, LAM, and the Split CEK are locally linear, and so also globally bilinear.
\end{theorem}

\proof
	\begin{enumerate}
		\item \emph{KAM/MAM}. Immediate: $\tomacha$ reduces the size of the code, that is bounded by $\size{\code}$ by the subterm invariant.
		\item \emph{CEK}. Consider the following measure for states: 
\begin{center}
$\card(\cekstate \codetwo \env \stack) := \begin{cases}
		\size{\codetwo}+\size{\codethree} & \textrm{if }\stack=\argst{\codethree,\envtwo}\cons\stacktwo \\
		\size{\codetwo} & \textrm{otherwise}
	\end{cases}$
	\end{center}
By direct inspection of the rules, it can be seen that both $\tomachaone$ and $\tomachatwo$ transitions decrease the value of $\card$ for CEK states, and so the relation $\tomachaone\cup\tomachatwo$ terminates (on reachable states). Moreover, both $\size{\codetwo}$ and $\size{\codethree}$ are bounded by $\size{\code}$ by the subterm invariant (\reflemma{cek-prop}.\ref{p:cek-prop-two}), and so $k\leq 2\cdot\size{\code}=O(\size{\code})$.

		\item \emph{LAM} and \emph{Split CEK}. Minor variations over the CEK, see \withproofs{the appendix (page \pageref{ss:complexity-proofs})}\withoutproofs{\cite{distillingTR}}.\qed
	\end{enumerate}\medskip

Call-by-need machines are not locally linear, because a sequence of $\tomachatwo$ steps (remember $\tomacha \defeq \tomachaone\cup\tomachatwo$) can be as long as the environment $\env$, that is not bound by $\size\code$ (as for the MAM). Luckily, being locally linear is not a necessary condition for global bilinearity. We are in fact going to show that call-by-need machines are globally bilinear. The key observation is that $\sizecomtwo{\exec}$ is not only locally but also globally bound by $\sizep\exec$, as the next lemma formalizes. 

We treat the WAM. The reasoning for the Merged WAM and for the Pointing WAM is analogous. Define $\size\stempty\defeq 0$ and $\size{(\genv,\var,\stack)\cons\fstack}\defeq 1+\size\fstack$. We have:

\begin{lemma}
\label{l:wam-var-linear}
Let $\state=\scekstate{\code}{\stack}{\fstack}{\genv}$ be a WAM state, reached by the execution $\exec$. Then 
\begin{enumerate}
	\item $\sizecomtwo{\exec}=\sizee{\exec}+\size\fstack$.
	\item $\size\genv +\size\fstack\leq \sizem\exec$
	\item \label{p:wam-var-linear3} $\sizecomtwo{\exec}\leq \sizee{\exec}+\sizem\exec=\sizep\exec$
\end{enumerate}
\end{lemma}

\proof
\begin{enumerate}
	\item Immediate, as $\tomachatwo$ is the only transition that pushes elements on $\fstack$ and $\tomache$ is the only transition that pops them.

	\item The only rule that produces substitutions  is $\tomachm$. Note that 1) $\tomachatwo$ and $\tomache$ preserve the global number of substitutions in a state; 2) $\genv$ and $\fstack$ are made out of substitutions, if one considers every entry $(\genv,\var,\stack)$ of the dump as a substitution on $\var$ (and so the statement follows); 3) the inequality is given by the fact that an entry of the dump stocks an environment (counting for many substitutions).
	\item Substitute Point 2 in Point 1.\qed
\end{enumerate}

\begin{theorem}
The WAM has globally linear commutations.
\end{theorem}

\proof
	Let $\exec$ be an execution of initial code $\code$. Define $\tomachnoap \defeq \tomache \cup \tomachm\cup\tomachatwo$ and note $\sizenoap\exec$ the number of its steps in $\exec$. We estimate $\tomacha \defeq \tomachaone\cup\tomachatwo$ by studying its components separately. For $\tomachatwo$, \reflemma{wam-var-linear}.\ref{p:wam-var-linear3} proves $\sizecomtwo\exec \leq \sizep\exec = O(\sizep\exec)$.
For $\tomachaone$, as for the KAM, the length of a maximal $\tomachaone$ subsequence of $\exec$ is bounded by $\size\code$. The number of $\tomachaone$ maximal subsequences of $\exec$ is bounded by $\sizenoap\exec$, that by \reflemma{wam-var-linear}.\ref{p:wam-var-linear3} is linear in $O(\sizep\exec)$. Then $\sizeap\exec=O(\size{\code}\cdot\sizelog{\exec})$. Summing up, 
	\begin{center}
	$\sizecomtwo\exec + \sizeap\exec = O(\sizep\exec) + O(\size{\code}\cdot\sizelog{\exec}) = O((\size{\code}+1)\cdot\sizelog{\exec})$\qed
	\end{center}

\ben{The analysis presented here is complemented by the study in \cite{valuevariables}, where the number of exponential steps $\togene$ in a derivation $\deriv$ is shown to be polynomial (actually quadratic in call-by-name and linear in call-by-value/need) in terms of the number of multiplicative steps $\togenm$ in $\deriv$. Given our distillation theorems, the results in \cite{valuevariables} equivalently relate the exponential and multiplicative transitions of the abstract machines studied in this paper. Note that these derived results could very hardly be obtained directly on the machines, \ie\ without distillation.}
\section{Conclusions}
The novelty of our study is the use of the linear substitution calculus (\lsc) to discriminate between abstract machine transitions: some of them---the principal ones---are simulated, and thus shown to be logically relevant, while the others---the commutative ones---are sent on the structural congruence and have to be considered as  bookkeeping operations. On one hand, the \lsc\ is a sharp tool to study abstract machines. On the other hand, it provides an alternative \cben{framework}{to abstract machines} which is \emph{simpler} while being \emph{conservative} at the level of complexity analysis.
%\claudio{==NOTE: there is a (sequence of) (rejected (also by me)) papers by Van Bakel (and Vignotti) that encode $\lambda$-calculi (and $X$-calculi, etc.) with explicit substitutions into (an ad-hoc version of) the $\pi$-calculus. One of the main feature of the encoding is that several computational steps in the calculus with explicit substitutions disappear, being adsorbed into the congruence of the $\pi$-calculus. I pointed them to Beniamino's work, but the converse is also true: you should also have a look at them and maybe cite them}.

%\emph{Future work.} Distance calculi naturally suggest the design of alternative machines \emph{at a distance} (with a unique global environment, no closures, and copy-based). We plan to investigate them, in particular for call-by-need. We also want to pursue the meta-theory of the new call-by-need calculus at a distance. We also wonder if the quadratic bound for the SAM may be improved by a sharper (amortized?) analysis. The extension of our approach to other machines (for instance Cregut's machine \cite{DBLP:journals/lisp/Cregut07}) is also a natural direction for future work.

\section*{Acknowledgments}
A special acknowledgement to Claudio Sacerdoti Coen, for many useful discussions, comments and corrections to the paper. In particular, we owe him the intuition that a global analysis of call-by-need commutative rules may provide a linear bound. This work was partially supported by the ANR projects \textsc{Logoi} (10-BLAN-0213-02) and \textsc{Coquas} (ANR-12-JS02-006-01), by the French-Argentinian Laboratory in Computer Science {\bf INFINIS}, the French-Argentinian project {\bf ECOS-Sud} A12E04, the Qatar National Research Fund under grant NPRP 09-1107-1-168.

\bibliographystyle{alpha}
\bibliography{\macrospath/biblio}
%% The bibliography should be embedded for final submission.
%
%\begin{thebibliography}{}
%\softraggedright
%
%\bibitem[Smith et~al.(2009)Smith, Jones]{smith02}
%P. Q. Smith, and X. Y. Jones. ...reference text...
%
%\end{thebibliography}

% !TEX root = main.tex
\section{Technical Appendix: proofs of the determinism of the calculi (\refprop{CbNDet})}
\label{sect:determ-proofs}
%%%%%%%%%%%%%%%%%%%%%%%%%%%%%%%%%%%%%%%%%%%%%%%%%%%%%%%%%%%%%%%%%%%%%%%%%%%%%%%%
% !TEX root = ../main.tex
\subsection{Call-by-Name}
\label{sect:ProofCbNDet}
\renewcommand{\evctx}{\whctx}
\renewcommand{\evctxtwo}{\whctxtwo}
\renewcommand{\evctxp}[1]{\whctxp{#1}}
\renewcommand{\evctxtwop}[1]{\whctxtwop{#1}}

Let $\tm=\evctx_1\ctxholep{\tmfour_1}=\evctx_2\ctxholep{\tmfour_2}$. By induction on the structure of $\tm$. Cases:
\begin{itemize}
\item \casealt{Variable or an abstraction} Vacuously true, because there is no redex.
\item \casealt{Application} Let $\tm=\tmtwo\tmthree$. Suppose that one of the two evaluation contexts, for instance $\evctx_1$, is equal to $\ctxhole$. Then, we must have $\tmtwo=\l\var.\tmtwo'$, but in that case it is easy to see that the result holds, because $\evctx_2$ cannot have its hole to the right of an application (in $\tmthree$) or under an abstraction (in $\tmtwo'$). We may then assume that none of $\evctx_1,\evctx_2$ is equal to $\ctxhole$. In that case, we must have $\evctx_1=\evctx_1'\tmthree$ and $\evctx_2=\evctx_2'\tmthree$, and we conclude by induction hypothesis.
\item \casealt{Substitution} Let $\tm=\tmtwo\esub\var\tmthree$. This case is entirely analogous to the previous one.
\end{itemize}
%%%%%%%%%%%%%%%%%%%%%%%%%%%%%%%%%%%%%%%%%%%%%%%%%%%%%%%%%%%%%%%%%%%%%%%%%%%%%%%%

\subsection{Left-to-Right Call-by-Value}
\label{sect:ProofCbvCekDet}
\renewcommand{\evctx}{\cbvctx}
\renewcommand{\evctxtwo}{\cbvctxtwo}
\renewcommand{\evctxp}[1]{\cbvctxp{#1}}
\renewcommand{\evctxtwop}[1]{\cbvctxtwop{#1}}

We prove the following statement, of which the determinism of the reduction is a consequence. 

\begin{lemma}
Let $\tm$ be a term. Then $\tm$ has at most one subterm $\tmtwo$ that
verifies both~(i) and~(ii):
\begin{enumerate}
	\item[(i)]  Either $\tmtwo$ is a variable $\var$, or $\tmtwo$ is an application $\sctxp{\val}\sctxtwop{\valtwo}$,
				for $\val, \valtwo$ being values.
	\item[(ii)] $\tmtwo$ is under a left-to-right call-by-value evaluation context, \ie\ $\tm = \evctxp{\tmtwo}$.
\end{enumerate}
\end{lemma}

From the statement it follows that there is at most one $\towhlcek$-redex in $\tm$, \ie\ $\towhlcek$ is deterministic.

\begin{proof}
by induction on the structure of $\tm$:
\begin{itemize}
\item \casealt{$\tm$ is a variable} There is only one subterm, under the
      empty evaluation context.

\item \casealt{$\tm$ is an abstraction} There are no subterms that verify both (i) and (ii), since the only possible evaluation context is the empty one.

\item \casealt{$\tm$ is an application $\tmthree\,\tmfour$} There are three possible situations:
      \begin{itemize}
      \item \caselight{The left subterm $\tmthree$ is not of the form $\sctxp{\val}$}. Then $\tmtwo$ cannot be at
            the root, \ie\ $\tmtwo \neq \tm$.
            Since $\tmthree\ctxhole$ is not an evaluation context, $\tmtwo$
            must be internal to $\ctxhole\tmfour$, which is an evaluation context.
            We conclude by \ih.
      \item \caselight{The left subterm $\tmthree$ is of the form $\sctxp{\val}$ with $\val$ a value, but the right subterm $\tmfour$ is not}. Then $\tmtwo$ cannot be a subterm of $\tmthree$,
            and also $\tmtwo \neq \tm$.
            Hence, if there is a subterm $\tmtwo$ as in the statement, it
            must be internal to the evaluation context $\tmthree\ctxhole$.
            We conclude by \ih.
      \item \caselight{Both subterms have that form}, \ie\ $\tmthree = \sctxp{\val}$ and $\tmfour = \sctxtwop{\valtwo}$
            with $\val$ and $\val'$ values. The only subterm that
            verifies both (i) and (ii) is $\tmtwo = \tm$.
      \end{itemize}
\item \casealt{$\tm$ is a substitution $\tmthree\esub{\var}{\tmfour}$}  Any
      occurrence of $\tmtwo$ must be internal to $\tmthree$ (because $\tmthree\esub{\var}{\ctxhole}$ is not an evaluation context).
      We conclude by \ih\ that there is at most one such occurrence.
\end{itemize}
\end{proof}
%%%%%%%%%%%%%%%%%%%%%%%%%%%%%%%%%%%%%%%%%%%%%%%%%%%%%%%%%%%%%%%%%%%%%%%%%%%%%%%%

\subsection{Right-to-Left Call-by-Value}
\label{sect:ProofCbvLamDet}
\renewcommand{\evctx}{\scbvctx}
\renewcommand{\evctxtwo}{\scbvctxtwo}
\renewcommand{\evctxp}[1]{\scbvctxp{#1}}
\renewcommand{\evctxtwop}[1]{\scbvctxtwop{#1}}
Exactly as in the case for left-to-right call-by-value, we prove the following
property, from which determinism of the reduction follows. 

\begin{lemma}
Let $\tm$ be a term. Then $\tm$ has at most one subterm
$\tmtwo$ that verifies both~(i) and~(ii):
\begin{enumerate}
    \item[(i)] $\tmtwo$ is either a variable $\var$ or an application
               $\sctxp{\val}\sctxtwop{\valtwo}$, where $\val$ and $\valtwo$
               are values.
    \item[(ii)] $\tmtwo$ is under a right-to-left call-by-value evaluation context,
                \ie\ $\tm = \evctxp{\tmtwo}$.
\end{enumerate}
\end{lemma}

As a corollary, any term $\tm$ has at most one $\towhllam$-redex.

\begin{proof}
By induction on the structure of $\tm$:
\begin{itemize}
\item \casealt{Variable or abstraction} Immediate.
\item \casealt{Application} If $\tm = \tmthree\,\tmfour$, there are three cases:
      \begin{itemize}
      \item \caselight{The right subterm $\tmfour$ is not of the form $\sctxtwop{\valtwo}$}. Then $\tmtwo$ cannot be at the root.
            Since $\ctxhole\,\tmfour$ is not an evaluation context,
            $\tmtwo$ must be internal to $\tmfour$ and we conclude by \ih.
      \item \caselight{The right subterm $\tmfour$ is of the form $\sctxtwop{\valtwo}$ but the left subterm $\tmthree$ is
            not}. Again $\tmtwo$ cannot be
            at the root. Moreover, $\tmfour$ has no applications or variables
            under an evaluation context.
            Therefore $\tmtwo$ must be internal to $\tmthree$ and we conclude by \ih.
      \item \caselight{Both subterms have that form}, \ie\ $\tmthree = \sctxp{\val}$ and $\tmfour = \sctxtwop{\valtwo}$. We first note that $\tmthree$ and $\tmfour$ have no applications or variables under an evaluation context. The only possibility
            that remains is that $\tmtwo$ is at the root, \ie\ $\tmtwo = \tm$.
      \end{itemize}
\item \casealt{Substitution} If $\tm = \tmthree\esub{\var}{\tmfour}$ is a substitution,
      $\tmtwo$ must be internal to $\tmthree$ (because $\tmthree\esub{\var}{\ctxhole}$ is not an evaluation context), and we conclude by\ \ih.
\end{itemize}
\end{proof}

%%%%%%%%%%%%%%%%%%%%%%%%%%%%%%%%%%%%%%%%%%%%%%%%%%%%%%%%%%%%%%%%%%%%%%%%%%%%%%%%

\subsection{Call-by-Need}
\label{sect:ProofCbNeedDet}
\renewcommand{\evctx}{\cbndctx}
\renewcommand{\evctxtwo}{\cbndctxtwo}
\renewcommand{\evctxp}[1]{\cbndctxp{#1}}
\renewcommand{\evctxtwop}[1]{\cbndctxtwop{#1}}

We first need an auxiliary result:
\begin{lemma}
	\label{l:CbNeedEvCtx}
	Let $\tm:=\evctxp\var$ for an evaluation context $\evctx$ such that $\var\in\fv\tm$. Then:
	\begin{enumerate}
		\item \label{p:CbNeedEvCtx-1} for every substitution context $\sctx$ and abstraction $\val$, $\tm\neq\sctxp\val$;
		\item \label{p:CbNeedEvCtx-2} for every evaluation context $\evctxtwo$ and variable $\vartwo$, $\tm=\evctxtwop\vartwo$ implies $\evctxtwo=\evctx$ and $\vartwo=\var$;
		\item \label{p:CbNeedEvCtx-3} $\tm$ is a call-by-need normal form.
	\end{enumerate}
\end{lemma}
\begin{proof}
	In all points we use a structural induction on $\evctx$. For point~1:
	\begin{itemize}
		\item $\evctx=\ctxhole$: obvious.
		\item $\evctx=\evctx_1\tmtwo$: obvious.
		\item $\evctx=\evctx_1\esub\vartwo\tmthree$: suppose that $\sctx=\sctxtwo\esub\vartwo\tmthree$ (for otherwise the result is obvious); then we apply the induction hypothesis to $\evctx_1$ to obtain $\evctx_1\ctxholep\var\neq\sctxtwop\val$.
		\item $\evctx=\evctx_1\ctxholep\vartwo\esub\vartwo{\evctx_2}$: suppose that $\sctx=\sctxtwo\esub\vartwo{\evctx_2\ctxholep\var}$ (for otherwise the result is obvious); then we apply the induction hypothesis to $\evctx_1$ to obtain $\evctx_1\ctxholep\vartwo\neq\sctxtwop\val$.
	\end{itemize}
	
	For point 2:
	\begin{itemize}
		\item $\evctx=\ctxhole$: obvious.
		\item $\evctx=\evctx_1\tmtwo$: we must necessarily have $\evctxtwo=\evctxtwo_1\tmtwo$ and we conclude by induction hypothesis.
		\item $\evctx=\evctx_1\esub\varthree\tmtwo$: in principle, there are two cases. First, we may have $\evctxtwo=\evctxtwo_1\esub\varthree\tmtwo$, which allows us to conclude immediately by induction hypothesis, as above. The second possibility would be $\evctxtwo=\evctxtwo_1\ctxholep\varthree\esub\varthree{\evctxtwo_2}$, with $\evctxtwo_2\ctxholep\vartwo=\tmtwo$, but this is actually impossible. In fact, it would imply $\evctx_1\ctxholep\var=\evctxtwo_1\ctxholep\varthree$, which by induction hypothesis would give us $\varthree=\var$, contradicting the hypothesis $\var\in\fv\tm$.
		\item $\evctx=\evctx_1\ctxholep\varthree\esub\varthree{\evctx_2}$: by symmetry with the above case, the only possibility is $\evctxtwo=\evctx_1\ctxholep\varthree\esub\varthree{\evctxtwo_2}$, which allows us to conclude immediately by induction hypothesis.
	\end{itemize}

	For point 3, let $\tmfour$ be a redex (\ie, a term matching the left hand side of $\rtodb$ or $\rtolsv$) and let $\evctxtwo$ be an evaluation context. We will show by structural induction on $\evctx$ that $\tm\neq\evctxtwop\tmfour$. We will do this by considering, in each inductive case, all the possible shapes of $\evctxtwo$.
	\begin{itemize}
		\item $\evctx=\ctxhole$: obvious.
		\item $\evctx=\evctx_1\tmtwo$: the result is obvious unless $\evctxtwo=\ctxhole$ or $\evctxtwo=\evctxtwo_1\tmtwo$. In the latter case, we conclude by induction hypothesis (on $\evctx_1$). In the former case, since $\tmfour$ is a redex, we are forced to have $\tmfour=\sctxp\val\tmtwo'$ for some abstraction $\val$, substitution context $\sctx$ and term $\tmtwo'$. Now, even supposing $\tmtwo'=\tmtwo$, we are still allowed to conclude because $\evctx_1\ctxholep\var\neq\sctxp\val$ by point 1.
		\item $\evctx=\evctx_1\esub\vartwo\tmtwo$: the result is obvious unless:
		\begin{itemize}
			\item $\evctxtwo=\ctxhole$: this time, the fact that $\tmfour$ is a redex forces $\tmfour=\evctxtwo_1\ctxholep\vartwo\esub\vartwo\tmtwo$. Even if we admit that $\tmtwo=\sctxp\val$, we may still conclude because $\var\neq\vartwo$ (by the hypothesis $\var\in\fv\tm$), hence $\evctx_1\ctxholep\var\neq\evctxtwo_1\ctxholep\vartwo$ by point 2.
			\item $\evctxtwo=\evctxtwo_1\esub\vartwo\tmtwo$: immediate by induction hypothesis on $\evctx_1$.
			\item $\evctxtwo=\evctxtwo_1\ctxholep\vartwo\esub\vartwo{\evctxtwo_2}$: even if $\evctxtwo_2\ctxholep\tmfour=\tmtwo$, we may still conclude because, again, $\var\neq\vartwo$ implies $\evctx_1\ctxholep\var\neq\evctxtwo_1\ctxholep\vartwo$ by point 2.
		\end{itemize}
		\item $\evctx=\evctx_1\ctxholep\vartwo\esub\vartwo{\evctx_2}$: again, the result is obvious unless:
		\begin{itemize}
			\item $\evctxtwo=\ctxhole$: the fact that $\tmfour$ is a redex implies $\tmfour=\evctxtwo_1\ctxholep\vartwo\esub\vartwo{\sctxp\val}$. Even assuming $\evctxtwo_1=\evctx_1$, we may still conclude because $\evctx_2\ctxholep\var\neq\sctxp\val$ by point 1.
			\item $\evctxtwo=\evctxtwo_1\esub\vartwo{\evctx_2\ctxholep\var}$: since $\vartwo\in\fv{\evctx_1\ctxholep\vartwo}$, we conclude because the induction hypothesis gives us $\evctx_1\ctxholep\vartwo\neq\evctxtwo_1\ctxholep\tmfour$.
			\item $\evctxtwo=\evctx_1\ctxholep\vartwo\esub\vartwo{\evctxtwo_2}$: we conclude at once by applying the induction hypothesis to $\evctx_2$.
		\end{itemize}
	\end{itemize}
\end{proof}

%	We do the proof
Now, the proof of \refprop{CbNeedDet} is 
by structural induction on $\tm:=\evctx_1\ctxholep{\tmfour_1}=\evctx_2\ctxholep{\tmfour_2}$. Cases:

\begin{itemize}
 \item \casealt{Variable or abstraction}
 Impossible, since variables and abstractions are both call-by-need normal.
 
 \item \casealt{Application, \ie\ $\tm=\tmtwo\tmthree$} This case is treated exactly as in the corresponding case of the proof of \refprop{CbNDet}.
  
 \item \casealt{Substitution, \ie\ $\tm=\tmtwo\esub\var\tmthree$}
 Cases:
 \begin{itemize}
 \item \caselight{Both contexts have their holes in $\tmtwo$ or $\tmthree$}. It follows from the \ih. 
 
 \item \caselight{One of the contexts---say $\evctx_1$---is empty, \ie\  $\tmtwo=\evctx_3\ctxholep\var$, $\tmthree=\sctxp\val$, and $\tmfour_1=\evctx_3\ctxholep\var\esub\var {\sctxp\val}$}. This case is impossible. Indeed, 1) the hole of $\evctx_2$ cannot be in $\sctxp\val$, because it is call-by-need normal, and 2) it cannot be inside $\evctx_3\ctxholep\var$ because by Lemma \ref{l:CbNeedEvCtx}.\ref{p:CbNeedEvCtx-3} $\evctxp{\var}$ is call-by-need normal.
 
 \item \caselight{One of the contexts---say $\evctx_1$---has its hole in $\tmthree$ and the other one has its hole in $\tmtwo$, \ie\  $\evctx_1=\evctx_3\ctxholep\var\esub\var{\evctx_4}$ and $\evctx_2=\evctx_5\esub\var\tmthree$}. This case is impossible, because by Lemma \ref{l:CbNeedEvCtx}.\ref{p:CbNeedEvCtx-3} $\evctx_3\ctxholep\var$ is call-by-need normal. 
 \end{itemize}
 \end{itemize}

%%%%%%%%%%%%%%%%%%%%%%%%%%%%%%%%%%%%%%%%%%%%%%%%%%%%%%%%%%%%%%%%%%%%%%%%%%%%%%%%
\section{Technical Appendix:  proofs of strong bisimulation}
\label{sect:strong-bisim-proofs}
\subsection{Proof of \refprop{strong-bis} ($\eqstruct$ is a strong bisimulation) for call-by-name}
\renewcommand{\evctx}{\whctx}
\renewcommand{\evctxtwo}{\whctxtwo}
\renewcommand{\evctxp}[1]{\whctxp{#1}}
\renewcommand{\evctxtwop}[1]{\whctxtwop{#1}}
\label{sect:ProofCbNStrongBis}

% !TEX root = ../main.tex

	Before proving the main result, we need two auxiliary lemmas, proved by straightforward inductions on $\evctx$:

\begin{lemma}
\label{l:eqstruct-and-ctx-name}
Let $\tm$ be a term, $\evctx$ be a call-by-name evaluation context not capturing any variable in $\fv{\tm}$, and $\var\notin\fv{\evctxp\vartwo}$. 
Then $\evctxp{\tm\esub\var\tmtwo}\eqstruct\evctxp\tm\esub\var\tmtwo$.
\end{lemma}

		\begin{lemma}
		\label{l:cbname_eqstruct_preserves_shapes}
		The equivalence relation $\tostruct$ as defined for call-by-name
		preserves the shape of $\evctxp{\var}$. More precisely,
		if $\evctxp{\var} \eqstruct \tm$, with $\var$ not captured by $\evctx$,
        then $\tm$ is of the form $\evctxtwop{\var}$, with $\var$ not captured by $\evctxtwo$.
		\end{lemma}

Now we turn to the proof of \refprop{strong-bis} itself.

Let $\tostructsym$ be the symmetric closure of the
union of the axioms defining $\eqstruct$ for call-by-name, that is of $\tostructgc\cup\tostructdup\cup\tostructap\cup\tostructcom\cup\tostructes$.
Note that $\eqstruct$ is the reflexive--transitive closure of $\tostructsym$.
The proof is in two parts:
\begin{itemize}
\item[] (I) Prove the property holds for $\tostructsym$, \ie\
            if $\tm\towhl_{a}\tmtwo$ and $\tm\tostructsym\tmthree$, there exists $\tmfour$ s.t. $\tmthree\towhl_{a}\tmfour$ and $\tmtwo\eqstruct\tmfour$.
\item[] (II) Prove the property holds for $\eqstruct$ (\ie\ for many steps of $\tostructsym$) by resorting to (I).
\end{itemize}

The proof of (II) is immediate by induction on the number of $\tostructsym$ steps.
The proof of (I) goes by induction on the rewriting step $\togen$ (that, since $\togen$ is closed by evaluation contexts, becomes a proof by induction on the evaluation context $\evctx$). In principle, we should always consider the two directions of $\tostructsym$. Most of the time, however, one direction is obtained by simply reading the diagram of the other direction bottom-up, instead than top-down; these cases are simply omitted, we distinguish the two directions only when it is relevant.

\begin{enumerate}
\item \label{p:bc} \casealt{Base case 1: multiplicative root step $\tm = \sctxp{\l\var.\tmp}{\tmtwop} \rtodb \sctxp{\tmp\esub{\var}{\tmtwop}} = \tmtwo$}
    
			If the $\tostructsym$ step is internal to $\tmtwop$ or internal to
			one of the substitutions in $\sctx$, the pattern of the $\tostructsym$ redex does
			not overlap with the $\rtodb$ step, and the proof is immediate, the two steps commute. Otherwise, we consider every possible case for $\tostructsym$:

    \begin{enumerate}

    % begin dB vs gc
    \item \caselight{Garbage Collection $\tostructgc$}.
    The garbage collected substitution must
      be one of the substitutions in $\sctx$, \ie\ $\sctx$ must
      be of the form $\sctxtwop{\sctxthree\esub{\vartwo}{\tmthreep}}$. Let $\sctxal \defeq \sctxtwop{\sctxthree}$. Then:
      
\begin{center}
 \begin{tikzpicture}[ocenter]
  \node (s) {\normalsize$\sctxp{\l\var.\tmp}\tmtwop$};
  \node at (s.center)  [below =\nodeVerDist](s2) {\normalsize$\sctxpal{\l\var.\tmp}\tmtwop$};
  \node at (s2.center) [right= \nodeHorDist](t) {\normalsize$\sctxpal{\tmp\esub{\var}{\tmtwop}}$};
  \node at (s-|t) [](s1){\normalsize$\sctxp{\tmp\esub{\var}{\tmtwop}}$};
  \draw[-o] (s) to node {\scriptsize $\db$} (s1);
  \node at (s.center)[below=\nodeVerDist/2](eq1){\normalsize$\tostructgc$};
  \node at (s1.center)[below=\nodeVerDist/2](eq2){\normalsize$\tostructgc$};
\draw[-o, dashed] (s2) to node {\scriptsize $\db$} (t);
\end{tikzpicture} 
\end{center}
    % end dB vs gc

    % begin dB vs dup
    \item \caselight{Duplication $\tostructdup$}.
      The duplicated substitution must
      be one of the substitutions in $\sctx$, \ie\ $\sctx$ must
      be of the form $\sctxtwop{\sctxthree\esub{\vartwo}{\tmthreep}}$. Then:
\begin{center}
 \begin{tikzpicture}[ocenter]
  \node (s) {\normalsize$\sctxtwop{\sctxthreep{\l\var.\tmp}\esub{\vartwo}{\tmthreep}}\tmtwop$};
  \node at (s.center)  [below =\nodeVerDist](s2) {\normalsize$\tm_2$};
  \node at (s2.center) [right= 2*\nodeHorDist](t) {\normalsize$\tm_3$};
  \node at (s-|t) [](s1){\normalsize$\tm_1$};
  \draw[-o] (s) to node {\scriptsize $\db$} (s1);
  \node at (s.center)[below=\nodeVerDist/2](eq1){\normalsize$\tostructdup$};
  \node at (s1.center)[below=\nodeVerDist/2](eq2){\normalsize$\tostructdup$};
\draw[-o, dashed] (s2) to node {\scriptsize $\db$} (t);
\end{tikzpicture} 
\end{center}
where
\begin{align*}
	\tm_1 &:= \sctxtwop{\sctxthreep{\tmp\esub{\var}{\tmtwop}}\esub{\vartwo}{\tmthreep}},\\
	\tm_2 &:= \sctxtwop{\ \varsplit{(\sctxthreep{\l\var.\tmp})}{\vartwo}{\varthree}\esub{\vartwo}{\tmthreep}\esub{\varthree}{\tmthreep}\ }\tmtwop,\\
	\tm_3 &:= \sctxtwop{\ \varsplit{(\sctxthreep{ \tmp\esub{\var}{\tmtwop}})}{\vartwo}{\varthree}\esub{\vartwo}{\tmthreep}\esub{\varthree}{\tmthreep}\ }.
\end{align*}
% \begin{center}
%  \begin{tikzpicture}[ocenter]
%   \node (s) {\normalsize$\sctxtwop{\sctxthreep{\l\var.\tmp}\esub{\vartwo}{\tmthreep}}\tmtwop$};
%   \node at (s.center)  [below =\nodeVerDist](s2) {\normalsize$\sctxtwop{\ \varsplit{(\sctxthreep{\l\var.\tmp})}{\vartwo}{\varthree}\esub{\vartwo}{\tmthreep}\esub{\varthree}{\tmthreep}\ }\tmtwop$};
%   \node at (s2.center) [right= 2*\nodeHorDist](t) {\normalsize$\sctxtwop{\ \varsplit{(\sctxthreep{ \tmp\esub{\var}{\tmtwop}})}{\vartwo}{\varthree}\esub{\vartwo}{\tmthreep}\esub{\varthree}{\tmthreep}\ }$};
%   \node at (s-|t) [](s1){\normalsize$\sctxtwop{\sctxthreep{\tmp\esub{\var}{\tmtwop}}\esub{\vartwo}{\tmthreep}}$};
%   \draw[-o] (s) to node {\scriptsize $\db$} (s1);
%   \node at (s.center)[below=\nodeVerDist/2](eq1){\normalsize$\tostructdup$};
%   \node at (s1.center)[below=\nodeVerDist/2](eq2){\normalsize$\tostructdup$};
% \draw[-o, dashed] (s2) to node {\scriptsize $\db$} (t);
% \end{tikzpicture} 
% \end{center}
    % end dB vs dup

    % begin dB vs ap
    \item \label{p:str-bis-name-base-mul-ap}\caselight{Commutation with application $\tostructap$}.
      Here $\tostructap$ can only be applied in one direction. The diagram is:

\begin{center}
 \begin{tikzpicture}[ocenter]
  \node (s) {\normalsize$\sctxp{\l\vartwo.\tmp}\esub{\var}{\tmfivep}\tmtwop\esub{\var}{\tmfivep}$};
  \node at (s.center)  [below =\nodeVerDist](s2) {\normalsize$\tm_4$};
  \node at (s2.center) [right= 2.4*\nodeHorDist](t) {\normalsize$\tm_2$};
  
  \node at (s-|t) [](s1){\normalsize$\tm_1$};

  \draw[-o] (s) to node {\scriptsize $\db$} (s2);
  \node at (s.center)[right=1.3*\nodeHorDist](eq1){\normalsize$\tostructap$};
  %\node at (s2-|eq1)(eq2){\normalsize$\tostructap$};
\draw[-o, dashed] (s1) to node {\scriptsize $\db$} (t);

\node at (s2.center) [below =\nodeVerDist](s3){\normalsize$=_\alpha$};

\node at (s3-|t) (t2) {\normalsize$\tm_3$};
  
\node at (s3.center) [below =\nodeVerDist](s4){\normalsize$\tm_5$};

\node at (s4-|t) (t3) {\normalsize$\tm_6$};

\node at (t.center)[below=\nodeVerDist/2](eq2){\normalsize$\tostructdup$};
\node at (t2.center)[below=\nodeVerDist/2](eq2){\normalsize$\tostructcom^*$};
\node at (s4.center)[right=1.3*\nodeHorDist](eq1){\normalsize$\tostructes$};
\end{tikzpicture} 
\end{center}
where
\begin{align*}
	\tm_1 &:= (\sctxp{\l\vartwo.\tmp}\tmtwop)\esub{\var}{\tmfivep},\\
	\tm_2 &:= (\sctxp{\tmp\esub{\vartwo}{\tmtwop}})\esub{\var}{\tmfivep},\\
	\tm_3 &:= (\sctxp{\tmp\esub{\vartwo}{\tmtwop\isub{\var}{\vartwo}}})\esub{\var}{\tmfivep}\esub{\vartwo}{\tmfivep},\\
	\tm_4 &:= \sctxp{\tmp\esub{\vartwo}{\tmtwop\esub{\var}{\tmfivep}}}\esub{\var}{\tmfivep},\\
	\tm_5 &:= (\sctxp{\tmp\esub{\vartwo}{\tmtwop\isub{\var}{\vartwo}\esub{\vartwo}{\tmfivep}}})\esub{\var}{\tmfivep},\\
	\tm_6 &:= (\sctxp{\tmp\esub{\vartwo}{\tmtwop\isub{\var}{\vartwo}}\esub{\vartwo}{\tmfivep}})\esub{\var}{\tmfivep}.
\end{align*}

    \item \caselight{Commutation of independent substitutions $\tostructcom$}.
      The substitutions that are commuted by the $\tostructcom$ rule
      must be both in $\sctx$, \ie\ $\sctx$ must be of the form
      $\sctxtwop{\sctxthree\esub{\vartwo}{\tmthreep}\esub{\varthree}{\tmfourp}}$
      with $\varthree \not\in \fv{\tmthreep}$.
      Let $\sctxal = \sctxtwop{\sctxthree\esub{\varthree}{\tmfourp}\esub{\vartwo}{\tmthreep}}$.
      Then:
 \begin{center}
 \begin{tikzpicture}[ocenter]
  \node (s) {\normalsize$\sctxp{\l\var.\tmp}\tmtwop$};
  \node at (s.center)  [below =\nodeVerDist](s2) {\normalsize$\sctxpal{\l\var.\tmp}\tmtwop$};
  \node at (s2.center) [right= \nodeHorDist](t) {\normalsize$\sctxpal{\tmp\esub{\var}{\tmtwop}}$};
  \node at (s-|t) [](s1){\normalsize$\sctxp{\tmp\esub{\var}{\tmtwop}}$};
  \draw[-o] (s) to node {\scriptsize $\db$} (s1);
  \node at (s.center)[below=\nodeVerDist/2](eq1){\normalsize$\tostructcom$};
  \node at (s1.center)[below=\nodeVerDist/2](eq2){\normalsize$\tostructcom$};
\draw[-o, dashed] (s2) to node {\scriptsize $\db$} (t);
\end{tikzpicture} 
\end{center}
    % end dB vs com

    % begin dB vs es
    \item \caselight{Composition of substitutions $\tostructes$}.
      The substitutions that appear in the left-hand side of the $\tostructes$ rule must both
      be in $\sctx$, \ie\ $\sctx$ must be of the form $\sctxtwop{\sctxthree\esub{\vartwo}{\tmthreep}\esub{\varthree}{\tmfourp}}$
      with $\varthree \not\in \fv{\sctxthreep{\l\var.\tmp}}$.
      Let $\sctxal = \sctxtwop{\sctxthree\esub{\vartwo}{\tmthreep\esub{\varthree}{\tmfourp}}}$. Exactly as in the previous case:
      \begin{center}
 \begin{tikzpicture}[ocenter]
  \node (s) {\normalsize$\sctxp{\l\var.\tmp}\tmtwop$};
  \node at (s.center)  [below =\nodeVerDist](s2) {\normalsize$\sctxpal{\l\var.\tmp}\tmtwop$};
  \node at (s2.center) [right= \nodeHorDist](t) {\normalsize$\sctxpal{\tmp\esub{\var}{\tmtwop}}$};
  \node at (s-|t) [](s1){\normalsize$\sctxp{\tmp\esub{\var}{\tmtwop}}$};
  \draw[-o] (s) to node {\scriptsize $\db$} (s1);
  \node at (s.center)[below=\nodeVerDist/2](eq1){\normalsize$\tostructes$};
  \node at (s1.center)[below=\nodeVerDist/2](eq2){\normalsize$\tostructes$};
\draw[-o, dashed] (s2) to node {\scriptsize $\db$} (t);
\end{tikzpicture} 
\end{center}
    % end dB vs es

    \end{enumerate}
  % end dB at empty context
 
  % begin ls at empty context
  \item \label{p:bs-es}\casealt{Base case 2: exponential root step $\tm = \evctxtwop{\var}\esub{\var}{\tmp} \rtowhls \evctxtwop{\tmp}\esub{\var}{\tmp} = \tmtwo$}
			If the $\tostructsym$ step is internal to $\tmp$,
			the proof is immediate, since there is no overlap with
			the pattern of the $\rtols$ redex.
			Similarly, if the $\tostructsym$ step is internal to
			$\evctxp{\var}$, the proof is straightforward by resorting
			to \reflemma{cbname_eqstruct_preserves_shapes}.

			Now we proceed by case analysis on the $\tostructsym$ step:
	
    \begin{enumerate}

    % begin ls vs gc
    \item \caselight{Garbage collection $\tostructgc$}.
    Note that $\tostructgc$ cannot remove $\esub{\var}{\tmp}$, because by hypothesis $\var$ does occur in its scope. If the removed substitution belongs to $\evctxtwo$, \ie\ $\evctxtwo= \evctxpthree{\evctxfour\esub{\vartwo}{\tmtwop}}$. Let $\evctxaltwo\defeq\evctxpthree{\evctxfour}$. Then:
    
       \begin{center}
 \begin{tikzpicture}[ocenter]
  \node (s) {\normalsize$\evctxtwop{\var}\esub{\var}{\tmp}$};
  \node at (s.center)  [below =\nodeVerDist](s2) {\normalsize$\evctxpaltwo{\var}\esub{\var}{\tmp}$};
  \node at (s2.center) [right= \nodeHorDist](t) {\normalsize$\evctxpaltwo{\tmp}\esub{\var}{\tmp}$};
  \node at (s-|t) [](s1){\normalsize$\evctxtwop{\tmp}\esub{\var}{\tmp}$};
  \draw[-o] (s) to node {\scriptsize $\lssym$} (s1);
  \node at (s.center)[below=\nodeVerDist/2](eq1){\normalsize$\tostructgc$};
  \node at (s1.center)[below=\nodeVerDist/2](eq2){\normalsize$\tostructgc$};
\draw[-o, dashed] (s2) to node {\scriptsize $\lssym$} (t);
\end{tikzpicture} 
\end{center}

If $\tostructgc$ adds a substitution as topmost constructor the diagram is analogous.
    % end ls vs gc

    % begin ls vs dup
    \item \caselight{Duplication $\tostructdup$}.
      Two sub-cases: 
      \begin{enumerate}
      \item \caselight{The equivalence $\tostructdup$ acts on a substitution internal to $\evctxtwo$}.
      This case goes as for Garbage collection.
 
      \item \caselight{The equivalence $\tostructdup$ acts on $\esub{\var}{\tmp}$}.
      There are two further sub-cases:
      
            \begin{itemize}
            \item \caselight{The substituted occurrence is renamed by $\tostructdup$}:
\begin{center}
 \begin{tikzpicture}[ocenter]
  \node (s) {\normalsize$\evctxtwop{\var}\esub{\var}{\tmp}$};
  \node at (s.center)  [below =\nodeVerDist](s2) {\normalsize$\varsplit{\evctxtwo}{\var}{\vartwo}\ctxholep{\vartwo}\esub{\var}{\tmp}\esub{\vartwo}{\tmp}$};
  \node at (s2.center) [right= 1.4*\nodeHorDist](t) {\normalsize$\tm_1$};
  \node at (s-|t) [](s1){\normalsize$\evctxtwop{\tmp}\esub{\var}{\tmp}$};
  \draw[-o] (s) to node {\scriptsize $\lssym$} (s1);
  \node at (s.center)[below=\nodeVerDist/2](eq1){\normalsize$\tostructgc$};
  \node at (s1.center)[below=\nodeVerDist/2](eq2){\normalsize$\tostructgc$};
\draw[-o, dashed] (s2) to node {\scriptsize $\lssym$} (t);
\end{tikzpicture} 
\end{center}

where $\tm_1:=\varsplit{\evctxtwo}{\var}{\vartwo}\ctxholep{\tmp}\esub{\var}{\tmp}\esub{\vartwo}{\tmp}$ and $\varsplit{\evctxtwo}{\var}{\vartwo}$ is the context obtained from $\evctxtwo$ by renaming some (possibly none) occurrences of $\var$ as $\vartwo$.

\item \caselight{The substituted occurrence is not renamed by $\tostructdup$}. Essentially as in the previous case:

\begin{center}
 \begin{tikzpicture}[ocenter]
  \node (s) {\normalsize$\evctxtwop{\var}\esub{\var}{\tmp}$};
  \node at (s.center)  [below =\nodeVerDist](s2) {\normalsize$\varsplit{\evctxtwo}{\var}{\vartwo}\ctxholep{\var}\esub{\var}{\tmp}\esub{\vartwo}{\tmp}$};
  \node at (s2.center) [right= 1.4*\nodeHorDist](t) {\normalsize$\tm_1$};
  \node at (s-|t) [](s1){\normalsize$\evctxtwop{\tmp}\esub{\var}{\tmp}$};
  \draw[-o] (s) to node {\scriptsize $\lssym$} (s1);
  \node at (s.center)[below=\nodeVerDist/2](eq1){\normalsize$\tostructdup$};
  \node at (s1.center)[below=\nodeVerDist/2](eq2){\normalsize$\tostructdup$};
\draw[-o, dashed] (s2) to node {\scriptsize $\lssym$} (t);
\end{tikzpicture} 
\end{center}
where $\tm_1:=\varsplit{\evctxtwo}{\var}{\vartwo}\ctxholep{\tmp}\esub{\var}{\tmp}\esub{\vartwo}{\tmp}$.
            \end{itemize}
      \end{enumerate}
    % end ls vs dup
 
    % begin ls vs ap
    \item \label{p:bs-es-ap} \caselight{Commutation with application $\tostructap$}.
      Two sub-cases: 
      \begin{enumerate}
      \item \caselight{The equivalence $\tostructap$ acts on a substitution internal to $\evctxtwo$}.
      This case goes as for Garbage collection.
 
      \item \caselight{The equivalence $\tostructap$ acts on $\esub{\var}{\tmp}$}. It must be the case
            that $\evctxtwo$ is of the form $\evctxthree\tmtwop$. Then:
            
            \begin{center}
 \begin{tikzpicture}[ocenter]
  \node (s) {\normalsize$(\evctxpthree{\var} \tmtwop)\esub{\var}{\tmp}$};
  \node at (s.center)  [below =\nodeVerDist](s2) {\normalsize$\tm_2$};
  \node at (s2.center) [right= 1.4*\nodeHorDist](t) {\normalsize$\tm_3$};
  \node at (s-|t) [](s1){\normalsize$\tm_1$};
  \draw[-o] (s) to node {\scriptsize $\lssym$} (s1);
  \node at (s.center)[below=\nodeVerDist/2](eq1){\normalsize$\tostructap$};
  \node at (s1.center)[below=\nodeVerDist/2](eq2){\normalsize$\tostructap$};
\draw[-o, dashed] (s2) to node {\scriptsize $\lssym$} (t);
\end{tikzpicture} 
\end{center}
where
\begin{align*}
	\tm_1 &:= (\evctxpthree{\tmp} \tmtwop)\esub{\var}{\tmp},\\
	\tm_2 &:= \evctxpthree{\var}\esub{\var}{\tmp} \tmtwop\esub{\var}{\tmp},\\
	\tm_3 &:= \evctxpthree{\tmp}\esub{\var}{\tmp} \tmtwop\esub{\var}{\tmp}.
\end{align*}

      \end{enumerate}
    % end ls vs ap

    % begin ls vs com
    \item \caselight{Commutation of independent substitutions $\tostructcom$}.
    Two sub-cases: 
      \begin{enumerate}
      \item \caselight{The equivalence $\tostructcom$ acts on two substitutions internal to $\evctxtwo$}.
      This case goes as for Garbage collection.
 
      \item \label{p:bc-es-comsc2}\caselight{The equivalence $\tostructcom$ acts on $\esub{\var}{\tmp}$}. It must be the case
            that $\evctxtwo$ is of the form $\evctxthree$. Then:
            
            \begin{center}
 \begin{tikzpicture}[ocenter]
  \node (s) {\normalsize$\evctxpthree{\var} \esub{\vartwo}{\tmtwop}\esub{\var}{\tmp}$};
  \node at (s.center)  [below =\nodeVerDist](s2) {\normalsize$\evctxpthree{\var}\esub{\var}{\tmp} \esub{\vartwo}{\tmtwop}$};
  \node at (s2.center) [right= 1.4*\nodeHorDist](t) {\normalsize$\evctxpthree{\tmp}\esub{\var}{\tmp} \esub{\vartwo}{\tmtwop}$};
  \node at (s-|t) [](s1){\normalsize$\evctxpthree{\tmp} \esub{\vartwo}{\tmtwop}\esub{\var}{\tmp}$};
  \draw[-o] (s) to node {\scriptsize $\lssym$} (s1);
  \node at (s.center)[below=\nodeVerDist/2](eq1){\normalsize$\tostructcom$};
  \node at (s1.center)[below=\nodeVerDist/2](eq2){\normalsize$\tostructcom$};
\draw[-o, dashed] (s2) to node {\scriptsize $\lssym$} (t);
\end{tikzpicture} 
\end{center}
      \end{enumerate}
    % end ls vs com

    % begin ls vs es
    \item \caselight{Composition of substitutions $\tostructes$}. Two sub-cases:
    	\begin{enumerate}
				\item \caselight{The equivalence $\tostructes$ acts on two substitutions internal to $\evctxtwo$}. This case goes as for Garbage collection.
				
				\item \label{p:bc-es-sb2}\caselight{The equivalence $\tostructes$ acts on $\esub{\var}{\tmp}$}.
				Note that the equivalence $\tostructes$ cannot be applied from left to right to $\esub{\var}{\tmp}$, because $\evctxtwop{\var}$ must be of the form $\evctxpthree{\var}\esub{\vartwo}{\tmtwop}$ with $\var\notin \fv{\evctxpthree{\var}}$, which is clearly not possible. It can be applied from right to left. The diagram is:
				
				\begin{center}
\begin{tikzpicture}[ocenter]
  \node (s) {\normalsize$\evctxtwop\var\esub\var{\tmp\esub\vartwo\tmtwo}$};
  \node at (s.center)  [below =\nodeVerDist](s2) {\normalsize$\tm_4$};
  \node at (s2.center) [right= 1.5*\nodeHorDist](t) {\normalsize$\tm_2$};
  
  \node at (s-|t) [](s1){\normalsize$\tm_1$};

  \draw[-o] (s) to node {\scriptsize $\lssym$} (s2);
  \node at (s.center)[right=\nodeHorDist](eq1){\normalsize$\tostructap$};
  %\node at (s2-|eq1)(eq2){\normalsize$\tostructap$};
\draw[-o, dashed] (s1) to node {\scriptsize $\lssym$} (t);

\node at (s2.center) [below =\nodeVerDist](s3){\normalsize$\tm_5$};

\node at (s3-|t) (t2) {\normalsize$\tm_3$};
  
\node at (s3.center) [below =\nodeVerDist](s4){\normalsize$\tm_6$};

\node at (s4-|t) (t3) {\normalsize$\tm_7$};

\node at (t.center)[below=\nodeVerDist/2](eq2){\normalsize$\tostructdup$};
\node at (t2.center)[below=\nodeVerDist/2](eq3){\normalsize$\tostructcom$};
\node at (s4.center)[right=\nodeHorDist](eq1){\normalsize$=_\alpha$};
\node at (s2.center)[below=\nodeVerDist/2](eq4){\normalsize$\eqstruct\mbox{ by \reflemma{eqstruct-and-ctx-name} } $};
\node at (s3.center)[below=\nodeVerDist/2](eq5){\normalsize$\tostructes$};
\end{tikzpicture} 
\end{center}
where
\begin{align*}
	\tm_1 &:= \evctxtwop\var \esub\var\tmp \esub\vartwo\tmtwo,\\
	\tm_2 &:= \evctxtwop\tmp \esub\var\tmp \esub\vartwo\tmtwo,\\
	\tm_3 &:= \evctxtwop{\tmp\isub\vartwo\varthree}  \esub\var\tmp \esub\varthree\tmtwo \esub\vartwo\tmtwo,\\
	\tm_4 &:= \evctxtwop{\tmp\esub\vartwo\tmtwo}\esub\var{\tmp\esub\vartwo\tmtwo},\\
	\tm_5 &:= \evctxtwop\tmp \esub\vartwo\tmtwo \esub\var{\tmp\esub\vartwo\tmtwo},\\
	\tm_6 &:= \evctxtwop\tmp \esub\vartwo\tmtwo \esub\var\tmp \esub\vartwo\tmtwo,\\
	\tm_7 &:= \evctxtwop{\tmp\isub\vartwo\varthree}  \esub\varthree\tmtwo \esub\var\tmp  \esub\vartwo\tmtwo.
\end{align*}
			\end{enumerate}
    \end{enumerate}
  % end ls at empty context

  % begin context / application
  \item \label{p:ic1}\casealt{Inductive case 1: left of an application $\evctx = \evctxtwo\tmfive$}
    The situation is:
    $$\tm = \tmp\tmfive \towhl_{a} \tmtwop\tmfive = \tmtwo$$
    for terms $\tmp, \tmtwop$ such that either $\tmp \towhldb \tmtwop$ or
    $\tmp \towhlls \tmtwop$. Two sub-cases:
    \begin{enumerate}
    \item \caselight{The $\tm \tostructsym \tmthree$ step is internal to $\tmp$}. The proof simply uses the \ih\ applied to the (strictly smaller) evaluation context of  the step $\tmp\towhl_{a}\tmtwop$.

    \item \label{p:ic1sc2}\caselight{The $\tm \tostructsym \tmthree$ step involves the topmost application}. The $\tostructsym$ step can only be a commutation with the root application.
          Moreover, for $\tmp\tmfive$ to match with the right-hand
          side of the $\tostructap$ rule,
          $\tmp$ must have the form $\tmthreep\esub{\var}{\tmfourp}$ and $\tmfive$ the form $\tmfivep\esub{\var}{\tmfourp}$, so that the $\tostructsym$ is:
          $$\tmthree=(\tmthreep\tmfivep)\esub{\var}{\tmfourp}\tostructap\tmthreep\esub{\var}{\tmfourp}\tmfivep\esub{\var}{\tmfourp}=\tm$$
          Three sub-cases:
          
          \begin{enumerate}
          \item \caselight{The rewriting step is internal to $\tmthreep$}. Then the two steps trivially commute. Let $a\in\set{\db,\lssym}$:
          \begin{center}
 \begin{tikzpicture}[ocenter]
  \node (s) {\normalsize$\tmthreep\esub{\var}{\tmfourp}\tmfivep\esub{\var}{\tmfourp}$};
  \node at (s.center)  [below =\nodeVerDist](s2) {\normalsize$(\tmthreep\tmfivep)\esub{\var}{\tmfourp}$};
  \node at (s2.center) [right= 1.4*\nodeHorDist](t) {\normalsize$(\tmthreepp\tmfivep)\esub{\var}{\tmfourp}$};
  \node at (s-|t) [](s1){\normalsize$\tmthreepp\esub{\var}{\tmfourp}\tmfivep\esub{\var}{\tmfourp}$};
  \draw[-o] (s) to node {\scriptsize $a$} (s1);
  \node at (s.center)[below=\nodeVerDist/2](eq1){\normalsize$\tostructap$};
  \node at (s1.center)[below=\nodeVerDist/2](eq2){\normalsize$\tostructap$};
\draw[-o, dashed] (s2) to node {\scriptsize $a$} (t);
\end{tikzpicture} 
\end{center}

\item \caselight{$\db$-step not internal to $\tmthreep$}. Exactly as the multiplicative root case \ref{p:str-bis-name-base-mul-ap} (read in the other direction).
          \end{enumerate}
          
\item \caselight{$\lssym$-step not internal to $\tmthreep$}. Not possible: the topmost constructor is an application, consequently any $\towhlls$ has to take place in $\tmthreep$.
    \end{enumerate}
  % end context / application

  % begin context / substitution
  \item 
  \casealt{Inductive case 2: left of a substitution $\evctx = \evctxtwo\esub{\var}{\tmfive}$}
    The situation is:
    $$\tm = \tmp\esub{\var}{\tmfive} \towhl \tmtwop\esub{\var}{\tmfive} = \tmtwo$$
    with $\tmp=\evctxtwop{\tmpp}$. If the $\tostructsym$ step is internal to $\evctxtwop{\tmp}$,
    the proof we conclude using the \ih. Otherwise:
    \begin{enumerate}

    % begin closure by substitutions vs gc
    \item \caselight{Garbage Collection $\tostructgc$}. If the garbage collected substitution is $\esub{\var}{\tmfive}$ then:      
\begin{center}
 \begin{tikzpicture}[ocenter]
  \node (s) {\normalsize$\tmp\esub{\var}{\tmfive}$};
  \node at (s.center)  [below =\nodeVerDist](s2) {\normalsize$\tmp$};
  \node at (s2.center) [right= \nodeHorDist](t) {\normalsize$\tmtwop$};
  \node at (s-|t) [](s1){\normalsize$\tmtwop\esub{\var}{\tmfive}$};
  \draw[-o] (s) to node {} (s1);
  \node at (s.center)[below=\nodeVerDist/2](eq1){\normalsize$\tostructgc$};
  \node at (s1.center)[below=\nodeVerDist/2](eq2){\normalsize$\tostructgc$};
\draw[-o, dashed] (s2) to node {} (t);
\end{tikzpicture} 
\end{center}
If the substitution is introduced out of the blue, \ie\ $\tmp\esub{\var}{\tmfive}\tostructgc \tmp\esub{\var}{\tmfive}\esub\vartwo\tmfivep$ or $\tmp\esub{\var}{\tmfive}\tostructgc \tmp\esub\vartwo\tmfivep\esub{\var}{\tmfive}$ the diagram is analogous.
    % end closure by substitutions vs gc

    % begin closure by substitutions vs dup
    \item \caselight{Duplication $\tostructdup$}. If the duplicated substitution is $\esub{\var}{\tmfive}$ then:      
\begin{center}
 \begin{tikzpicture}[ocenter]
  \node (s) {\normalsize$\tmp\esub{\var}{\tmfive}$};
  \node at (s.center)  [below =\nodeVerDist](s2) {\normalsize$\varsplit{\tmp}{\var}{\vartwo}\esub{\var}{\tmfive}\esub{\vartwo}{\tmfive}$};
  \node at (s2.center) [right= \nodeHorDist](t) {\normalsize$\varsplit{\tmtwop}{\var}{\vartwo}\esub{\var}{\tmfive}$};
  \node at (s-|t) [](s1){\normalsize$\tmtwop\esub{\var}{\tmfive}$};
  \draw[-o] (s) to node {} (s1);
  \node at (s.center)[below=\nodeVerDist/2](eq1){\normalsize$\tostructdup$};
  \node at (s1.center)[below=\nodeVerDist/2](eq2){\normalsize$\tostructdup$};
\draw[-o, dashed] (s2) to node {} (t);
\end{tikzpicture} 
\end{center}
If duplication is applied in the other direction, \ie\ $\tmp=\tmpp\esub\vartwo\tmfive$ and 
$$\tmp\esub{\var}{\tmfive}=\tmpp\esub\vartwo\tmfive\esub{\var}{\tmfive}\tostructdup \tmpp\isub\vartwo\var\esub{\var}{\tmfive}= \tmp\esub{\var}{\tmfive}$$
the interesting case is when $\tmpp=\evctxthreep\vartwo$ and the step is exponential:
\begin{center}
 \begin{tikzpicture}[ocenter]
  \node (s) {\normalsize$\evctxthreep\vartwo\esub\vartwo\tmfive\esub{\var}{\tmfive}$};
  \node at (s.center)  [below =\nodeVerDist](s2) {\normalsize$\evctxthreep\var\isub\vartwo\var \esub{\var}{\tmfive}$};
  \node at (s2.center) [right= \nodeHorDist](t) {\normalsize$\evctxthreep\tmfive\isub\vartwo\var \esub{\var}{\tmfive}$};
  \node at (s-|t) [](s1){\normalsize$\evctxthreep\tmfive\esub\vartwo\tmfive\esub{\var}{\tmfive}$};
  \draw[-o] (s) to node {$\lssym$} (s1);
  \node at (s.center)[below=\nodeVerDist/2](eq1){\normalsize$\tostructdup$};
  \node at (s1.center)[below=\nodeVerDist/2](eq2){\normalsize$\tostructdup$};
\draw[-o, dashed] (s2) to node {$\lssym$} (t);
\end{tikzpicture} 
\end{center}
If $\tmp$ is $\evctxthreep\var$ it is an already treated base case and if $\tmp$ has another form the rewriting step does not interact with the duplication, and so they simply commute.
    % end closure by substitutions vs dup

    % begin closure by substitutions vs ap
    \item \caselight{Commutation with application $\tostructap$}. Then $\tmp=\tmpp\tmtwopp$. Three sub-cases: 
    
    \begin{enumerate}
    \item \label{p:ic2ap-sc1} \caselight{The $\towhl$ step is internal to $\tmpp$}. Then:
    \begin{center}
 \begin{tikzpicture}[ocenter]
  \node (s) {\normalsize$(\tmpp\tmtwopp)\esub\var\tmfive$};
  \node at (s.center)  [below =\nodeVerDist](s2) {\normalsize$\tmpp\esub\var\tmfive\tmtwopp\esub\var\tmfive$};
  \node at (s2.center) [right= \nodeHorDist](t) {\normalsize$\tmpp'\esub\var\tmfive\tmtwopp\esub\var\tmfive$};
  \node at (s-|t) [](s1){\normalsize$(\tmpp'\tmtwopp)\esub\var\tmfive$};
  \draw[-o] (s) to node {} (s1);
  \node at (s.center)[below=\nodeVerDist/2](eq1){\normalsize$\tostructap$};
  \node at (s1.center)[below=\nodeVerDist/2](eq2){\normalsize$\tostructap$};
\draw[-o, dashed] (s2) to node {} (t);
\end{tikzpicture} 
\end{center}

	\item \caselight{The $\towhl$ step is a multiplicative step}. If $\tmpp=\sctxp{\l\vartwo.\tmpp'}$ then it goes like the diagram of the multiplicative root case \ref{p:str-bis-name-base-mul-ap} (read in the other direction).

	\item \caselight{The $\towhl$ step is an exponential step}. Then it must be $\esub\var\tmfive$ that substitutes on the head variable, but this case has already been treated as a base case (case \ref{p:bs-es-ap}).
    \end{enumerate}

    % end closure by substitutions vs ap

    % begin closure by substitutions vs com
    \item \caselight{Commutation of independent substitutions $\tostructcom$}. It must be $\tmp=\tmpp\esub\vartwo\tmfivep$ with $\var\notin\fv{\tmfivep}$, so that $\tmpp\esub\vartwo\tmfivep\esub\var\tmfive\tostructcom \tmpp\esub\var\tmfive\esub\vartwo\tmfivep$. Three sub-cases:
    \begin{enumerate}
    \item \caselight{Reduction takes place in $\tmpp$}. Then reduction and the equivalence simply commute, as in case \ref{p:ic2ap-sc1}.
    \item \caselight{Exponential steps involving $\esub\var\tmfive$}. This is an already treated base case (case \ref{p:bc-es-comsc2}).
    \item \caselight{Exponential step involving $\esub\vartwo\tmfivep$}. This case is solved reading bottom-up the diagram of case \ref{p:bc-es-comsc2}.
    \end{enumerate}
    % end closure by substitutions vs com

    % begin closure by substitutions vs es
    \item \caselight{Composition of substitutions $\tostructes$}. It must be $\tmp=\tmpp\esub\vartwo\tmfivep$ with $\var\notin\fv{\tmpp}$, so that $\tmpp\esub\vartwo\tmfivep\esub\var\tmfive\tostructes \tmpp\esub\var{\tmfive\esub\vartwo\tmfivep}$. Three sub-cases:
    \begin{enumerate}
    \item \caselight{Reduction takes place in $\tmpp$}. Then reduction and the equivalence simply commute, as in case \ref{p:ic2ap-sc1}.
    \item \caselight{Exponential steps involving $\esub\var\tmfive$}. This case is solved reading bottom-up the diagram of case \ref{p:bc-es-sb2}.
    \item \caselight{Exponential step involving $\esub\vartwo\tmfivep$}. Impossible, because by hypothesis $\var\notin\fv{\tmpp}$.
    \end{enumerate}
    \end{enumerate}
  % end context / substitution

\end{enumerate}

%%%% Detailed proof of part (II)
%To finish the proof of Prop.~\ref{prop:strong-bis}, we are left to show item (II).
%Let $\tm \towhl_a \tmtwo$ and $\tm \eqstruct \tmthree$. Since $\eqstruct$ is the
%transitive closure of $\tostructsym$, there exists $n \in \nat$ such that
%$\tm \tostructsym^n \tmthree$. We prove that there exists
%a term $\tmfour$ such that $\tmtwo \eqstruct \tmfour$ and $\tmthree \towhl_a \tmfour$
%by induction on $n$.
%\begin{itemize}
%\item The base case $n = 0$ is trivial taking $\tmfour = \tmtwo$.
%\item Now suppose $\tm \tostructsym^n \tmp \tostructsym \tmthree$.
%      By \ih\ there exists a term $\tmtwop$ such that
%      $\tmp \towhl_a \tmtwop \eqstruct \tmtwo$.
%      By applying item (I) on $\tmthree \tostructsym \tmp \towhl_a \tmtwop$,
%      we conclude there exists a term $\tmfour$ such that
%      $\tmthree \towhl_a \tmfour \eqstruct \tmtwop \eqstruct \tmtwo$,
%      as desired.
%\end{itemize}
%%%%

%%%%%%%%%%%%%%%%%%%%%%%%%%%%%%%%%%%%%%%%%%%%%%%%%%%%%%%%%%%%%%%%%%%%%%%%%%%%%%%%

\subsection{Proof of \refprop{strong-bis} ($\eqstruct$ is a Strong Bisimulation) for Left-to-Right Call-by-Value}
\renewcommand{\evctx}{\cbvctx}
\renewcommand{\evctxtwo}{\cbvctxtwo}
\renewcommand{\evctxthree}{\cbvctxthree}
\renewcommand{\evctxp}[1]{\cbvctxp{#1}}
\renewcommand{\evctxtwop}[1]{\cbvctxtwop{#1}}
\renewcommand{\evctxthreep}[1]{\cbvctxthreep{#1}}
\label{sect:ProofNonStrictCbVStrongBis}
% !TEX root = ../main.tex
\newcommand{\towhlcekannot}[1]{\stackrel{\mathtt{ns}}{\multimap}_{#1}}
\renewcommand{\towhlcekannot}[1]{\multimap_{#1}}

We follow the structure of the proof in \refsect{ProofCbNStrongBis} for call-by-name. Structural equivalence for call-by-value is defined exactly in the same way.

Before proving the main result, we need the following auxiliary lemmas, proved by straightforward inductions on the contexts.
\reflemma{nonstrict_cbv_eqstruct_preserves_shapes}.\ref{p:nonstrict_cbv_eqstruct_preserves_shapes-two} is the adaptation of \reflemma{cbname_eqstruct_preserves_shapes}
already stated for call-by-name:

\begin{lemma}
\label{l:nonstrict_cbv_eqstruct_preserves_shapes}
The equivalence relation $\tostruct$ preserves the ``shapes'' of $\sctxp{\val}$ and $\evctxp{\var}$. Formally:
\begin{enumerate}
  \item \label{p:nonstrict_cbv_eqstruct_preserves_shapes-one}If $\sctxp{\val} \eqstruct \tm$, then $\tm$ is of the form $\sctxtwop{\valtwo}$.
  \item \label{p:nonstrict_cbv_eqstruct_preserves_shapes-two}If $\evctxp{\var} \eqstruct \tm$, with $\var$ not bound by $\evctx$,
  then $\tm$ is of the form $\evctxtwop{\var}$, with $\var$ not bound by $\evctxtwo$.
\end{enumerate}
\end{lemma}

\begin{lemma}
\label{l:nonstrict_cbv_eqstruct_duplication}
$\sctxp{\tm\esub{\var}{\tmtwo}} \eqstruct \sctxp{\tm\esub{\var}{\sctxp{\tmtwo}}}$
\end{lemma}
\begin{proof}
By induction on $\sctx$. The base case is trivial.
For $\sctx = \sctxtwo\ctxhole\esub{\vartwo}{\tmthree}$, by
\ih\ we have:
$$
\sctxtwop{\tm\esub{\var}{\tmtwo}}\esub{\vartwo}{\tmthree}
\eqstruct
\sctxtwop{\tm\esub{\var}{\sctxtwop{\tmtwo}}}\esub{\vartwo}{\tmthree}
$$
Let $\varsplit{(\sctxtwop{\tmtwo})}{\vartwo}{\varthree}$ be the result
of replacing {\em all} occurrences of 
$\vartwo$ by $\varthree$ in $\sctxtwop{\tmtwo}$. Then:
$$
\begin{array}{ll}
&
\sctxtwop{\tm\esub{\var}{\sctxtwop{\tmtwo}}}\esub{\vartwo}{\tmthree}
\\
\tostructdup &
\sctxtwop{\tm\esub{\var}{\varsplit{(\sctxtwop{\tmtwo})}{\vartwo}{\varthree}}}\esub{\vartwo}{\tmthree}\esub{\varthree}{\tmthree}
\\
\tostructcom^* &
\sctxtwop{\tm\esub{\var}{\varsplit{(\sctxtwop{\tmtwo})}{\vartwo}{\varthree}}\esub{\varthree}{\tmthree}}\esub{\vartwo}{\tmthree}
\\
\tostructes &
\sctxtwop{\tm\esub{\var}{\varsplit{(\sctxtwop{\tmtwo})}{\vartwo}{\varthree}\esub{\varthree}{\tmthree}}}\esub{\vartwo}{\tmthree}
\\
=_{\alpha} &
\sctxtwop{\tm\esub{\var}{\sctxtwop{\tmtwo}\esub{\vartwo}{\tmthree}}}\esub{\vartwo}{\tmthree}
\end{array}
$$
\qed
\end{proof}

Now we prove the strong bisimulation property, by induction on $\togen$.

\begin{enumerate}
\item \casealt{Base case 1: multiplicative root step $
            \tm = \sctxp{\l\var.\tmp}\sctxtwop{\val}
            \rtodbv
            \tmtwo = \sctxp{\tmp\esub{\var}{\sctxtwop{\val}}}
            $}
    The nontrivial cases are when the $\tostructsym$ step overlaps
    the pattern of the $\dbv$-redex. Note that by \reflemma{nonstrict_cbv_eqstruct_preserves_shapes}.\ref{p:nonstrict_cbv_eqstruct_preserves_shapes-one},
    if the $\tostructsym$ is internal to $\sctxtwop{\val}$,
    the proof is direct, since the $\dbv$-redex is preserved. More precisely,
    if $\sctxtwop{\val} \tostructsym \sctxthreep{\valtwo}$, we have:

          $$
          \commutesdbvEEABCD{\tostructsym}{\tostructsym}{
            \sctxp{\l\var.\tmp}\,\sctxtwop{\val}
          }{
            \sctxp{\tmp\esub{\var}{\sctxtwop{\val}}}
          }{
            \sctxp{\l\var.\tmp}\,\sctxthreep{\valtwo}
          }{
            \sctxp{\tmp\esub{\var}{\sctxthreep{\valtwo}}}
          }
          $$
    
    Consider the remaining possibilities for $\tostructsym$:

    \begin{enumerate}

    % dbv vs. gc
    \item \caselight{Garbage collection $\tostructgc$.}
          The garbage collected substitution must be in $\sctx$,
          \ie\ $\sctx$ must be of the form $\sctxOnep{\sctxTwo\esub{\vartwo}{\sctxthreep{\valtwo}}}$
          with $\vartwo \not\in \fv{\sctxTwop{\l\var.\tmp}}$.
          Let $\sctxal := \sctxOnep{\sctxTwo}$. Then:
          $$
          \commutesdbvEEABCD{\tostructgc}{\tostructgc}{
            \sctxp{\l\var.\tmp}\,\sctxtwop{\val}
          }{
            \sctxp{\tmp\esub{\var}{\sctxtwop{\val}}}
          }{
            \sctxpal{\l\var.\tmp}\,\sctxtwop{\val}
          }{
            \sctxpal{\tmp\esub{\var}{\sctxtwop{\val}}}
          }
          $$

    % dbv vs. dup
    \item \caselight{Duplication $\tostructdup$}.
        The duplicated substitution must be in $\sctx$,
        \ie\ $\sctx$ must be of the form $\sctxOnep{\sctxTwo\esub{\vartwo}{\tmtwop}}$.
        Let $\sctxal := \sctxOnep{\ctxhole\esub{\vartwo}{\tmtwop}\esub{\varthree}{\tmtwop}}$.
        Then:
        $$
        \commutesdbvEEABCD{\tostructdup}{\tostructdup}{
            \sctxp{\l\var.\tmp}\,\sctxtwop{\val}
        }{
            \sctxp{\tmp\esub{\var}{\sctxtwop{\val}}}
        }{
            \sctxpal{\varsplit{(\sctxTwop{\l\var.\tmp})}{\vartwo}{\varthree}}\,\sctxtwop{\val}
        }{
            \qquad\tm_1
        }
        $$
        where $\tm_1:=\sctxpal{\varsplit{(\sctxTwop{\tmp\esub{\var}{ \sctxtwop{\val} }})}{\vartwo}{\varthree}}$.

    % dbv vs. ap
    \item \label{p:str-bis-nscbv-base-mul-ap}\caselight{Commutation with application $\tostructap$}.
        The axiom can be applied only in one direction and there must be the same explicit substitution $\esub{\vartwo}{\tmfive}$ as topmost constructor of each of the two sides of the application. The diagram is:
        
        $$
          \commutesdbvEEABCD{\tostructap}{\eqstruct}{
              \sctxp{\l\var.\tmp}\esub{\var}{\tmfive}\,\sctxtwop{\val}\esub{\vartwo}{\tmfive}
          }{
              \tm_1
          }{
              (\sctxp{\l\var.\tmp}\,\sctxtwop{\val})\esub{\vartwo}{\tmfive}
          }{
              \tm_2
          }
          $$
          where
          \begin{align*}
          	\tm_1 &:= \sctxp{\tmp\esub{\var}{ \sctxtwop{\val}\esub{\vartwo}{\tmfive} }}\esub{\var}{\tmfive},\\
          	\tm_2 &:= \sctxp{\tmp\esub{\var}{\sctxtwop{\val}}}\esub{\vartwo}{\tmfive}.
          \end{align*}

          To prove the equivalence on the right, let 
          $\varsplit{\sctxtwop{\val}}{\var}{\varthree}$
          denote the result of replacing {\em all} occurrences of $\var$ by
          a fresh variable $\varthree$ in $\sctxtwop{\val}$.
          The equivalence holds because:
          $$
          \begin{array}{ll}
          &
              \sctxp{\tmpp\esub{\vartwo}{\sctxtwop{\val}}}\esub{\var}{\tmfive}
          \\
          \tostructdup &
              \sctxp{\tmpp\esub{\vartwo}{  \varsplit{\sctxtwop{\val}}{\var}{\varthree}  }}\esub{\var}{\tmfive}\esub{\varthree}{\tmfive}
          \\
          \tostructcom^* &
              \sctxp{\tmpp\esub{\vartwo}{  \varsplit{\sctxtwop{\val}}{\var}{\varthree}  }\esub{\varthree}{\tmfive}}\esub{\var}{\tmfive}
          \\
          \tostructes &
              \sctxp{\tmpp\esub{\vartwo}{  \varsplit{\sctxtwop{\val}}{\var}{\varthree}\esub{\varthree}{\tmfive}  }}\esub{\var}{\tmfive}
          \\
          =_{\alpha} &
              \sctxp{\tmpp\esub{\vartwo}{  \sctxtwop{\val}\esub{\var}{\tmfive}  }}\esub{\var}{\tmfive}
          \end{array}
          $$

    % dbv vs. com
    \item \caselight{Commutation of independent substitutions $\tostructcom$}.
        The commutation of substitutions must be in $\sctx$,
        \ie\ $\sctx$ must be of the form $\sctxOnep{\sctxTwo\esub{\vartwo}{\tmtwop}\esub{\varthree}{\tmthreep}}$
        with $\varthree \not\in \fv{\tmtwop}$.
        Let $\sctxal := \sctxOnep{\sctxTwo\esub{\varthree}{\tmthreep}\esub{\vartwo}{\tmtwop}}$.
        Then:
        $$
        \commutesdbvEEABCD{\tostructcom}{\tostructcom}{
            \sctxp{\l\var.\tmp}\,\sctxtwop{\val}
        }{
            \sctxp{\tmp\esub{\var}{\sctxtwop{\val}}}
        }{
            \sctxpal{\l\var.\tmp}\,\sctxtwop{\val}
        }{
            \sctxpal{\tmp\esub{\var}{\sctxtwop{\val}}}
        }
        $$

    % dbv vs. es
    \item \caselight{Composition of substitutions $\tostructes$}.
        The composition of substitutions must be in $\sctx$,
        \ie\ $\sctx$ must be of the form $\sctxOnep{\sctxTwo\esub{\vartwo}{\tmtwop}\esub{\varthree}{\tmthreep}}$
        with $\varthree \not\in \fv{\sctxTwop{\l\var.\tmp}}$.
        Let $\sctxal := \sctxOnep{\sctxTwo\esub{\vartwo}{\tmtwop\esub{\varthree}{\tmthreep}}}$.
        As in the previous case:
        $$
        \commutesdbvEEABCD{\tostructes}{\tostructes}{
            \sctxp{\l\var.\tmp}\,\sctxtwop{\val}
        }{
            \sctxp{\tmp\esub{\var}{\sctxtwop{\val}}}
        }{
            \sctxpal{\l\var.\tmp}\,\sctxtwop{\val}
        }{
            \sctxpal{\tmp\esub{\var}{\sctxtwop{\val}}}
        }
        $$
    \end{enumerate}

  \item \casealt{Base case 2: exponential root step $
            \tm = \evctxp{\var}\esub{\var}{\sctxp{\val}}
            \rtolsv
            \tmtwo = \sctxp{\evctxp{\val}\esub{\var}{\val}}
            $}
    Consider first the case when the $\tostructsym$-redex is internal
    to $\evctxp{\var}$. By \reflemma{nonstrict_cbv_eqstruct_preserves_shapes}.\ref{p:nonstrict_cbv_eqstruct_preserves_shapes-two}
    we know $\tostructsym$ preserves the shape of $\evctxp{\var}$,
    \ie\ $\evctxp{\var} \tostructsym \evctxpal{\var}$. Then:
    $$
    \commuteslsvEEABCD{\tostructsym}{\eqstruct}{
        \evctxp{\var}\esub{\var}{\sctxp{\val}}
    }{
        \sctxp{\evctxp{\val}\esub{\var}{\val}}
    }{
        \evctxpal{\var}\esub{\var}{\sctxp{\val}}
    }{
        \sctxp{\evctxpal{\val}\esub{\var}{\val}}
    }
    $$
    If the $\tostructsym$-redex is internal to one of the substitutions in $\sctx$, the proof is straightforward.
    Note that the $\tostructsym$-redex has always a substitution
    at the root. The remaining possibilities are such that 
    substitution is in $\sctx$, or that it is precisely $\esub{\var}{\sctxp{\val}}$.
    Axiom by axiom:
    
    \begin{enumerate}

    % lsv vs. gc
    \item \caselight{Garbage collection $\tostructgc$.}
        If the garbage collected substitution is in $\sctx$,
        let $\sctxal$ be $\sctx$ without such substitution.
        Then:
        $$
        \commuteslsvEEABCD{\tostructgc}{\tostructgc}{
            \evctxp{\var}\esub{\var}{\sctxp{\val}}
        }{
            \sctxp{\evctxp{\val}\esub{\var}{\val}}
        }{
            \evctxp{\var}\esub{\var}{\sctxpal{\val}}
        }{
            \sctxpal{\evctxp{\val}\esub{\var}{\val}}
        }
        $$
        The garbage collected substitution cannot be $\esub{\var}{\sctxp{\val}}$,
        since this would imply $\var \not\in \fv{\evctxp{\var}}$,
        which is a contradiction.

    % lsv vs. dup
    \item \caselight{Duplication $\tostructdup$}.
        If the duplicated substitution is in $\sctx$,
        then $\sctx$ is of the form $\sctxOnep{\sctxTwo\esub{\vartwo}{\tmp}}$.
        Let $\sctxal = \sctxOnep{\esub{\vartwo}{\tmp}\esub{\varthree}{\tmp}}$.
        Then:
        $$
        \commuteslsvEEABCD{\tostructdup}{\tostructdup}{
            \evctxp{\var}\esub{\var}{\sctxp{\val}}
        }{
            \sctxp{\evctxp{\val}\esub{\var}{\val}}
        }{
            \tm_1\quad
        }{
            \qquad\tm_2
        }
        $$
        where
        \begin{align*}
        	\tm_1 &:= \evctxp{\var}\esub{\var}{\sctxpal{\sctxvsTwop{\vartwo}{\varthree}{\varsplit{\val}{\vartwo}{\varthree}}}},\\
        	\tm_2 &:= \sctxpal{\sctxvsTwop{\vartwo}{\varthree}{
            \evctxp{ \varsplit{\val}{\vartwo}{\varthree} }\esub{\var}{ \varsplit{\val}{\vartwo}{\varthree} }
            }}.
        \end{align*}
        If the duplicated substitution is $\esub{\var}{\sctxp{\val}}$, there
        are two possibilities, depending on whether the occurrence of $\var$
        substituted by the $\rtolsv$ step is replaced by the fresh variable
        $\vartwo$, or left untouched. If it is not replaced:
            $$
            \begin{tikzpicture}[ocenter]
             \node (s) {\normalsize\ensuremath{ \evctxp{\var}\esub{\var}{\sctxp{\val}} }};
             \node at (s.center)  [right=2*\nodeHorDist](s1){\normalsize\ensuremath{\sctxp{\evctxp{\val}\esub{\var}{\val}}}};
             \node at (s1.center) [below=\nodeVerDist](t) {\normalsize\ensuremath{\tm_2}};
             \node at (t.center) [below=\nodeVerDist](s4) {\normalsize\ensuremath{ 
             \tm_3}};
             \node at (s4-|s) [](s3) {\normalsize\ensuremath{\tm_4}};

             \node at (s.center)  [below=\nodeVerDist](eq1){\normalsize\ensuremath{\tostructdup}};
             \node at (s1.center) [below=\nodeVerDist/2](eq2){\normalsize\ensuremath{\tostructdup}};
             \node at (t.center) [below=0.4*\nodeVerDist](eq2){\normalsize\ensuremath{\eqstruct \text{ (\reflemma{nonstrict_cbv_eqstruct_duplication})}}};
             \draw[-o] (s) to node {\lsvsym} (s1);
             \draw[-o, dashed] (s3) to node {\lsvsym} (s4);
            \end{tikzpicture} 
            $$
            where
            \begin{align*}
			\tm_2 &:=  
             \sctxp{\varsplit{(\evctxp{\val})}{\var}{\vartwo}\esub{\var}{\val}\esub{\vartwo}{\val}},\\
			\tm_3 &:= \sctxp{\varsplit{(\evctxp{\val})}{\var}{\vartwo}\esub{\var}{\val}\esub{\vartwo}{\sctxp{\val}}},\\
			\tm_4 &:=  \varsplit{(\evctxp{\var})}{\var}{\vartwo}\esub{\var}{\sctxp{\val}}\esub{\vartwo}{\sctxp{\val}}.
            \end{align*}

        If the occurrence of $\var$ substituted by the $\rtolsv$ step
        is replaced by the fresh variable $\vartwo$, the situation is essentially analogous.

    % lsv vs. ap
    \item \caselight{Commutation with application $\tostructap$}.
        The only possibility is that the substitution $\esub{\var}{\sctxp{\val}}$
        is commuted with the outermost application in $\evctxp{\var}$.
        Two cases:
        \begin{enumerate} 
        \item \caselight{The substitution acts on the left of the application, \ie\ $\evctx = \evctxtwo\tmp$.}
            $$
            \begin{tikzpicture}[ocenter]
             \node (s) {\normalsize\ensuremath{ (\evctxtwop{\var}\,\tmp)\esub{\var}{\sctxp{\val}} }};
             \node at (s.center)  [right=1.7*\nodeHorDist](s1){\normalsize\ensuremath{\tm_1}};
             \node at (s1.center) [below=\nodeVerDist](t) {\normalsize\ensuremath{ 
             \tm_2
             }};
             \node at (s.center)  [below=2*\nodeVerDist](s3) {\normalsize\ensuremath{ \tm_3 }};
             \node at (s3.center) [right=1.7*\nodeHorDist](s4) {\normalsize\ensuremath{ \tm_4 }};
             \node at (s.center)  [below=\nodeVerDist](eq1){\normalsize\ensuremath{\tostructap}};
             \node at (s1.center) [below=\nodeVerDist/2](eq2){\normalsize\ensuremath{\tostructap^*}};
             \node at (s1.center) [below=3*\nodeVerDist/2](eq2){\normalsize\ensuremath{\tostructes^*}};
             \draw[-o] (s) to node {\lsvsym} (s1);
             \draw[-o, dashed] (s3) to node {\lsvsym} (s4);
            \end{tikzpicture} 
            $$
            where
            \begin{align*}
				\tm_1 &:= \sctxp{(\evctxtwop{\val}\,\tmp)\esub{\var}{\val}},\\
				\tm_2 &:= \sctxp{\evctxtwop{\val}\esub{\var}{\val}}\sctxp{\tmp\esub{\var}{\val}},\\
				\tm_3 &:= \evctxtwop{\var}\esub{\var}{\sctxp{\val}}\tmp\esub{\var}{\sctxp{\val}},\\
				\tm_4 &:= \sctxp{\evctxtwop{\val}\esub{\var}{\val}}\tmp\esub{\var}{\sctxp{\val}}.
            \end{align*}

        \item \caselight{The substitution acts on the right of the application, \ie\ $\evctx = \sctxtwop{\valtwo}\evctxtwo$.}
            Similar to the previous case:
            $$
            \begin{tikzpicture}[ocenter]
             \node (s) {\normalsize\ensuremath{ (\sctxtwop{\valtwo}\,\evctxtwop{\var})\esub{\var}{\sctxp{\val}} }};
             \node at (s.center)  [right=1.7*\nodeHorDist](s1){\normalsize\ensuremath{ \tm_1 }};
             \node at (s2.center) [right=1.7*\nodeHorDist](t) {\normalsize\ensuremath{
                \tm_2
             }};
             \node at (s.center)  [below=2*\nodeVerDist](s3) {\normalsize\ensuremath{
                \tm_3
             }};
             \node at (s3.center) [right=1.7*\nodeHorDist](s4) {\normalsize\ensuremath{
                \tm_4
             }};
             \node at (s.center)  [below=\nodeVerDist](eq1){\normalsize\ensuremath{\tostructap}};
             \node at (s1.center) [below=\nodeVerDist/2](eq2){\normalsize\ensuremath{\tostructap^*}};
             \node at (s1.center) [below=3*\nodeVerDist/2](eq2){\normalsize\ensuremath{\tostructes^*}};
             \draw[-o] (s) to node {\lsvsym} (s1);
             \draw[-o, dashed] (s3) to node {\lsvsym} (s4);
            \end{tikzpicture} 
            $$
        \end{enumerate}
        where
        \begin{align*}
        	\tm_1 &:= \sctxp{(\sctxtwop{\valtwo}\,\evctxtwop{\val})\esub{\var}{\val}},\\
        	\tm_2 &:= \sctxp{\sctxtwop{\valtwo}\esub{\var}{\val}}\sctxp{\evctxtwop{\val}\esub{\var}{\val}},\\
        	\tm_3 &:= \sctxtwop{\valtwo}\esub{\var}{\sctxp{\val}}\evctxtwop{\var}\esub{\var}{\sctxp{\val}},\\
        	\tm_4 &:= \sctxtwop{\valtwo}\esub{\var}{\sctxp{\val}}\sctxp{\evctxtwop{\val}\esub{\var}{\val}}.
        \end{align*}

    % lsv vs. com
    \item \caselight{Commutation of independent substitutions $\tostructcom$}.
        If the commuted substitutions both belong to $\sctx$,
        let $\sctxal$ be the result of commuting them,
        and the situation is exactly as for Garbage collection.

        The remaining possibility is that 
        $\evctx = \evctxtwo\esub{\vartwo}{\tmp}$ and $\esub{\var}{\sctxp{\val}}$ commutes with $\esub{\vartwo}{\tmp}$ (which implies $\var \not\in \fv{\tmp}$). Then:
        $$
        \commuteslsvEEABCD{\tostructcom}{\tostructcom^*}{
            \evctxtwop{\var}\esub{\vartwo}{\tmp}\esub{\var}{\sctxp{\val}}
        }{
            \sctxp{\evctxtwop{\val}\esub{\vartwo}{\tmp}\esub{\var}{\val}}
        }{
            \evctxtwop{\var}\esub{\var}{\sctxp{\val}}\esub{\vartwo}{\tmp}
        }{
            \sctxp{\evctxtwop{\val}\esub{\var}{\val}}\esub{\vartwo}{\tmp}
        }
        $$

    % lsv vs. es
    \item \caselight{Composition of substitutions $\tostructes$}.
        If the composed substitutions both belong to $\sctx$,
        let $\sctxal$ be the result of composing them,
        and the situation is exactly as for Garbage collection.

        The remaining possibility is that $\esub{\var}{\sctxp{\val}}$ is
        the outermost substitution composed by $\tostructes$.
        This is not possible if the rule is applied from left to right,
        since it would imply that
        $\evctxp\var = \evctxtwop\var\esub{\vartwo}{\tmp}$ with
        $\var \not\in \evctxtwop{\var}$, which is a contradiction.

        Finally, if the $\tostructes$ rule is applied from right to left,
        $\sctx$ is of the form $\sctxtwo\esub{\vartwo}{\tmp}$ and:
        $$
        \commuteslsvEEABCD{\tostructes}{=}{
            \evctxp{\var}\esub{\var}{\sctxtwop{\val}\esub{\vartwo}{\tmp}}
        }{
            \sctxtwop{\evctxp{\val}\esub{\var}{\val}}\esub{\vartwo}{\tmp}
        }{
            \evctxp{\var}\esub{\var}{\sctxtwop{\val}}\esub{\vartwo}{\tmp}
        }{
            \sctxtwop{\evctxp{\var}\esub{\var}{\val}}\esub{\vartwo}{\tmp}
        }
        $$
    \end{enumerate}
		% end of root lsv.
  
\item \casealt{Inductive case 1: left of an application $\evctx = \evctxtwo\tmfive$}
  The situation is:
  $$\tm = \evctxtwop{\tmp}\,\tmfive \towhlcek \evctxtwop{\tmtwop}\,\tmfive = \tmtwo$$
  If the $\tostructsym$ step is internal to $\evctxtwop{\tmp}$, the result follows
  by \ih. The proof is also direct if $\tostructsym$ is internal to $\tmfive$.
  The nontrivial case is when the $\tostructsym$ step overlaps $\evctxtwop{\tmp}$ and $\tmfive$.
  There are two possibilities. The first is trivial: $\tostructgc$ is used to introduce a substitution out of the blue, but this case clearly commutes with reduction.

  The second is that the application is commuted with a substitution
  via the $\tostructap$ rule (applied from right to left). There are two cases:
  \begin{enumerate}
  \item \caselight{The substitution comes from $\tmp$.}
        That is, $\evctxtwo = \ctxhole$ and $\tmp$ has a substitution at its root.
        Then $\tmp$ must be a $\rtolsv$-redex
        $\tmp = \evctxthreep{\var}\esub{\var}{\sctxp{\val}}$.
        Moreover $\tmfive = \tmfivep\esub{\var}{\sctxp{\val}}$.
        We have:
        $$
        \commuteslsvEEABCD{\tostructap}{\eqstruct}{
            \evctxthreep{\var}\esub{\var}{\sctxp{\val}}\,\tmfivep\esub{\var}{\sctxp{\val}}
        }{
            \tm_1
        }{
            \tm_2
        }{
            \tm_3
        }
        $$
        where
        \begin{align*}
        	\tm_1 &:= \sctxp{\evctxthreep{\val}\esub{\var}{\val}}\,\tmfivep\esub{\var}{\sctxp{\val}},\\
        	\tm_2 &:= (\evctxthreep{\var}\,\tmfivep)\esub{\var}{\sctxp{\val}},\\
        	\tm_3 &:= \sctxp{(\evctxthreep{\val}\,\tmfivep)\esub{\var}{\val}}.
        \end{align*}

        For the equivalence on the right note that:
        $$
        \begin{array}{ll}
        &
        \sctxp{\evctxthreep{\val}\esub{\var}{\val}}\,\tmfivep\esub{\var}{\sctxp{\val}}
        \\
        \tostructes^* &
        \sctxp{\evctxthreep{\val}\esub{\var}{\val}}\,\sctxp{\tmfivep\esub{\var}{\val}}
        \\
        \tostructap^* &
        \sctxp{\evctxthreep{\val}\esub{\var}{\val}\,\tmfivep\esub{\var}{\val}}
        \\
        \tostructap &
        \sctxp{(\evctxthreep{\val}\,\tmfivep)\esub{\var}{\val}}
        
        \end{array}
        $$

%%%%%%%%%%%%%%%%%%%%%%%%%%%%%%%%%%%%%%%%%%%%%%%
        
  \item \caselight{The substitution comes from $\evctxtwo$.}
        That is: $\evctxtwo = \evctxthree\esub{\var}{\tmthreep}$.
        Moreover, $\tmfive = \tmfivep\esub{\var}{\tmthreep}$.
        The proof is then straightforward:
        $$
        \commutesredEEABCD{\tostructap}{\tostructap}{
            \evctxthreep{\tmp}\esub{\var}{\tmthreep}\,\tmfivep\esub{\var}{\tmthreep}
        }{
            \tm_1
        }{
            \tm_2
        }{
            \tm_3
        }
        $$
        where
        \begin{align*}
        	\tm_1 &:= \evctxthreep{\tmtwop}\esub{\var}{\tmthreep}\,\tmfivep\esub{\var}{\tmthreep},\\
        	\tm_2 &:= (\evctxthreep{\tmp}\,\tmfivep)\esub{\var}{\tmthreep},\\
        	\tm_3 &:= (\evctxthreep{\tmtwop}\,\tmfivep)\esub{\var}{\tmthreep}.
        \end{align*}

  \end{enumerate}

\item \casealt{Inductive case 2: right of an application $\evctx = \sctxp{\val}\evctxtwo$}
  The situation is:
  $$\tm = \sctxp{\val}\,\evctxtwop{\tmp} \towhlcek \sctxp{\val}\,\evctxtwop{\tmtwop} = \tmtwo$$
  Reasoning as in the previous case (\caselight{left of an application}),
  if the $\tostructsym$ step is internal to $\evctxtwop{\tmp}$, the result follows
  by \ih, and if it is internal to $\sctxp{\val}$, it is straightforward to close
  the diagram by resorting to the fact that $\eqstruct$ preserves the shape of
  $\sctxp{\val}$ (\reflemma{nonstrict_cbv_eqstruct_preserves_shapes}).

  The remaining possibility is that the $\tostructsym$ step overlaps
  both $\sctxp{\val}$ and $\evctxtwop{\tmp}$. As in the previous case,
  this can only be possible
  if $\tostructgc$ introduces a substitution out of the blue, which is a trivial case, or because of
  a \caselight{Commutation with application} rule ($\tostructap$, from right to left).
  This again leaves two possibilities:
  \begin{enumerate}
  \item \caselight{The substitution comes from $\tmp$.}
  That is, $\evctxtwo = \ctxhole$ and $\tmp$ is a $\rtolsv$-redex $\tmp = \evctxthreep{\vartwo}\esub{\vartwo}{\sctxtwop{\valtwo}}$.
  Moreover, $\sctx = \sctxthree\esub{\vartwo}{\sctxtwop{\valtwo}}$. Then: 

        $$
        \commuteslsvEEABCD{\tostructap}{\eqstruct}{
            \sctxthreep{\val}\esub{\vartwo}{\sctxtwop{\valtwo}}\,
            \evctxthreep{\vartwo}\esub{\vartwo}{\sctxtwop{\valtwo}}
        }{
            \tm_1
        }{
            \tm_2
        }{
            \tm_3
        }
        $$
        where
        \begin{align*}
        	\tm_1 &:= \sctxthreep{\val}\esub{\vartwo}{\sctxtwop{\valtwo}}\,
            \sctxtwop{\evctxthreep{\valtwo}\esub{\vartwo}{\valtwo}},\\
         \tm_2 &:= (\sctxthreep{\val}\,
            \evctxthreep{\vartwo})\esub{\vartwo}{\sctxtwop{\valtwo}},\\
            \tm_3 &:= \sctxtwop{(\sctxthreep{\val}\,
            \evctxthreep{\valtwo})\esub{\vartwo}{\valtwo}}.
        \end{align*}

        Exactly as in the previous case, for the equivalence on the right consider:
        $$
        \begin{array}{ll}
        &   
            \sctxthreep{\val}\esub{\vartwo}{\sctxtwop{\valtwo}}\,
            \sctxtwop{\evctxthreep{\valtwo}\esub{\vartwo}{\valtwo}}
        \\
        \tostructes^* &
            \sctxtwop{\sctxthreep{\val}\esub{\vartwo}{\valtwo}}\,
            \sctxtwop{\evctxthreep{\valtwo}\esub{\vartwo}{\valtwo}}
        \\
        \tostructap^* &
            \sctxtwop{\sctxthreep{\val}\esub{\vartwo}{\valtwo}\,
                      \evctxthreep{\valtwo}\esub{\vartwo}{\valtwo}}
        \\
        \tostructap &
            \sctxtwop{(\sctxthreep{\val}\,\evctxthreep{\valtwo})\esub{\vartwo}{\valtwo}}
        \\
        \end{array}
        $$
  \item \caselight{The substitution comes from $\evctxtwo$.}
        That is, $\evctxtwo = \evctxthree\esub{\var}{\tmthreep}$.
        Moreover, $\sctx = \sctxtwo\esub{\var}{\tmthreep}$.
        This case is then straightforward:
        $$
        \commutesredEEABCD{\tostructap}{\tostructap}{
            \sctxtwop{\val}\esub{\var}{\tmthreep}\,\evctxthreep{\tmp}\esub{\var}{\tmthreep}
        }{
            \sctxtwop{\val}\esub{\var}{\tmthreep}\,\evctxthreep{\tmtwop}\esub{\var}{\tmthreep}
        }{
            (\sctxtwop{\val}\,\evctxthreep{\tmp})\esub{\var}{\tmthreep}
        }{
            (\sctxtwop{\val}\,\evctxthreep{\tmtwop})\esub{\var}{\tmthreep}
        }
        $$
  \end{enumerate}

\item \casealt{Inductive case 3: left of a substitution $\evctx = \evctxtwo\esub{\var}{\tmfive}$}
  The situation is:
  $$\tm = \evctxtwop{\tmp}\esub{\var}{\tmfive} \towhlcek \evctxtwop{\tmtwop}\esub{\var}{\tmfive} = \tmtwo$$
  If the $\tostructsym$ step is internal to $\evctxtwop{\tmp}$,
  the result follows by \ih. If it is internal to $\tmfive$, the steps are
  orthogonal, which makes the diagram trivial.
  If the equivalence $\tostructgc$  introduces a substitution out of the blue the steps trivially commute.

  The remaining possibility is that the substitution $\esub{\var}{\tmfive}$
  is involved in the $\tostructsym$ redex.
  By case analysis on the kind of the step $\eqstruct_b$:
  \begin{enumerate}
  \item \caselight{Garbage collection $\tostructgc$.}
    We know $\var \not\in \fv{\evctxtwop{\tmp}}$ and therefore
    also $\var \not\in \fv{\evctxtwop{\tmtwop}}$. We get:
    $$
    \commutesredEEABCD{\tostructgc}{\tostructgc}{
        \evctxtwop{\tmp}\esub{\var}{\tmfive}
    }{
        \evctxtwop{\tmtwop}\esub{\var}{\tmfive}
    }{
        \evctxtwop{\tmp}
    }{
        \evctxtwop{\tmtwop}
    }
    $$
  \item \caselight{Duplication $\tostructdup$}.
    The important fact is that if $\evctxtwop{\tmp} \towhlcek \evctxtwop{\tmtwop}$
    and $\varsplit{\evctxtwop{\tmp}}{\var}{\vartwo}$ denotes the result
    of renaming some (arbitrary) occurrences of $\var$ by $\vartwo$ in $\evctxtwop{\tmp}$,
    then $\varsplit{\evctxtwop{\tmp}}{\var}{\vartwo} \towhlcek \varsplit{\evctxtwop{\tmtwop}}{\var}{\vartwo}$,
    where $\varsplit{\evctxtwop{\tmtwop}}{\var}{\vartwo}$ denotes the result
    of renaming some occurrences of $\var$ by $\vartwo$ in $\evctxtwop{\tmtwop}$.
    By this we conclude:
    $$
    \commutesredEEABCD{\tostructdup}{\tostructdup}{
        \evctxtwop{\tmp}\esub{\var}{\tmfive}
    }{
        \evctxtwop{\tmtwop}\esub{\var}{\tmfive}
    }{
        \varsplit{(\evctxtwop{\tmp})}{\var}{\vartwo}\esub{\var}{\tmfive}\esub{\vartwo}{\tmfive}
    }{
        \varsplit{(\evctxtwop{\tmtwop})}{\var}{\vartwo}\esub{\var}{\tmfive}\esub{\vartwo}{\tmfive}
    }
    $$
  \item \caselight{Commutation with application $\tostructap$}.
    $\evctxtwop{\tmp}$ must be an application. This allows for three
    possibilities:
    \begin{enumerate}
    \item \caselight{The application comes from $\tmp$.}
          That is, $\evctxtwo = \ctxhole$ and $\tmp$ is a $\rtodbv$-redex
          $\tmp = \sctxp{\l\vartwo.\tmpp}\,\sctxtwop{\val}$. The diagram is exactly as for the multiplicative base case \ref{p:str-bis-nscbv-base-mul-ap} (read bottom-up).

    \item \caselight{The application comes from $\evctxtwo$, left case.}
        That is, $\evctxtwo = \evctxthree\,\tmthreep$. This case is direct:
          $$
          \commutesredEEABCD{\tostructap}{\tostructap}{
            (\evctxthreep{\tmp}\,\tmthreep)\esub{\var}{\tmfive}
          }{
            \tm_1
          }{
            \tm_2
          }{
            \tm_3
          }
          $$
          where
          \begin{align*}
          	\tm_1 &:= (\evctxthreep{\tmtwop}\,\tmthreep)\esub{\var}{\tmfive},\\
          	\tm_2 &:= \evctxthreep{\tmp}\esub{\var}{\tmfive}\,\tmthreep\esub{\var}{\tmfive},\\
          	\tm_3 &:= \evctxthreep{\tmtwop}\esub{\var}{\tmfive}\,\tmthreep\esub{\var}{\tmfive}.
          \end{align*}

    \item \caselight{The application comes from $\evctxtwo$, right case.}
         That is, $\evctxtwo = \sctxp{\val}\,\evctxthree$.
         Analogous to the previous case.
    \end{enumerate}

  \item \caselight{Commutation of independent substitutions $\tostructcom$}.
	Since $\evctxtwop{\tmp}$ must have a substitution at the root,
    there are two possibilities:
    \begin{enumerate}
    \item \caselight{The substitution comes from $\tmp$.}
          That is, $\evctxtwo = \ctxhole$ and $\tmp$ is a $\rtolsv$-redex
          $\tmp = \evctxthreep{\vartwo}\esub{\vartwo}{\sctxp{\val}}$,
          with $\var \not\in \fv{\sctxp{\val}}$. Then:
          $$
          \commuteslsvEEABCD{\tostructcom}{\tostructcom^*}{
			\evctxthreep{\vartwo}\esub{\vartwo}{\sctxp{\val}}\esub{\var}{\tmfive}
          }{
			\sctxp{\evctxthreep{\val}\esub{\vartwo}{\val}}\esub{\var}{\tmfive}
          }{
			\evctxthreep{\vartwo}\esub{\var}{\tmfive}\esub{\vartwo}{\sctxp{\val}}
          }{
			\sctxp{\evctxthreep{\val}\esub{\var}{\tmfive}\esub{\vartwo}{\val}}
          }
          $$
    \item \caselight{The substitution comes from $\evctxtwo$.}
		That is, $\evctxtwo = \evctxthree\esub{\vartwo}{\tmthreep}$
		with $\var \not\in \fv{\tmthreep}$. This case is
		direct:
          $$
          \commuteslsvEEABCD{\tostructcom}{\tostructcom}{
			\evctxthreep{\tmp}\esub{\vartwo}{\tmthreep}\esub{\var}{\tmfive}
          }{
			\evctxthreep{\tmtwop}\esub{\vartwo}{\tmthreep}\esub{\var}{\tmfive}
          }{
			\evctxthreep{\tmp}\esub{\var}{\tmfive}\esub{\vartwo}{\tmthreep}
          }{
			\evctxthreep{\tmtwop}\esub{\var}{\tmfive}\esub{\vartwo}{\tmthreep}
          }
          $$
    \end{enumerate}
  \item \caselight{Composition of substitutions $\tostructes$}.
	As in the previous case, there are two possibilities:
    \begin{enumerate}
    \item \caselight{The substitution comes from $\tmp$.}
          That is, $\evctxtwo = \ctxhole$ and $\tmp$ is a $\rtolsv$-redex
          $\tmp = \evctxthreep{\vartwo}\esub{\vartwo}{\sctxp{\val}}$,
          with $\var \not\in \fv{\evctxthreep{\vartwo}}$. Then:
          $$
          \commuteslsvEEABCD{\tostructes}{=}{
			\evctxthreep{\vartwo}\esub{\vartwo}{\sctxp{\val}}\esub{\var}{\tmfive}
          }{
			\sctxp{\evctxthreep{\val}\esub{\vartwo}{\val}}\esub{\var}{\tmfive}
          }{
			\evctxthreep{\vartwo}\esub{\vartwo}{\sctxp{\val}\esub{\var}{\tmfive}}
          }{
			\sctxp{\evctxthreep{\val}\esub{\vartwo}{\val}}\esub{\var}{\tmfive}
          }
          $$
    \item \caselight{The substitution comes from $\evctxtwo$.}
		That is, $\evctxtwo = \evctxthree\esub{\vartwo}{\tmthreep}$
		with $\var \not\in \fv{\evctxthreep{\tmp}}$. The proof for this
	    case is direct:
          $$
          \commutesredEEABCD{\tostructes}{\tostructes}{
			\evctxthreep{\tmp}\esub{\vartwo}{\tmthreep}\esub{\var}{\tmfive}
          }{
			\evctxthreep{\tmtwop}\esub{\vartwo}{\tmthreep}\esub{\var}{\tmfive}
          }{
			\evctxthreep{\tmp}\esub{\vartwo}{\tmthreep\esub{\var}{\tmfive}}
          }{
			\evctxthreep{\tmtwop}\esub{\vartwo}{\tmthreep\esub{\var}{\tmfive}}
          }
          $$
	\end{enumerate}
  \end{enumerate}
  
\end{enumerate}

%%%%%%%%%%%%%%%%%%%%%%%%%%%%%%%%%

\subsection{Proof of \refprop{strong-bis} ($\eqstruct$ is a Strong Bisimulation) for Right-to-Left Call-by-Value}
\renewcommand{\evctx}{\cbvctx}
\renewcommand{\evctxtwo}{\cbvctxtwo}
\renewcommand{\evctxthree}{\cbvctxthree}
\renewcommand{\evctxp}[1]{\cbvctxp{#1}}
\renewcommand{\evctxtwop}[1]{\cbvctxtwop{#1}}
\renewcommand{\evctxthreep}[1]{\cbvctxthreep{#1}}
\label{sect:ProofStrictCbVStrongBis}
The proof is obtained as a minimal variation over the proof for left-to-right call-by-value (previous subsection), and is therefore omitted.

%%%%%%%%%%%%%%%%%%%

\subsection{Proof of \refprop{strong-bis} ($\eqstruct$ is a Strong Bisimulation) for Call-by-Need}
\renewcommand{\evctx}{\cbndctx}
\renewcommand{\evctxtwo}{\cbndctxtwo}
\renewcommand{\evctxthree}{\cbndctxthree}
\renewcommand{\evctxp}[1]{\cbndctxp{#1}}
\renewcommand{\evctxtwop}[1]{\cbndctxtwop{#1}}
\renewcommand{\evctxthreep}[1]{\cbndctxthreep{#1}}
\label{sect:ProofNeedBis}

We need two preliminary lemmas, proved by straightforward inductions on $\evctx$:

\begin{lemma}
\label{l:eqstruct-and-ctx-need}
Let $\tm$ be a term, $\evctx$ be a call-by-need evaluation context not capturing any variable in $\fv{\tm}$, and $\var\notin\fv{\evctxp\vartwo}$. 
Then $\evctxp{\tm\esub\var\tmtwo}\eqstructneed\evctxp\tm\esub\var\tmtwo$.
\end{lemma}

		\begin{lemma}
		\label{l:cbneed_eqstruct_preserves_shapes}
		The equivalence relation $\tostructneed$ 
		preserves the shape of $\evctxp{\var}$. More precisely,
		if $\evctxp{\var} \eqstructneed \tm$, with $\var$ not captured by $\evctx$,
        then $\tm$ is of the form $\evctxtwop{\var}$, with $\var$ not captured by $\evctxtwo$.
		\end{lemma}

% !TEX root = ../main.tex
We follow the structure of the previous proofs of strong bisimulation, in particular the proof is by induction on $\togen$ and to ease the notation we write $\eqstruct$ for $\eqstructneed$. Remember that for call-by-need the definition of the structural equivalence is different, it is the one given only by axioms $\tostructapl$, $\tostructcom$, and $\tostructes$.

\begin{enumerate}
\item \casealt{Base case 1: multiplicative root step $
            \tm = \sctxp{\l\var.\tmp}\tmfive
            \rtodb
            \tmtwo = \sctxp{\tmp\esub \var \tmfive}
            $}
Every application of $\eqstruct$ inside $\tmfive$ or inside one of the substitutions in $\sctx$ trivially commutes with the step. The interesting cases are those where structural equivalence has a critical pair with the step:
  \begin{enumerate}
\item \label{p:str-bis-need-base-db-apl}\caselight{Commutation with left of an application $\tostructapl$}.
		If $\sctx=\sctxtwo\esub\vartwo\tmfour$ then
		
		        $$\commutesdbEEABCD{\tostructapl}{=}{
            \sctxtwop{\l\var.\tmp}\esub\vartwo\tmfour\tmfive
        }{
            \sctxtwop{\tmp\esub\var\tmfive}\esub\vartwo\tmfour
        }{
            (\sctxtwop{\l\var.\tmp}\tmfive)\esub\vartwo\tmfour
        }{
            \sctxtwop{\tmp\esub\var\tmfive}\esub\vartwo\tmfour
        }
        $$

    % end dB vs apl

    % begin dB vs com
    \item \caselight{Commutation of independent substitutions $\tostructcom$}.
      The substitutions that are commuted by the $\tostructcom$ rule
      must be both in $\sctx$, \ie\ $\sctx$ must be of the form
      $\sctxtwop{\sctxthree\esub{\vartwo}{\tmthreep}\esub{\varthree}{\tmfourp}}$
      with $\varthree \not\in \fv{\tmthreep}$.
      Let $\sctxal = \sctxtwop{\sctxthree\esub{\varthree}{\tmfourp}\esub{\vartwo}{\tmthreep}}$.
      Then:
 \begin{center}
 \begin{tikzpicture}[ocenter]
  \node (s) {\normalsize$\sctxp{\l\var.\tmp}\tmtwop$};
  \node at (s.center)  [below =\nodeVerDist](s2) {\normalsize$\sctxpal{\l\var.\tmp}\tmtwop$};
  \node at (s2.center) [right= \nodeHorDist](t) {\normalsize$\sctxpal{\tmp\esub{\var}{\tmtwop}}$};
  \node at (s-|t) [](s1){\normalsize$\sctxp{\tmp\esub{\var}{\tmtwop}}$};
  \draw[-o] (s) to node {\scriptsize $\db$} (s1);
  \node at (s.center)[below=\nodeVerDist/2](eq1){\normalsize$\tostructcom$};
  \node at (s1.center)[below=\nodeVerDist/2](eq2){\normalsize$\tostructcom$};
\draw[-o, dashed] (s2) to node {\scriptsize $\db$} (t);
\end{tikzpicture} 
\end{center}
    % end dB vs com

    % begin dB vs es
    \item \caselight{Composition of substitutions $\tostructes$}.
      The substitutions that appear in the left-hand side of the $\tostructes$ rule must both
      be in $\sctx$, \ie\ $\sctx$ must be of the form $\sctxtwop{\sctxthree\esub{\vartwo}{\tmthreep}\esub{\varthree}{\tmfourp}}$
      with $\varthree \not\in \fv{\sctxthreep{\l\var.\tmp}}$.
      Let $\sctxal = \sctxtwop{\sctxthree\esub{\vartwo}{\tmthreep\esub{\varthree}{\tmfourp}}}$. Exactly as in the previous case:
      \begin{center}
 \begin{tikzpicture}[ocenter]
  \node (s) {\normalsize$\sctxp{\l\var.\tmp}\tmtwop$};
  \node at (s.center)  [below =\nodeVerDist](s2) {\normalsize$\sctxpal{\l\var.\tmp}\tmtwop$};
  \node at (s2.center) [right= \nodeHorDist](t) {\normalsize$\sctxpal{\tmp\esub{\var}{\tmtwop}}$};
  \node at (s-|t) [](s1){\normalsize$\sctxp{\tmp\esub{\var}{\tmtwop}}$};
  \draw[-o] (s) to node {\scriptsize $\db$} (s1);
  \node at (s.center)[below=\nodeVerDist/2](eq1){\normalsize$\tostructes$};
  \node at (s1.center)[below=\nodeVerDist/2](eq2){\normalsize$\tostructes$};
\draw[-o, dashed] (s2) to node {\scriptsize $\db$} (t);
\end{tikzpicture} 
\end{center}
  \end{enumerate}
%% end base multiplicative case

  \item \casealt{Base case 2: exponential root step $
            \tm = \evctxp{\var}\esub{\var}{\sctxp{\val}}
            \rtolsv
            \tmtwo = \sctxp{\evctxp{\val}\esub{\var}{\val}}
            $}
    Consider first the case when the $\tostructsym$-redex is internal
    to $\evctxp{\var}$. By \reflemma{cbneed_eqstruct_preserves_shapes}
    we know $\tostructsym$ preserves the shape of $\evctxp{\var}$,
    \ie\ $\evctxp{\var} \tostructsym \evctxpal{\var}$. Then:
    $$
    \commuteslsvEEABCD{\tostructsym}{\eqstruct}{
        \evctxp{\var}\esub{\var}{\sctxp{\val}}
    }{
        \sctxp{\evctxp{\val}\esub{\var}{\val}}
    }{
        \evctxpal{\var}\esub{\var}{\sctxp{\val}}
    }{
        \sctxp{\evctxpal{\val}\esub{\var}{\val}}
    }
    $$
    If the $\tostructsym$-redex is internal to one of the substitutions in $\sctx$, the proof is straightforward.
    Note that the $\tostructsym$-redex has always a substitution
    at the root. The remaining possibilities are that such
    substitution is in $\sctx$, or that it is precisely $\esub{\var}{\sctxp{\val}}$.
    Axiom by axiom:
    
    \begin{enumerate}

    % lsv vs. apl
    \item \label{p:str-bis-need-exp-apl}\caselight{Commutation with the left of an application $\tostructapl$}.
        The only possibility is that the substitution $\esub{\var}{\sctxp{\val}}$
        is commuted with the outermost application in $\evctxp{\var}$, \ie\ $\evctx = \evctxtwo\tmp$. The diagram is:
        
            $$
    \commuteslsvEEABCD{\tostructapl}{\tostructapl^*}{
        (\evctxtwop{\var}\,\tmp)\esub{\var}{\sctxp{\val}}
    }{
        \sctxp{(\evctxtwop{\val}\,\tmp)\esub{\var}{\val}}
    }{
        \evctxtwop{\var}\esub{\var}{\sctxp{\val}}\tmp
    }{
        \sctxp{\evctxtwop{\val}\esub{\var}{\val}}\,\tmp
    }
    $$

    % lsv vs. com
    \item \caselight{Commutation of independent substitutions $\tostructcom$}. Two sub-cases:
    \begin{enumerate}
        \item \label{p:str-bis-need-exp-com-one} \caselight{The commuted substitutions both belong to $\sctx$}. Let $\sctxal$ be the result of commuting them,
        and the diagram is:
                $$
        \commuteslsvEEABCD{\tostructcom}{\tostructcom}{
            \evctxp{\var}\esub{\var}{\sctxp{\val}}
        }{
            \sctxp{\evctxp{\val}\esub{\var}{\val}}
        }{
            \evctxp{\var}\esub{\var}{\sctxpal{\val}}
        }{
            \sctxpal{\evctxp{\val}\esub{\var}{\val}}
        }
        $$

        \item \label{p:str-bis-base-exp-com-two} \caselight{One of the commuted substitutions is $\esub{\var}{\sctxp{\val}}$}. Then
        $\evctx = \evctxtwo\esub{\vartwo}{\tmp}$ and $\esub{\var}{\sctxp{\val}}$ commutes with $\esub{\vartwo}{\tmp}$ (which implies $\var \not\in \fv{\tmp}$). Then:
        $$
        \commuteslsvEEABCD{\tostructcom}{\tostructcom^*}{
            \evctxtwop{\var}\esub{\vartwo}{\tmp}\esub{\var}{\sctxp{\val}}
        }{
            \sctxp{\evctxtwop{\val}\esub{\vartwo}{\tmp}\esub{\var}{\val}}
        }{
            \evctxtwop{\var}\esub{\var}{\sctxp{\val}}\esub{\vartwo}{\tmp}
        }{
            \sctxp{\evctxtwop{\val}\esub{\var}{\val}}\esub{\vartwo}{\tmp}
        }
        $$
\end{enumerate}
		%end of lsv vs. com

    % lsv vs. es
    \item \caselight{Composition of substitutions $\tostructes$}. Two sub-cases:
    \begin{enumerate}
        \item \caselight{The composed substitutions both belong to $\sctx$}. 
        Analogous to case \ref{p:str-bis-need-exp-com-one}.

				\item \label{p:str-bis-need-exp-comp-two}\caselight{One of the composed subtitutions is $\esub{\var}{\sctxp{\val}}$}.
        This is not possible if the rule is applied from left to right,
        since it would imply that
        $\evctxp\var = \evctxtwop\var\esub{\vartwo}{\tmp}$ with
        $\var \not\in \evctxtwop{\var}$, which is a contradiction.

        Finally, if the $\tostructes$ rule is applied from right to left,
        $\sctx$ is of the form $\sctxtwo\esub{\vartwo}{\tmp}$ and:
        $$
        \commuteslsvEEABCD{\tostructes}{=}{
            \evctxp{\var}\esub{\var}{\sctxtwop{\val}\esub{\vartwo}{\tmp}}
        }{
            \sctxtwop{\evctxp{\val}\esub{\var}{\val}}\esub{\vartwo}{\tmp}
        }{
            \evctxp{\var}\esub{\var}{\sctxtwop{\val}}\esub{\vartwo}{\tmp}
        }{
            \sctxtwop{\evctxp{\var}\esub{\var}{\val}}\esub{\vartwo}{\tmp}
        }
        $$
            \end{enumerate}
            % end of lsv vs. es
    \end{enumerate}
    % end of the base exponential case
  
\item \casealt{Inductive case 1: left of an application $\evctx = \evctxtwo\tmfive$}
  The situation is:
  $$\tm = \evctxtwop{\tmp}\,\tmfive \towhlcek \evctxtwop{\tmtwop}\,\tmfive = \tmtwo$$
  If the $\tostructsym$ step is internal to $\evctxtwop{\tmp}$, the result follows
  by \ih. The proof is also direct if $\tostructsym$ is internal to $\tmfive$.
  The nontrivial cases are those where $\tostructsym$ overlaps $\evctxtwop{\tmp}$ and $\tmfive$. The only possible case is that a substitution commutes with the topmost application via $\tostructapl$ (applied from right to left). There are two cases:
  \begin{enumerate}
  \item \caselight{The substitution comes from $\tmp$.}
        That is, $\evctxtwo = \ctxhole$ and $\tmp$ has a substitution at its root.
        Then $\tmp$ must be a $\rtolsv$-redex
        $\tmp = \evctxthreep{\var}\esub{\var}{\sctxp{\val}}$.
        We have:
        $$
        \commuteslsvEEABCD{\tostructapl}{\tostructapl^*}{
            \evctxthreep{\var}\esub{\var}{\sctxp{\val}}\,\tmfive
        }{
            \sctxp{\evctxthreep{\val}\esub{\var}{\val}}\,\tmfive
        }{
            (\evctxthreep{\var}\,\tmfive)\esub{\var}{\sctxp{\val}}
        }{
            \sctxp{(\evctxthreep{\val}\,\tmfive)\esub{\var}{\val}}
        }
        $$
 
%%%%%%%%%%%%%%%%%%%%%%%%%%%%%%%%%%%%%%%%%%%%%%%
        
  \item \label{p:str-bis-need-ind-leftap-apl-two}\caselight{The substitution comes from $\evctxtwo$.}
        That is: $\evctxtwo = \evctxthree\esub{\var}{\tmthreep}$.
        The proof is then straightforward:
        $$
        \commutesredEEABCD{\tostructapl}{\tostructapl}{
            \evctxthreep{\tmp}\esub{\var}{\tmthreep}\,\tmfive
        }{
            \evctxthreep{\tmtwop}\esub{\var}{\tmthreep}\,\tmfive
        }{
            (\evctxthreep{\tmp}\,\tmfive)\esub{\var}{\tmthreep}
        }{
            (\evctxthreep{\tmtwop}\,\tmfive)\esub{\var}{\tmthreep}
        }
        $$
  \end{enumerate}

\item \casealt{Inductive case 2: left of a substitution $\evctx = \evctxtwo\esub{\var}{\tmfive}$}
  The situation is:
  $$\tm = \evctxtwop{\tmp}\esub{\var}{\tmfive} \towhlcek \evctxtwop{\tmtwop}\esub{\var}{\tmfive} = \tmtwo$$
  If the $\tostructsym$ step is internal to $\evctxtwop{\tmp}$,
  the result follows by \ih. If it is internal to $\tmfive$, the steps are
  orthogonal, which makes the diagram trivial.
  The remaining possibility is that the substitution $\esub{\var}{\tmfive}$
  is involved in the $\tostructsym$ redex.
  By case analysis on the kind of the step $\eqstruct_b$:
  \begin{enumerate}
  \item \caselight{Commutation with the left of an application $\tostructapl$}.
    $\evctxtwop{\tmp}$ must be an application. Two
    possibilities:
    \begin{enumerate}
    \item \caselight{The application comes from $\tmp$.}
          That is, $\evctxtwo = \ctxhole$ and $\tmp$ is a $\rtodb$-redex
          $\tmp = \sctxp{\l\vartwo.\tmpp}\,\tmfour$. This is exactly as the base case \ref{p:str-bis-need-base-db-apl} (read bottom-up).

    \item \caselight{The application comes from $\evctxtwo$, \ie\ $\evctxtwo = \evctxthree\,\tmthreep$}. This is exactly as the inductive case \ref{p:str-bis-need-ind-leftap-apl-two} (read bottom-up).
    \end{enumerate}

  \item \caselight{Commutation of independent substitutions $\tostructcom$}.
	Since $\evctxtwop{\tmp}$ must have a substitution at the root,
    there are two possibilities:
    \begin{enumerate}
    \item \caselight{The substitution comes from $\tmp$.}
          That is, $\evctxtwo = \ctxhole$ and $\tmp$ is a $\rtolsv$-redex
          $\tmp = \evctxthreep{\vartwo}\esub{\vartwo}{\sctxp{\val}}$,
          with $\var \not\in \fv{\sctxp{\val}}$. This case is exactly as the base exponential case \ref{p:str-bis-base-exp-com-two} (read bottom-up).
    \item \caselight{The substitution comes from $\evctxtwo$.}
		That is, $\evctxtwo = \evctxthree\esub{\vartwo}{\tmthreep}$
		with $\var \not\in \fv{\tmthreep}$. The diagram is:
          $$
          \commuteslsvEEABCD{\tostructcom}{\tostructcom}{
			\evctxthreep{\tmp}\esub{\vartwo}{\tmthreep}\esub{\var}{\tmfive}
          }{
			\evctxthreep{\tmtwop}\esub{\vartwo}{\tmthreep}\esub{\var}{\tmfive}
          }{
			\evctxthreep{\tmp}\esub{\var}{\tmfive}\esub{\vartwo}{\tmthreep}
          }{
			\evctxthreep{\tmtwop}\esub{\var}{\tmfive}\esub{\vartwo}{\tmthreep}
          }
          $$
    \end{enumerate}
  \item \caselight{Composition of substitutions $\tostructes$}.
	As in the previous case, there are two possibilities:
    \begin{enumerate}
    \item \caselight{The substitution comes from $\tmp$.}
          That is, $\evctxtwo = \ctxhole$ and $\tmp$ is a $\rtolsv$-redex
          $\tmp = \evctxthreep{\vartwo}\esub{\vartwo}{\sctxp{\val}}$,
          with $\var \not\in \fv{\evctxthreep{\vartwo}}$. 
          This case is exactly as the base exponential case \ref{p:str-bis-need-exp-comp-two} (read bottom-up).

    \item \caselight{The substitution comes from $\evctxtwo$.}
		That is, $\evctxtwo = \evctxthree\esub{\vartwo}{\tmthreep}$
		with $\var \not\in \fv{\evctxthreep{\tmp}}$. The diagram is:
          $$
          \commutesredEEABCD{\tostructes}{\tostructes}{
			\evctxthreep{\tmp}\esub{\vartwo}{\tmthreep}\esub{\var}{\tmfive}
          }{
			\evctxthreep{\tmtwop}\esub{\vartwo}{\tmthreep}\esub{\var}{\tmfive}
          }{
			\evctxthreep{\tmp}\esub{\vartwo}{\tmthreep\esub{\var}{\tmfive}}
          }{
			\evctxthreep{\tmtwop}\esub{\vartwo}{\tmthreep\esub{\var}{\tmfive}}
          }
          $$
	\end{enumerate}
	% end of left-of-sub vs. comp
	
  \end{enumerate}
  % end of left-of-sub
  
  \item \casealt{Inductive case 3: inside a hereditary head substitution $\evctx = \evctxtwop\var\esub{\var}{\evctxthree}$}
  The situation is:
  $$\tm = \evctxtwop\var\esub{\var}{\evctxthreep\tmfive} \togen \evctxtwop\var\esub{\var}{\evctxthreep\tmfivep} = \tmtwo$$
  If $\tostructsym$ is internal to $\evctxtwop\var$ the two steps clearly commutes. If $\tostructsym$ is internal to $\evctxthreep\tmfive$ we conclude using the \ih. The remaining cases are when $\tostructsym$ overlaps with the topmost constructor. Axiom by axiom:
  \begin{enumerate}
  \item \caselight{Commutation with the left of an application $\tostructapl$}. It must be that $\evctxtwop\var=\evctx''''\ctxholep\var\tmfour$ with $\var\notin\fv\tmfour$. Then  the two steps simply commute:
          $$
        \commutesredEEABCD{\tostructapl}{\tostructapl}{
            (\evctx''''\ctxholep\var\tmfour)\esub{\var}{\evctxthreep\tmfive}
        }{
            (\evctx''''\ctxholep\var\tmfour)\esub{\var}{\evctxthreep\tmfivep}
        }{
            \evctx''''\ctxholep\var\esub{\var}{\evctxthreep\tmfive}\tmfour
        }{
            \evctx''''\ctxholep\var\esub{\var}{\evctxthreep\tmfivep}\tmfour
        }
        $$

  \item \caselight{Commutation of independent substitutions $\tostructcom$}. It must be that $\evctxtwop\var=\evctx''''\ctxholep\var\esub\vartwo\tmfour$ with $\var\notin\fv\tmfour$. Then the two steps simply commute:
          $$
        \commutesredEEABCD{\tostructapl}{\tostructapl}{
            \evctx''''\ctxholep\var\esub\vartwo\tmfour\esub{\var}{\evctxthreep\tmfive}
        }{
            \tm_1
        }{
            \tm_2
        }{
            \tm_3
        }
        $$
        where
        \begin{align*}
        	\tm_1 := \evctx''''\ctxholep\var\esub\vartwo\tmfour\esub{\var}{\evctxthreep\tmfivep},\\
        	\tm_2 := \evctx''''\ctxholep\var\esub{\var}{\evctxthreep\tmfive}\esub\vartwo\tmfour,\\
        	\tm_3 := \evctx''''\ctxholep\var\esub{\var}{\evctxthreep\tmfivep}\esub\vartwo\tmfour.
        \end{align*}

  \item \caselight{Composition of substitutions $\tostructes$}. There are various sub-cases
    \begin{enumerate}
    	\item \caselight{$\esub{\var}{\evctxthreep\tmfive}$ enters in a substitution}. It must be that $\evctxtwop\var=\evctx_1\ctxholep\vartwo\esub\vartwo{\evctx_2\ctxholep\var}$ with $\var\notin\fv{\evctx_1\ctxholep\vartwo}$. Then the diagram is:
	        $$
        \commutesredEEABCD{\tostructes}{\tostructes}{
            \evctx_1\ctxholep\vartwo\esub\vartwo{\evctx_2\ctxholep\var}\esub{\var}{\evctxthreep\tmfive}
        }{
            \tm_1
        }{
            \tm_2
        }{
            \tm_3
        }
        $$
        \begin{align*}
        	\tm_1 &:= \evctx_1\ctxholep\vartwo\esub\vartwo{\evctx_2\ctxholep\var}\esub{\var}{\evctxthreep\tmfivep},\\
        	\tm_2 &:= \evctx_1\ctxholep\vartwo\esub\vartwo{\evctx_2\ctxholep\var\esub{\var}{\evctxthreep\tmfive}},\\
        	\tm_3 &:= \evctx_1\ctxholep\vartwo\esub\vartwo{\evctx_2\ctxholep\var\esub{\var}{\evctxthreep\tmfivep}}.
        \end{align*}

			\item \caselight{a substitution pops out of $\esub{\var}{\evctxthreep\tmfive}$}. Two sub-cases:
			\begin{enumerate}
				\item \caselight{The substitution comes from $\evctxthree$}. Then $\evctxthreep\tmfive=\evctx''''\ctxholep\tmfive\esub\vartwo\tmfour$. The diagram is:
				$$\commutesredEEABCD{\tostructes}{\tostructes}{
            \evctxtwop\var\esub{\var}{\evctx''''\ctxholep\tmfive\esub\vartwo\tmfour}
        }{
            \tm_1            
        }{
            \tm_2
        }{
            \tm_3
        }$$
        where
        \begin{align}
        	\tm_1 &:= \evctxtwop\var\esub{\var}{\evctx''''\ctxholep\tmfivep\esub\vartwo\tmfour},\\
        	\tm_2 &:= \evctxtwop\var\esub{\var}{\evctx''''\ctxholep\tmfive}\esub\vartwo\tmfour,\\
        	\tm_3 &:= \evctxtwop\var\esub{\var}{\evctx''''\ctxholep\tmfivep}\esub\vartwo\tmfour.
        \end{align}

				\item \caselight{The substitution comes from $\tmfive$}. Then $\evctxthree=\ctxhole$ and $\tmfive$ is a $\rtolsv$-redex
          $\tmp = \evctx''''\ctxholep\vartwo\esub{\vartwo}{\sctxp{\val}}$ and the diagram is:
   $$\commuteslsvEEABCD{\tostructes}{\tostructes^*}{
            \evctxtwop\var\esub{\var}{ \evctx''''\ctxholep\vartwo\esub{\vartwo}{\sctxp{\val}} }
        }{
            \tm_1            
        }{
            \tm_2
        }{
            \tm_3
        }$$
        where
        \begin{align}
        	\tm_1 &:= \evctxtwop\var\esub{\var}{ \sctxp{\evctx''''\ctxholep\val\esub{\vartwo}{\val}}},\\
        	\tm_2 &:= \evctxtwop\var\esub{\var}{ \evctx''''\ctxholep\vartwo}\esub{\vartwo}{\sctxp{\val}},\\
        	\tm_3 &:= \sctxp{\evctxtwop\var\esub{\var}{ \evctx''''\ctxholep\val}\esub{\vartwo}{\val}}.
        \end{align}
  		\end{enumerate}				
  	\end{enumerate}
  \end{enumerate}
\end{enumerate}

%%%%%%%%%%%%%%%%%%%%%%%%%%%%%%%%%%%%%%%%%%%%%%%%%%%%%%%%%%%%%%%%%%%%%%%%%%%%%%%%

\subsection{Proofs for the LAM}
\label{ss:lam-proofs}

%% INVARIANTS
\begin{proof}[Invariants, \reflemma{lam-prop}]
By induction on the length of the execution leading to $\state$, and straightforward inspection of the transition rules.\qed
\end{proof}

%% SIMULATION THEOREM
% !TEX root = ../main.tex

\proof[Distillation, \refth{lam-simulation}]
\begin{enumerate}
\item \casealt{Commutative 1} We have $\lamstate{\code\codetwo}{\env}{\stack}
\ \tomachaone\
\lamstate{\codetwo}{\env}{\fnst{\code,\env}\cons\stack}$, and:
\begin{center}
$\begin{array}{ccccccccccc}
\decode{\lamstate{\code\codetwo}{\env}{\stack}} &=&
\decstackp{\decenvp{\code\codetwo}} &\tostructap^*&
\decstackp{\decenvp{\code}\decenvp{\codetwo}}&=&
\decode{\lamstate{\codetwo}{\env}{\fnst{\code,\env}\cons\stack}}
\end{array}$
\end{center}
As before, we use that $\stack$ is a right-to-left call-by-value evaluation context,
which enables us to use the $\tostructap$ rule. 

\item \casealt{Commutative 2} We have $\lamstate{\codeval}{\env}{\fnst{\code,\envtwo}\cons\stack}
\ \tomachatwo\
\lamstate{\code}{\envtwo}{\argst{\codeval,\env}\cons\stack}$, and:
\begin{center}
$\begin{array}{ccccccccccc}
\decode{\lamstate{\codeval}{\env}{\fnst{\code,\envtwo}\cons\stack}}&=&
\decstackp{\decenvtwop{\code}\decenvp{\codeval}} &=&
\decode{\lamstate{\code}{\envtwo}{\argst{\codeval,\env}\cons\stack}}
\end{array}$
\end{center}

\item \casealt{Multiplicative} We have 
$\lamstate{\l\var.\code}{\env}{\argst{\clos}\cons\stack}
\ \tomachm\
\lamstate{\code}{\esub{\var}{\clos}\cons\env}{\stack}$, and:
\begin{center}
$\begin{array}{ccccccccccc}
\decode{\lamstate{\l\var.\code}{\env}{\argst{\clos}\cons\stack}} &=&
\decstackp{\decenvp{\l\var.\code}\decclos} &\towhllamdb&
\decstackp{\decenvp{\code\esub{\var}{\decclos}}}% &=&
\end{array}$
\end{center}
which is equal to $\decode{\lamstate{\code}{\esub{\var}{\clos}\cons\env}{\stack}}$.

\item \casealt{Exponential}
Let $\env=\envthree\cons\esub{\var}{(\code,\envtwo)}\cons\envfour$. We have
$\lamstate{\var}{\env}{\stack}
\ \tomache\ 
\lamstate{\code}{\envtwo}{\stack}$, and::

\begin{center}
$\begin{array}{ccccccccccc}
\decode{\lamstate{\var}{\env}{\stack}} &=&
\decstackp{\decenvp{\var}} &\towhllamls&
\decstackp{\decenvfourp{\decenvtwop{\decenvthreep{\code}\esub{\var}{\code}}}}&\tostructgc^*&
\decstackp{\decenvtwop{\code}}&=&
\decode{\lamstate{\code}{\decenvtwo}{\stack}}
\end{array}$
\end{center}
Note that by \reflemma{lam-prop}.\ref{p:lam-prop-three},
$\code$ is an abstraction, and thus we are able to apply
$\towhllamls$. Moreover, by \reflemma{lam-prop}.\ref{p:lam-prop-one},
$\env$ binds variables to closures,
and $\decenvtwop{\code}$ is closed; this allows $\decenvthree$ and
$\decenvfour$ to be garbage collected. For doing so, the
$\tostructgc$ rule must be applied below a right-to-left call-by-value
evaluation context, which follows from
\reflemma{lam-prop}.\ref{p:lam-prop-four}.

\end{enumerate}

 \noindent \emph{Progress.}  Let $\state=\kamstate\code\env\stack$ be a commutative normal form s.t. $\decode\state\togen\tmtwo$. If $\code$ is 
        \begin{itemize}
        	    \item \emph{an application $\codetwo\codethree$}. Then a $\tomachaone$ transition applies and $\state$ is not a commutative normal form, absurd.

        \item \emph{an abstraction $\l\var.\codetwo$}.
          Then $\decode{\state} = \decstackp{\decenvp{\l\var.\codetwo}}$ is not in normal form.
          There can only be a $\towhllamdb$-redex,
          so $\decstack$ must be of the form $\decstackp{\ctxhole\clos}$.
          This implies there is a $\tomachm$ transition from $\state$.\qed

        \item \emph{a variable $\var$}. Then $\decode{\state} = \decstackp{\decenvp{\var}}$ is not in normal form.
          There can only be a $\towhllamls$-redex, and it must involve $\var$,
          thus $\decenv = \decenvfourp{\decenvthree\esub{\var}{\decenvtwop{\codeval}}}$.
          This implies there is a $\tomache$ transition from $\state$.
        \end{itemize}

%         Note by \reflemma{lam-prop} $\stack$ and $\env$ are right-to-left call-by-value
%        evaluation contexts.
%        Since it is in commutative normal form, $\code$ is not an
%        application, so it is either a variable or an abstraction.

%%%%%%%%%%%%%%%%%%%%%%%%%%%%%%%%%%%%%%%%%%%%%%%%%%%%%%%%%%%%%%%%%%%%%%%%%%%%%%%%

\subsection{Proofs for the MAM}
\label{ss:mam-proofs}
% !TEX root = ../main.tex
\proof
Let $\eqmamsym$ be the symmetric and contextual closure of the $\sim$ rule
by which $\eqmam$ is defined. Note $\eqmam$ is the reflexive--transitive closure of
$\eqmamsym$. It suffices to show that the property holds for $\eqmamsym$,
\ie\ that $\tmthree \eqmamsym \towhl \tmtwo$ implies $\tmthree \towhl \eqmam \tmtwo$.
The fact that $\eqmamsym^*$ is a bisimulation then follows by induction
on the number of $\eqmamsym$ steps.

Let $\tmthree \eqmamsym \tm \towhl \tmtwo$. The proof of
$\tmthree \towhl \eqmam \tmtwo$ goes by induction
on the call-by-need context $\evctx$ under which the $\towhl$-redex
in $\tm$ is contracted. Note that since $\tm_1 \sim \tm_2$ determines
a bijection between the redexes of $\tm_1$ and $\tm_2$, it suffices
to check the cases when $\sim$ is applied from left to right (\ie\ $\tm \sim \tmthree$).
For the right-to-left cases, all diagrams can be considered from
bottom to top.

\begin{itemize}
\item \casealt{Base case, \ie\ empty context $\evctx = \ctxhole$}
  Two cases, depending on the $\towhl$ step contracting a
  $\rtodb$ or a $\rtols$ redex:
  \begin{enumerate}
  \item \case{$\tm = \sctxp{\l\var.\tmp}\tmtwop \rtodb \sctxp{\tmp\esub{\var}{\tmtwop}}$}
        There are no $\sim$ redexes in $\tm$,
        since any application in $\tm$ must be either $\tm$ itself
        or below $\l\var$, which is not a call-by-name evaluation
        context.
  \item \case{$\tm = \evctxp{\var}\esub{\var}{\tmp} \rtols \evctxp{\tmp}\esub{\var}{\tmp}$}
        Any $\sim$ redex must be internal to $\evctx$, in the
        sense that $\evctx = \evctxtwop{(\evctxthree\tmtwop)\esub{\vartwo}{\tmthreep}}$
        with $\vartwo \not\in \fv{\tmtwop}$.
        Let $\evctxal = \evctxtwop{\evctxthree\esub{\vartwo}{\tmthreep}\tmtwop}$.
        Then:
        \begin{center}
         \begin{tikzpicture}[ocenter]
          \node (s) {\normalsize$\evctxp{\var}\esub{\var}{\tmp}$};
          \node at (s.center)  [below =\nodeVerDist](s2) {\normalsize$\evctxpal{\var}\esub{\var}{\tmp}$};
          \node at (s2.center) [right= \nodeHorDist](t) {\normalsize$\evctxpal{\tmp}\esub{\var}{\tmp}$};
          \node at (s-|t) [](s1){\normalsize$\evctxp{\tmp}\esub{\var}{\tmp}$};
          \draw[->] (s) to node {\scriptsize $\lssym$} (s1);
          \node at (s.center)[below=\nodeVerDist/2](eq1){\normalsize$\sim$};
          \node at (s1.center)[below=\nodeVerDist/2](eq2){\normalsize$\sim$};
        \draw[->, dashed] (s2) to node {\scriptsize $\lssym$} (t);
        \end{tikzpicture} 
        \end{center}
  \end{enumerate}
\item \casealt{Inductive case $\evctx = \evctxtwo\tmfive$}
  Since the application of $\eqmamsym$ must be internal to $\evctxtwo$,
  the result follows directly by \ih.
\item \casealt{Inductive case $\evctx = \evctxtwo\esub{\var}{\tmfive}$}
  If the $\eqmamsym$ step is internal to $\evctxtwo$, the result
  follows again by applying \ih.
  The remaining possibility is that $\evctxtwop{\tm}$ is an application.
  Here there are two cases:
  \begin{enumerate}
  \item \casealt{$\evctxtwo = \ctxhole$, \ie\ $\sim$ interacts with a redex}
    The redex in question must be a $\db$-redex, since it must have
    an application at the root. The situation is the following, with
    $\var \not\in \fv{\tmtwop}$:
    \begin{center}
     \begin{tikzpicture}[ocenter]
      \node (s) {\normalsize$(\sctxp{\l\vartwo.\tmp}\,\tmtwop)\esub{\var}{\tmfive}$};
      \node at (s.center)  [below=\nodeVerDist](s2) {\normalsize$\sctxp{\l\vartwo.\tmp}\esub{\var}{\tmfive}\,\tmtwop$};
      \node at (s2.center) [right=2*\nodeHorDist](t) {\normalsize$\sctxp{\tmp\esub{\vartwo}{\tmtwop}}\esub{\var}{\tmfive}$};
      \node at (s) [right=2*\nodeHorDist](s1){\normalsize$\sctxp{\tmp\esub{\vartwo}{\tmtwop}}\esub{\var}{\tmfive}$};
      \draw[->] (s) to node {\scriptsize $\db$} (s1);
      \node at (s.center)[below=\nodeVerDist/2](eq1){\normalsize$\sim$};
      \node at (s1.center)[below=\nodeVerDist/2](eq2){\normalsize$=$};
    \draw[->, dashed] (s2) to node {\scriptsize $\db$} (t);
    \end{tikzpicture} 
    \end{center}

  \item \casealt{$\evctxtwo = \evctxthree\tmp$, \ie\ there is no interaction between $\sim$ and a redex}
    This case is straightforward, since the contraction of the
    $\towhl$ redex and the application of $\sim$ are orthogonal.\qed
  \end{enumerate}
  
\end{itemize}

%\begin{lemma}[Progress]
%Every reachable commutative normal form $\state$ such that $\decode{\state}$ is not normal is not a final state.
%\end{lemma}
%\begin{proof}
%Let $\state = \mamstate{\code}{\env}{\stack}$. By hypothesis, $\code$ can only be
%\begin{itemize}
%\item {\em a variable $\var$}:
%by point \ref{p:mam-prop-one} of \reflemma{mam-prop},
%$\var$ must be bound by $\env$. So $\env = \envtwop{\envthree\esub{\var}{\codetwo}}$
%and a $\tokamvar$ transition applies.
%\item {\em an abstraction $\l\var.\codetwo$}:
%the decoding
%$\envp{\stackp{\l\var.\codetwo}}$ must have a multiplicative redex,
%so $\stack$ is not empty and a $\tokamlam$ transition applies.
%\end{itemize}
%\end{proof}

%%%%%%%%%%%%%%%%%%%%%%%%%%%%%%%%%%%%%%%%%%%%

\subsection{Proofs for the Split CEK}
\label{sect:scek-proofs}

% !TEX root = ../main.tex
\proof[Split CEK Distillation, \refth{sCEK-sim}]
\emph{Properties of the decoding}:
\begin{enumerate}
    \item \emph{Commutative 1}.
      We have
      $\scekstate{\code\,\codetwo}{\env}{\stack}{\fstack}
       \tomachaone
       \scekstate{\code}{\env}{(\codetwo,\env)\cons\stack}{\fstack}$,
      and:
      $$
      \begin{array}{llll}
      \decode{\scekstate{\code\,\codetwo}{\env}{\stack}{\fstack}} & = &
      \dfstackp{\dstackp{\denvp{\code\,\codetwo}}} & \tostructap^* \\
      && \dfstackp{\dstackp{\denvp{\code}\,\denvp{\codetwo}}} & = \\
      && \decode{\scekstate{\code}{\env}{(\codetwo,\env)\cons\stack}{\fstack}}
      \end{array}
      $$
    \item \emph{Commutative 2}.
      We have $\scekstate{\codeval}{\env}{(\code,\envtwo)\cons\stack}{\fstack}
               \tomachatwo
               \scekstate{\code}{\envtwo}{\stempty}{((\codeval,\env),\stack)\cons\fstack}$, and:
      $$
      \begin{array}{llll}
      \decode{\scekstate{\codeval}{\env}{(\code,\envtwo)\cons\stack}{\fstack}} & = &
      \dfstackp{\dstackp{\denvp{\val}\,\denvtwop{\code}}} & \tostructgc^* \\
      && \dfstackp{\dstackp{\denvp{\val}\,\denvp{\denvtwop{\code}}}} & \tostructap^* \\
      && \dfstackp{\dstackp{\denvp{\val\,\denvtwop{\code}}}} & = \\
      && \decode{\scekstate{\code}{\envtwo}{\stempty}{((\codeval,\env),\stack)\cons\fstack}}
      \end{array}
      $$
    \item \emph{Multiplicative}.
      We have $\scekstate{\codeval}{\env}{\stempty}{((\l\var.\code,\envtwo),\stack)\cons\fstack}
               \tomachm
               \scekstate{\code}{\esub\var{(\codeval,\env)}\cons\envtwo}{\stack}{\fstack}$, and:
      $$
      \begin{array}{ll}
      \decode{\scekstate{\codeval}{\env}{\stempty}{((\l\var.\code,\envtwo),\stack)\cons\fstack}} & = \\
      \dfstackp{\dstackp{\denvtwop{(\l\var.\tm)\,\denvp{\val}}}} & \tomachm \\
      \dfstackp{\dstackp{\denvtwop{\tm\esub{\var}{\denvp{\val}}}}} & = \\
      \decode{\scekstate{\code}{\esub\var{(\codeval,\env)}\cons\envtwo}{\stack}{\fstack}}
      \end{array}
      $$
    \item \emph{Exponential}.
      We have $\scekstate{\var}{\env_1\cons\esub{\var}{(\codeval,\env)}\cons\env_2}{\stack}{\fstack}
               \tomache
               \scekstate{\codeval}{\env}{\stack}{\fstack}$, and:
      $$
      \begin{array}{ll}
	  \decode{\scekstate{\var}{\env_1\cons\esub{\var}{(\codeval,\env)}\cons\env_2}{\stack}{\fstack}} & = \\
      \dfstackp{\dstackp{\denvsubtwop{\denvsubonep{\var}\esub{\var}{\denvp{\codeval}}}}} & \tomache \\
      \dfstackp{\dstackp{\denvsubtwop{  \denvp{ \denvsubonep{\codeval}\esub{\var}{\codeval} } }}} & \tostructgc^* \\
      \dfstackp{\dstackp{\denvp{\codeval}}} & = \\
      \decode{\scekstate{\codeval}{\env}{\stack}{\fstack}}
      \end{array}
      $$
      We use that $\denvp{\codeval}$ is closed by \reflemma{scek-prop}.\ref{p:scek-prop-closure}
      to ensure that $\decode{\env_1}$, $\decode{\env_2}$, and $\esub{\var}{\codeval}$ can be garbage collected.
  \end{enumerate}

  \emph{Progress}.  Let $\state=\kamstate\code\env\stack$ be a commutative normal form s.t. $\decode\state\togen\tmtwo$. If $\code$ is 
  
	\begin{itemize}
		\item \emph{an application $\codetwo\codethree$}. Then a $\tomachaone$ transition applies and $\state$ is not a commutative normal form, absurd.

  \item  \emph{an abstraction $\val$}.
  The decoding
  $\decode\state = \decode\fstack\ctxholep{\decode\stack\ctxholep{\decode\env\ctxholep{\codeval}}}$
  must have a multiplicative redex, because it must have a redex and $\codeval$ is not a variable. So $\codeval$ is applied to something, \ie\ there must be at least one application node in
  $\decode\fstack\ctxholep{\decode\stack}$. Moreover, the stack $\stack$ must be empty,
  otherwise there would be an administrative $\tomachatwo$ transition,
  contradicting the hypothesis. So $\fstack$ is not empty. Let
  $\fstack = ((\codetwo,\envtwo),\stacktwo)\cons\fstacktwo$.
  By point \ref{p:scek-prop-value} of \reflemma{scek-prop}, $\codetwo$ must be a value,
  and a $\tomachm$ transition applies.
  
		\item \emph{a variable $\var$}. By point \ref{p:scek-prop-closure} of \reflemma{scek-prop},
  $\var$ must be bound by~$\env$, so $\env = \env_1\cons\esub{\var}{(\codetwo,\envtwo)}\cons\env_2$
  and a $\tomache$ transition applies.
  \end{itemize}
\qed
  
%%%%%%%%%%%%%%%%%%%%%%%%%%%%%%%%%%%%

%%%%%%%%%%%%%%%%%%%%%%%%%%%%%%%%%%%%%%%%%%%%

\subsection{Proofs for the Merged WAM}
\label{sect:mwam-proofs}

% !TEX root = ../main.tex
\proof[Distillation, \refth{mwam-sim}]
  \begin{enumerate}
  
    \item \emph{Commutative 1}.
       We have $\mgwamstate{\code\,\codetwo}{\stack}{\genv}
	            \tomachaone
	            \mgwamstate{\code}{\argst{\codetwo}\cons\stack}{\genv}$,
       and:
        $$
        \decode{\mgwamstate{\code\,\codetwo}{\stack}{\genv}} =
        \dgenvp{\dstackp{\code\,\codetwo}} =
        \decode{\mgwamstate{\code}{\argst{\codetwo}\cons\stack}{\genv}}
        $$
  \item \emph{Commutative 2}.
	    We have $\mgwamstate{\var}{\stack}{\genv_1 \cons \esub\var\code \cons \genv_2}
	             \tomachatwo
	             \mgwamstate{\code}{\headst{\genv_1,\var}\cons\stack}{\genv_2}$,
        and:
        $$
        \begin{array}{llll}
        \decode{\mgwamstate{\var}{\stack}{\genv_1 \cons \esub\var\code \cons \genv_2}} & = &
        \dgenvsubtwop{\dgenvsubonep{\dstackp{\var}}\esub{\var}{\code}} & = \\
	    && \decode{\mgwamstate{\code}{\headst{\genv_1,\var}\cons\stack}{\genv_2}}
        \end{array}
        $$
        
  \item \emph{Multiplicative}.
        We have $\mgwamstate{\l\var.\code}{\argst{\codetwo}\cons\stack}{\genv}
                 \tomachm
	             \mgwamstate{\code}{\stack}{\esub\var\codetwo\cons\genv}$,
        and:
        $$
        \begin{array}{llll}
        \decode{\mgwamstate{\l\var.\code}{\argst{\codetwo}\cons\stack}{\genv}} & = &
        \dgenvp{\dstackp{(\l\var.\code)\,\codetwo}} & \towhlcekdb \\
        && \dgenvp{\dstackp{\code\esub{\var}{\codetwo}}} & \eqstructneed \mbox{ Lem. \ref{l:ev-comm-struct}} \\
        && \dgenvp{\dstackp{\code}\esub{\var}{\codetwo}} & = \\
        && \decode{\mgwamstate{\code}{\stack}{\esub{\var}{\codetwo}\cons\genv}}
        \end{array}
        $$
  \item \emph{Exponential}.
        We have $\mgwamstate{\codeval}{\headst{\genv_1,\var}\cons\stack}{\genv_2}
	             \tomache
	             \mgwamstate{\rename{\codeval}}{\stack}{\genv_1 \cons \esub\var\codeval \cons \genv_2}$,
        and:
        $$
        \begin{array}{llll}
        \decode{\mgwamstate{\codeval}{\headst{\genv_1,\var}\cons\stack}{\genv_2}} & = &
        \dgenvsubtwop{\dgenvsubonep{\dstackp{\var}}\esub{\var}{\codeval}} & \towhlcekls \\
        && \dgenvsubtwop{\dgenvsubonep{\dstackp{\codeval}}\esub{\var}{\codeval}} & \alphaequiv \\
        && \dgenvsubtwop{\dgenvsubonep{\dstackp{\rename\codeval}}\esub{\var}{\codeval}} & = \\
        && \decode{\mgwamstate{\rename{\codeval}}{\stack}{\genv_1 \cons \esub\var\codeval \cons \genv_2}}
        \end{array}
        $$
        	\end{enumerate}
	
\noindent    \emph{Progress}. Let $\state=\kamstate\code\stack\genv$ be a commutative normal form s.t. $\decode\state\togen\tmtwo$. If $\code$ is 
    \begin{enumerate}
    \item \emph{an application $\codetwo\codethree$}. Then a $\tomachaone$ transition applies and $\state$ is not a commutative normal form, absurd.

    \item \emph{an abstraction $\val$}. The decoding $\decode\state$ is of the form $\dgenvp{\dstackp{\val}}$.
    The stack $\stack$ cannot be empty,
    since then $\decode\state = \dgenvp{\val}$ would be normal.
    So either the a $\tomache$ or a $\tomachm$ transition
    applies.
    
        \item \emph{a variable $\var$}. By the global closure invariant, $\var$ is bound by $\genv$. Then a $\tomachatwo$ transition applies and $\state$ is not a commutative normal form, absurd.
 \qed

  \end{enumerate}

\subsection{Proofs for the Pointing WAM}
\label{ss:sam-proofs}

% !TEX root = ../main.tex
\begin{proof}[Pointing WAM Invariants, \reflemma{pwam-invariants}]

  By induction on the length of the execution.
  Points \ref{p:pwam-invariants-subterm} and \ref{p:pwam-invariants-value}
  are by direct inspection of the rules.
  Assuming $\genv\dual\dump$, point \ref{p:pwam-invariants-ctx} is immediate by induction on the length of $\fstack$.

  Thus we are only left to check point \ref{p:pwam-invariants-dual}.
  We use point \ref{p:pwam-invariants-value}, \ie\ that substitutions
  in $\genv$ bind pairwise distinct variables. Following we show that
  transitions preserve the invariant:
  \begin{enumerate}

  \item \caselight{Conmutative 1}. We have:
        $$\pwamstate{\code\,\codetwo}{\stack}{\genv}{\fstack} \tomachaone \pwamstate{\code}{\codetwo\cons\stack}{\genv}{\fstack}$$
        Trivial, since the dump and the environment are the same and
        $\decode{(\codetwo\cons\stack)}\ctxholep{\code} = \dstackp{\code\,\codetwo}$.

  \item \caselight{Conmutative 2}. We have $\state \tomachatwo \statetwo$ with:
        $$\state = \pwamstate{\var}{\stack}{\genv_1\cons\esub{\var}{\code}\cons\genv_2}{\fstack}$$
        $$\statetwo = \pwamstate{\code}{\stempty}{\genv_1\cons\esub{\var}{\pwammark}\cons\genv_2}{(\var,\stack)\cons\fstack}$$

        Note that since by \ih\ $\pwamclosed{(\dstackp{\var},(\genv_1\cons\esub{\var}{\code}\cons\genv_2)\envslice)}$
        and $\var$ is free in $\dstackp{\var}$, there cannot be any dumped substitutions in $\genv_2$.
        Then $(\genv_1\cons\esub{\var}{\tm}\cons\genv_2)\envslice = \genv_1\envslice\cons\esub{\var}{\tm}\cons\genv_2$
        and we know:
        \begin{equation} \label{eq:pwam-invariant-var1}
        \pwamclosed{(\dstackp{\var},\genv_1\envslice\cons\esub{\var}{\code}\cons\genv_2)}
        \end{equation}

        For \ref{p:pwam-invariants-dual1},
        note $(\genv_1\cons\esub{\var}{\pwammark}\cons\genv_2)\envslice = \genv_2$.
        Then we must show $\pwamclosed{(\code,\genv_2)}$, which is implied by \eqref{eq:pwam-invariant-var1}.

        \medskip
        For \ref{p:pwam-invariants-dual2}, there are two cases:
        \begin{itemize}
        \item If the pair is $(\var,\stack)$, we must show
              $$\pwamclosed{(\dstackp{\var},(\genv_1\cons\esub{\var}{\pwammark}\cons\genv_2)\envslicevar{\var})} \text{, \ie}$$
              $$\pwamclosed{(\dstackp{\var},\genv_1\envslice\cons\esub{\var}{\pwammark}\cons\genv_2)}$$
              which is implied by \eqref{eq:pwam-invariant-var1}.
        \item If the pair is $(\vartwo,\stacktwo)$ in $\fstack$, with $\vartwo \neq \var$,
              note first that
              $$(\genv_1\cons\esub{\var}{\code}\cons\genv_2)\envslicevar{\vartwo} = 
                \genv_1\envslicevar{\vartwo}\cons\esub{\var}{\code}\cons\genv_2$$
              And similarly for $(\genv_1\cons\esub{\var}{\pwammark}\cons\genv_2)\envslicevar{\vartwo}$.
              Moreover, by the invariant on $\state$ we know
              $$\pwamclosed{(\dstacktwop{\vartwo}, \genv_1\envslicevar{\vartwo}\cons\esub{\var}{\code}\cons\genv_2)}$$
              and this implies
              $$\pwamclosed{(\dstacktwop{\vartwo}, \genv_1\envslicevar{\vartwo}\cons\esub{\var}{\pwammark}\cons\genv_2)}$$
              as required.
        \end{itemize}

        \medskip
        For \ref{p:pwam-invariants-dual3}, we have already observed that $\genv_2$ has no dumped
        substitutions. Then $\esub{\var}{\pwammark}$ is the rightmost dumped substitution in the environment of $\statetwo$,
        while $(\var,\stack)$ is the leftmost pair in the dump.
        We conclude by the fact that the invariant already holds for $\state$.
        
  \item \caselight{Multiplicative, empty dump}.
        We have $\state \tomachm \statetwo$ with:
        $$\state = \pwamstate{\l\var.\code}{\codetwo\cons\stack}{\genv}{\stempty}$$
        $$\statetwo = \pwamstate{\code}{\stack}{\esub{\var}{\codetwo}\cons\genv}{\stempty}$$

        First note that, since the environment and the dump are dual in $\state$,
        there are no dumped substitutions in $\genv$.

        \medskip
        For point \ref{p:pwam-invariants-dual1}, we know that:
        \begin{equation} \label{eq:pwam-invariant-var2a}
        \pwamclosed{(\dstackp{(\l\var.\code)\,\codetwo},\genv)}
        \end{equation}
        and we have to check:
        $$\pwamclosed{(\dstackp{\code},\esub{\var}{\codetwo}\cons\genv)}$$
        Let $\vartwo \in \fv{\dstackp{\code}}$. Then either $\vartwo = \var$, which is bound by $\esub{\var}{\codetwo}$,
        or $\vartwo \in \fv{\dstackp{\l\var.\code}}$, in which case $\vartwo$ is bound by $\genv$.
        Moreover, since $\decode\stack$ is an application context, by
        \eqref{eq:pwam-invariant-var2a} we get $\pwamclosed{(\codetwo,\genv)}$.

        \medskip
        Points \ref{p:pwam-invariants-dual2} and \ref{p:pwam-invariants-dual3} are
        trivial since the dump is empty and the environment has no dumped substitutions.

  \item \caselight{Multiplicative, non-empty dump}.
        We have $\state \tomachm \statetwo$ with:
        $$\state = \pwamstate{\l\var.\code}{\codetwo\cons\stack}{\genv_1\cons\esub{\vartwo}{\pwammark}\cons\genv_2}{(\vartwo,\stacktwo)\cons\fstack}$$
        $$\statetwo = \pwamstate{\code}{\stack}{\genv_1\cons\esub{\vartwo}{\pwammark}\cons\esub{\var}{\codetwo}\cons\genv_2}{(\vartwo,\stacktwo)\cons\fstack}$$

        Note first that since the invariant holds for $\state$, we know $\esub{\vartwo}{\pwammark}$
        is the rightmost dumped substitution in the environment of both $\state$ and $\statetwo$.
        Therefore $(\genv_1\cons\esub{\vartwo}{\pwammark}\cons\genv_2)\envslice = \genv_2$

        \medskip
        For proving point \ref{p:pwam-invariants-dual1}, we have:
        $$\pwamclosed{(\dstackp{(\l\var.\code)\,\codetwo}, \genv_2)}$$
        and we must show:
        $$\pwamclosed{(\dstackp{\code}, \esub{\var}{\codetwo}\cons\genv_2)}$$
        The situation is exactly as in point \ref{p:pwam-invariants-dual1} for the
        $\tomachm$ transition, empty dump case.

        \medskip
        For point \ref{p:pwam-invariants-dual2}, let $(\varthree,\stackthree)$ be any pair
        in $(\vartwo,\stacktwo)\cons\fstack$.  Let also
        $$\genvtwo_1 := \begin{cases}
                         \genv_1\envslice               & \text{if $\vartwo = \varthree$} \\
                         \genv_1\envslicevar{\varthree} & \text{otherwise}
                       \end{cases}$$
        and note that $(\genv_1\cons\esub{\vartwo}{\pwammark}\cons\genv)\envslicevar{\vartwo} = \genvtwo_1\cons\esub{\vartwo}{\pwammark}\cons\genv$
        for any environment $\genv$ that contains no dumped substitutions.
        By the invariant on $\state$, we have that:
        $$\pwamclosed{(\dstackthreep{\varthree},\genvtwo_1\cons\esub{\vartwo}{\pwammark}\cons\genv_2)}$$
        Moreover, from point \ref{p:pwam-invariants-dual1} we know $\pwamclosed{(\codetwo,\genv_2)}$.
        Both imply:
        $$\pwamclosed{(\dstackthreep{\varthree},\genvtwo_1\cons\esub{\vartwo}{\pwammark}\cons\esub{\var}{\codetwo}\cons\genv_2)}$$
        as required.

        \medskip
        For point \ref{p:pwam-invariants-dual3}, just note that the substitution
        $\esub{\var}{\codetwo}$ added to the environment is not dumped, and so duality holds because it holds for $\state$ by \ih.

  \item \caselight{Exponential}. We have $\state \tomache \statetwo$ with:
        $$\state = \pwamstate{\codeval}{\stempty}{\genv_1\cons\esub{\var}{\pwammark}\cons\genv_2}{(\var,\stack)\cons\fstack}$$
        $$\statetwo = \pwamstate{\rename\codeval}{\stack}{\genv_1\cons\esub{\var}{\codeval}\cons\genv_2}{\fstack}$$
        
        First note that since the environment and the dump are dual in $\state$,
        we know $\genv_2$ has no dumped substitutions.

        \medskip
        For proving point \ref{p:pwam-invariants-dual1}, by resorting to point \ref{p:pwam-invariants-dual1} on the
        state $\state$, for which the invariant already holds, we have that:
        \begin{equation} \label{eq:pwam-invariant-lam1a}
        \pwamclosed{(\codeval,\genv_2)}
        \end{equation}
        Moreover, by point \ref{p:pwam-invariants-dual2} on $\state$, specialized on the pair $(\var,\stack)$, we also know:
        \begin{equation} \label{eq:pwam-invariant-lam1b}
        \pwamclosed{(\dstackp{\var},\genv_1\envslice\cons\esub{\var}{\pwammark}\cons\genv_2)}
        \end{equation}
        We must check that:
        $$
        \pwamclosed{(\dstackp{\rename\codeval},\genv_1\envslice\cons\esub{\var}{\codeval}\cons\genv_2)}
        $$
        Any free variable in $\dstackp{\rename\val}$ is either free in $\stack$,
        in which case by \eqref{eq:pwam-invariant-lam1a} it must be bound by $\genv_1\envslice\cons\esub{\var}{\pwammark}\cons\genv_2$,
        or free in $\codeval$, in which case by \eqref{eq:pwam-invariant-lam1a} it must be bound by $\genv_2$.
        In both cases it is bound by $\genv_1\envslice\cons\esub{\var}{\codeval}\cons\genv_2$, as required.
        To conclude the proof of point \ref{p:pwam-invariants-dual1}, note that
        by combining \eqref{eq:pwam-invariant-lam1a} and \eqref{eq:pwam-invariant-lam1b}
        we get $\pwamclosed{\genv_1\envslice\cons\esub{\var}{\val}\cons\genv_2}$.

        \medskip
        For proving point \ref{p:pwam-invariants-dual2}, let $(\vartwo,\stacktwo)$ be a pair in $\fstack$.
        Using that $\var \neq \vartwo$, by the invariant on $\state$ we know:
        $$
        \pwamclosed{(\dstacktwop{\vartwo}, \genv_1\envslicevar{\vartwo}\cons\esub{\var}{\pwammark}\cons\genv_2)}
        $$
        and this implies:
        $$
        \pwamclosed{(\dstacktwop{\vartwo}, \genv_1\envslicevar{\vartwo}\cons\esub{\var}{\codeval}\cons\genv_2)}
        $$
        as wanted.

        \medskip
        Point \ref{p:pwam-invariants-dual3} is immediate, given that the environment
        and the dump are already dual in $\state$.

  \end{enumerate}

\end{proof}

%%%%%%%%%%%%%%%%%%%%%%
%%%%%%%%%%%%%%%%%%%%%%

\subsection{Proofs for Distillation is Complexity Preserving}
\label{ss:complexity-proofs}

\begin{proof}[\refth{value-name-glob-bilin}]
	\begin{enumerate}
		\item \emph{LAM}. As for the CEK, using the corresponding subterm invariant and the following measure:
    \begin{center}
$\card(\cekstate \codetwo \env \stack) := \begin{cases}
		\size{\codetwo}+\size{\codethree} & \textrm{if }\stack=\fnst{\codethree,\envtwo}\cons\stacktwo \\
		\size{\codetwo} & \textrm{otherwise}
	\end{cases}$
	\end{center}
	
			\item \emph{Split CEK}. As for the CEK, using the corresponding subterm invariant and the following measure:
  $$
  \card(\scekstate{\codetwo}{\env}{\stack}{\fstack}) := \begin{cases}
  \size{\codetwo} + \size{\codethree} & \text{if $\stack = (\codethree,\envtwo)\cons\stacktwo$} \\
  \size{\codetwo}                   & \text{otherwise}
  \end{cases}
  $$

	\end{enumerate}
\end{proof}

\end{document}